\documentclass[runningheads]{llncs}
\usepackage[T1]{fontenc}
\usepackage{graphicx}
\usepackage{amsmath}
\usepackage{tikz}
\usepackage[heavycircles]{stmaryrd}
\usepackage{mathtools}
\usepackage[inline]{enumitem}
\usepackage{listings}
\usepackage{mathpartir}
\usepackage[normalem]{ulem}
\usepackage{hyperref}

\usepackage[capitalize]{cleveref}
\usepackage{thmtools}
\usepackage{soul}
\usepackage{eucal}
\usepackage{ebproof}

\newcommand{\jcref}[1]{\cite{DBLP:journals/corr/abs-2506-20356}}

\makeatletter
\newcommand*{\inlineequation}[2][]{%
  \begingroup
    \refstepcounter{equation}%
    \ifx\\#1\\%
    \else
      \label{#1}%
    \fi
    \relpenalty=10000 %
    \binoppenalty=10000 %
    \ensuremath{%
      #2%
    }%
    ~\@eqnnum
  \endgroup
}
\makeatother

\usepackage{bbold}

\newcommand{\keyword}[1]{\mathsf{#1}}

\newcommand{\typeof}[1]{\keyword{typeof}({#1}) }
\newcommand{\fv}[1]{\keyword{fv}({#1}) }
\newcommand{\fpv}[1]{\keyword{fpv}({#1})}
\newcommand{\ftv}[1]{\keyword{ftv}({#1})}

\newcommand{\grmeq}{\; ::= \;\;}
\newcommand{\grmor}{\;\mid\;}

\newcommand*{\priority}[1]{\textcolor{teal}{#1}}
\newcommand{\oracons}{\ensuremath{\priority{\kappa}}}

\newcommand{\lowestpriority}[1]{\priority{|#1|}}

\newcommand{\nuLowestpriorityEnv}[2]{\priority{|\textcolor{black}{#1}|}_{#2}}

\newcommand{\glb}{\sqcup}
\newcommand{\lub}{\sqcap}
\newcommand*{\pr}[1][{}]{\priority{\rho_{#1}}}
\newcommand*{\prsigma}{\priority{\sigma}}

\newcommand{\prvar}{\priority{\iota}}
\newcommand{\prbot}{\priority{\bot}}
\newcommand{\prtop}{\priority{\top}}
\newcommand{\prvalue}{\priority{n}}
\newcommand{\prlub}[2]{\priority{{#1} \lub {#2}}}
\newcommand{\prglb}[2]{\priority{{#1} \glb {#2}}}
\newcommand{\prdisp}[2]{\priority{{#1}+{#2}}}
\newcommand{\intervalvar}{\priority{\varphi}}
\newcommand{\intervalcc}[2]{[\priority{#1}, \priority{#2}]}
\newcommand{\intervaloc}[2]{(\priority{#1}, \priority{#2}]}
\newcommand{\intervalco}[2]{[\priority{#1}, \priority{#2})}
\newcommand{\intervaloo}[2]{(\priority{#1}, \priority{#2}) }

\newcommand{\tunit}{\keyword{Unit}}
\newcommand{\tint}{\keyword{Int}}
\newcommand{\tarrow}[5]{{#1} \rightarrow_{#5}^{\priority{#3},\priority{#4}} {#2}}

\newcommand{\tprod}[2]{{#1}\times {#2}}

\newcommand{\tskip}[0]{\keyword{Skip}}
\newcommand{\tout}[2]{!^{\priority{#2}}{#1}}
\newcommand{\tin}[2]{?^{\priority{#2}}{#1}}
\newcommand{\tseq}[2]{{#1};{#2}}
\newcommand{\tclose}[1]{\keyword{Close}^{\priority{#1}}}
\newcommand{\twait}[1]{\keyword{Wait}^{\priority{#1}}}
\newcommand{\trec}[2]{\mu{#1}.{#2}}
\newcommand{\tvar}{\alpha}
\newcommand{\svar}{\beta}

\newcommand{\tpolyt}[4]{\forall {#1}^{\priority{#2}}. {#3}}
\newcommand{\tpolyp}[4]{{\fforall}^{\fallstyle{#1}} {\priority{#2}}^{\priority{#3}}. {#4}}
\newcommand{\tintchoice}[4]{\oplus^{\priority{#4}}\{{#1}\colon {#2}\}_{#3}}
\newcommand{\textchoice}[4]{\&^{\priority{#4}}\{{#1}\colon {#2}\}_{#3}}

\newcommand{\linmult}{\keyword{lin}}
\newcommand{\unmult}{\keyword{un}}

\newcommand{\emptyenv}{\epsilon}
\newcommand*{\envtvars}{\Gamma}
\newcommand{\envpoly}{\Delta}
\newcommand{\envpr}{\Theta}
\newcommand{\envmap}{\Psi}
\newcommand*{\dom}[1]{\keyword{dom}({#1})}
\newcommand*{\fst}[1]{\priority{\lceil} #1\priority{\rangle\!\rangle}}
\newcommand*{\tail}[1]{\priority{\rangle\!\rangle}#1\priority{\rfloor}}

\newcommand*{\unr}[1]{\keyword{unr}({#1})}

\newcommand{\dualof}[1]{\overline{#1}}
\newcommand{\subs}[3]{#1[\raisebox{.5ex}{\small$#2$}\! / \mbox{\small$#3$}]}
\newcommand{\enewk}{\keyword{new}}
\newcommand{\eunit}{\,()\,}
\newcommand{\eabs}[4]{\lambda_{#4} {#1} : {#2}. {#3}}
\newcommand{\etabs}[2]{\Lambda {#1}. {#2}}
\newcommand{\eprabs}[2]{\mathbb{\Lambda} {#1}. {#2}}
\newcommand{\eapp}[2]{{#1}{#2}}
\newcommand{\etapp}[2]{{#1}[{#2}]}
\newcommand{\eprapp}[2]{{#1}\{{#2}\}}
\newcommand{\epair}[2]{({#1},{#2})}
\newcommand{\elet}[3]{\keyword{let}\, {#1} = {#2}\, \keyword{in}\, {#3}}
\newcommand{\eletpair}[4]{\keyword{let}\, \epair{#1}{#2} = {#3}\, \keyword{in}\, {#4}}
\newcommand{\eseq}[2]{{#1}; {#2}}
\newcommand{\eselect}[1]{\keyword{select}\,{#1}}
\newcommand{\ematch}[4]{\keyword{match}\,{#1}\, \keyword{with}\, \{{#2}\rightarrow{#3}\}_{#4}}
\newcommand{\enewpoly}[3]{\enewk\,{\mathcal{#1}}\,{\priority{#2}}\,{\priority{#3}}}
\newcommand{\enew}[1]{\enewk\,{\mathcal{#1}}}
\newcommand{\efork}{\keyword{fork}}
\newcommand{\esend}{\keyword{send}}
\newcommand{\ereceive}{\keyword{receive}}
\newcommand{\efix}{\keyword{fix}}
\newcommand{\einst}[2]{\keyword{inst}^{\textcolor{darkgray}{#2}}\,{#1}}

\newcommand{\evar}{x}
\newcommand{\e}{e}
\newcommand{\econst}{c}
\newcommand{\enat}{n}
\newcommand{\eclose}{\keyword{close}}
\newcommand{\ewait}{\keyword{wait}}
\newcommand{\enext}{\keyword{next}}
\newcommand{\fallstyle}[1]{\textsc{\textcolor{gray}{#1}}}

\newcommand{\thread}{\phi}
\newcommand{\mainthread}{\bullet}
\newcommand{\childthread}{\circ}
\newcommand{\confthread}[2]{{#1}\langle{#2}\rangle}
\newcommand{\confpar}[2]{{#1} \,\|\, {#2}}
\newcommand{\confnu}[4]{(\nu {#1} {#2}^{\priority{#4}})\, {#3}}

\newcommand{\config}[1]{\mathcal{#1}}

\newcommand{\typingpoly}[9]{{#1} \mid {#2} \mid {#3} \mid {#4} \vdash {#5} : {#6} \,\,\textcolor{darkgray}{\diamond}\,\, \priority{{#7}} \,\,\textcolor{darkgray}{\diamond}\,\,\priority{{#8}}}

\newcommand{\wfpolyctx}[3]{{#1} \mid {#2} \vdash {#3}}
\newcommand{\istype}[3]{{#1} \mid {#2} \vdash {#3}}

\newcommand{\runtimetyping}[4]{{#1} \vdash_{\priority{#4}}^{#3} {#2}}
\newcommand{\ctxsplit}[5]{{#1}\mid{#2} \vdash {#3} = {#4} \circ {#5}}

\newcommand*{\structcong}{\equiv\,}

\newcommand{\red}[2]{{#1}\rightarrow{#2}}
\newcommand{\evalhole}[1]{[{#1}]}
\newcommand{\evalctx}{E}
\newcommand{\threadctx}{\mathcal{F}}
\newcommand{\configctx}{\mathcal{G}}
\newcommand{\evalapp}[2]{{#1}\evalhole{#2}}

\newcommand{\declrel}[1]{\emph{#1}\hfill{ }}

\newcommand\Small{\small}

\definecolor{darkviolet}{rgb}{0.5,0,0.4}
\definecolor{darkgreen}{rgb}{0,0.4,0.2}
\definecolor{darkblue}{rgb}{0.1,0.1,0.9}
\definecolor{darkgrey}{rgb}{0.5,0.5,0.5}
\definecolor{lightblue}{rgb}{0.4,0.4,1}

\lstdefinestyle{eclipse}{
numbers=left,
  breaklines=true,
  basicstyle=\sffamily\Small,
  emphstyle=\color{red}\bfseries,
  keywordstyle=\color{darkviolet}\bfseries,
  commentstyle=\color{gray},
  stringstyle=\color{darkblue},
  numberstyle=\color{darkgrey},
  xleftmargin=5.5ex, 
  emphstyle=\color{red},
  showstringspaces=false,
  moredelim=**[is][\color{teal}]{~}{~},
}

\makeatletter
\newcommand*{\fforall}{%
  {\mathpalette\fforallAux{}}%
}
\newcommand*{\fforallAuxx}[1]{%
  \sbox0{$\m@th#1\forall$}%
  \sbox2{%
    \rlap{%
      \raisebox{\depth}{$\m@th#1\backslash$}%
    }%
    \kern\ht0 %
  }%
  \sbox2{\resizebox{\ht2}{\height}{\copy2}}%
  \sbox2{\resizebox{!}{\ht0}{\copy2}}%
  \wd2=0pt %
  \copy2
  \forall
}
\newsavebox\forallBox
\newdimen\forallLineWidth
\newdimen\forallSep
\newcommand*{\fforallAux}[1]{%
  \sbox\forallBox{$\m@th#1\forall$}%
  \setlength{\forallLineWidth}{.06\wd\forallBox}%
  \setlength{\forallSep}{.09\wd\forallBox}%
  \tikz[
    inner sep=0pt,
    line cap=round,
    line width=\forallLineWidth,
  ]
  \draw
    (0,0) node (A) {\copy\forallBox}
    (A.south) ++(-\forallSep-\forallLineWidth,.4\forallLineWidth)
    coordinate (A1)
    (A.north west) ++(-\forallSep,-\forallLineWidth)
    coordinate (A2)
    (A1) -- (A2)
  ;%
}
\makeatother 
\begin{document}
\title{Deadlock-free Context-free Session Types}
%
%
 \author{Andreia Mordido\inst{1}\orcidID{0000-0002-1547-0692} 
 \and
  Jorge A. P\'{e}rez\inst{2}
  \orcidID{0000-0002-1452-6180} 
  }
 \authorrunning{Mordido and P\'{e}rez}
 \institute{LASIGE, Faculdade de Ci\^{e}ncias, Universidade de Lisboa, Portugal \and
University of Groningen, The Netherlands
}
\maketitle              
\begin{abstract}
  We tackle the problem of statically ensuring that message-pas\-sing programs never run into \emph{deadlocks}. 
  We focus on concurrent functional programs 
  governed by \emph{context-free session types}, which can express rich tree-like structures not expressible in regular session types.
  We propose a new type system based on context-free session types: it enforces both protocol conformance and deadlock freedom, also for programs implementing cyclic  communication topologies with recursion and polymorphism.
  We show how the \emph{priority-based} approach to 
  deadlock freedom {can be extended to this expressive setting.}
  We prove that well-typed concurrent programs 
  respect their protocols and never deadlock.
\end{abstract}

\section{Introduction}
A long-standing issue 
in concurrent programming
is ensuring the absence of \emph{deadlocks}: states in which processes are forever blocked awaiting a mes\-sage. 
Motivated by this challenge, this paper develops new compositional verification techniques that ensure deadlock freedom for  processes that exchange messages by following some \emph{protocols}.
We focus on \emph{session types}, which specify reciprocal protocols on heterogeneously typed 
channels~\cite{DBLP:conf/concur/Honda93,DBLP:conf/esop/HondaVK98}. 
As key novelty, we consider \emph{context-free} session types (CFSTs)~\cite{DBLP:conf/icfp/ThiemannV16}, in which 
sequential composition, non-regular recursion, and polymorphism can jointly specify tree-like structures not expressible in regular session types. By targeting CFSTs, our work broadens the range of concurrent processes that can be  verified to be deadlock-free.

The expressiveness of CFSTs leads to challenges in the study of notions such as type  equivalence~\cite{DBLP:conf/tacas/AlmeidaMV20} and subtyping~\cite{DBLP:conf/concur/SilvaMV23}. Addressing deadlock freedom is not an exception, for two reasons. 
\emph{First},  enforcing deadlock freedom entails performing \emph{non-local analyses} on the behavior of processes. The goal is to identify circular dependencies between processes that cannot be detected by analyzing each process in isolation. Existing disciplines based on CFSTs prioritize programming flexibility; as they rely  on purely local analyses, they accept deadlocked processes as typable.
\emph{Second}, there is no ``obviously small'' formulation of CFSTs for studying deadlock freedom. For regular session types, useful deadlock analyses can be done in a setting without recursion,  whereas in the case of CFSTs one must include recursion \emph{and} polymorphism---indeed, they are ``entangled''~\cite{DBLP:conf/icfp/ThiemannV16}: limiting recursion would make 
our system degenerate into regular session types; 
removing polymorphism would prevent us from having recursion. 
Defined on top of a concurrent functional language, our
CFSTs thus combine three distinctive features: 
sequential composition, non-regular recursion, and polymorphism. 

We handle this formulation of CFSTs by extending the original proposal by Thiemann and Vasconcelos~\cite{DBLP:conf/icfp/ThiemannV16}. To carry out the non-local analyses required, we follow the \emph{priority-based} 
approach pioneered by Ko\-ba\-yashi~\cite{DBLP:conf/unu/Kobayashi02,DBLP:conf/concur/Kobayashi06}, later streamlined by Padovani~\cite{padovani_linear_pi}. The key idea is to endow each communication action  with a \emph{priority}, an integer that denotes its urgency: immediately enabled actions have a higher priority than those that appear later on in the process. This idea allows to detect processes that are stuck due to insidious circular dependencies between threads, as their associated actions cannot be given a valid priority. The approach supports processes organized in \emph{cyclic network topologies}, which abound in practical  distributed scenarios; it also extends to processes with recursion and polymorphism, in which priority management requires care. 

Originally developed for variants of the $\pi$-cal\-cu\-lus with \emph{linear types}~\cite{DBLP:journals/toplas/KobayashiPT99},   here we build upon the type systems by Padovani and Novara~\cite{padovani:hal-00954364,DBLP:conf/forte/PadovaniN15}, who adopted the priority-based approach in~\cite{padovani_linear_pi} to higher-order concurrent programs.  
Following their approach, but casting it in the more expressive setting of CFSTs, we  develop a type system and associated meta-theoretical results of \emph{type preservation} (protocol conformance, \Cref{lem:subject-red-processes}) and \emph{deadlock freedom} (\Cref{thm:config-deadlock-free}).

\paragraph{Contributions and Outline.} 
This paper's contributions are two-fold. 
First,  in the context of a call-by-value concurrent language, we endow  CFSTs with a priority-based mechanism to guarantee deadlock freedom.
Second, we devise novel mechanisms for supporting polymorphic types with priorities, namely  constructs for channel creation and polymorphic recursion at the level of priorities.

\Cref{sec:deadlocks} motivates our work via examples
written in FreeST~\cite{freest}, a functional programming language with CFSTs.
\Cref{sec:types} introduces the type language and notions of type duality and priorities. 
\Cref{sec:semantics} introduces expressions, configurations, and the typing rules. 
\Cref{sec:safety} presents the operational semantics and proves type preservation and deadlock freedom.  
\Cref{sec:conclusion} discusses related work and  
concludes the paper. 
\emph{The appendices contain omitted definitions and proofs, extra examples, and extended related work discussions}.

\section{Deadlock Freedom in Programs with CFSTs}
\label{sec:deadlocks}

To motivate our developments, we present two examples written in FreeST~\cite{freest}, a functional programming language with CFSTs:
a classic deadlock scenario avoided by priorities and
a recursive \emph{stream}, that motivates \emph{polymorphic recursion} at the level of priorities.
An example of non-regular recursion in binary trees is in~\cref{ss:trees}. 
Details on cyclic topologies are deferred to~\cref{ap:cyclic-scheduler}.

\paragraph{A Classic Deadlock Scenario.}
Consider a communicating 
program involving two processes that share two channels. They 
aim to send in one channel the value received in the other, but do so in an order that prevents synchronization.
Each process holds two linear \emph{endpoints}: one 
obeys the type 
\lstinline|T = ?Int; Wait|, meaning that it
first \emph{receives} an integer and then \emph{waits} to be closed; 
the other follows  the type \emph{dual} to \lstinline|T| (denoted \lstinline|!Int; Close|), which \emph{sends} an integer and then \emph{closes} the channel.
In FreeST, each process can be implemented  as follows:
\begin{lstlisting}
f : T -> dualof T 1-> ()
f x y = let (n,x) = receive x in
        let y = send n y in
        wait x; close y
\end{lstlisting}
where  `\lstinline|1->|' in the type of \lstinline|f| represents a \emph{linear} function, to be used \emph{once}.
When a FreeST program launches two parallel processes, 
each running an instance of function \lstinline|f|, 
the program type-checks, but deadlocks: both processes block indefinitely, each waiting to receive a value that the other never sends.

To avoid deadlocks, one solution is to prescribe 
the intended order of operations in the program body. 
We annotate the types with  {priorities} associated to the corresponding 
operation, following  Padovani and Novara~\cite{padovani:hal-00954364,DBLP:conf/forte/PadovaniN15}.
Priorities are integers that indicate an action's \emph{urgency} in the protocol.
Throughout the paper we use 
colour \textcolor{teal}{teal} to 
highlight priorities in our
type system.

Revisiting the classic deadlock scenario, 
the types that govern the two channels  should now be typed
differently, priority-wise: 
\begin{lstlisting}
type R = ?~p~Int; Wait ~(p+2)~
type S = ?~q~Int; Wait ~(q+2)~
\end{lstlisting}
Type \lstinline|R| extends \lstinline|T| with priorities: it receives an \lstinline|Int| at priority \lstinline|~p~|
and then waits at priority \lstinline|~p+2~|; 
type \lstinline|S| extends  \lstinline|T| similarly, under priority 
\lstinline|~q~|.
To showcase the use of priorities, we could consider two variants of function 
\lstinline|f| (the meaning of priority bounds on arrows, shown within 
\lstinline|~[]~| after the arrow, is 
formalized in~\cref{sec:types}):

\lstset{firstnumber=3}
\begin{lstlisting}
f : R ->~[top, bot]~ dualof S 1->~[p, q+2]~ ()
f x y =
  let (n,x) = receive x in -- priority: p
  let y = send n y in      -- priority: q
  ...
g : S ->~[top, bot]~ dualof R 1->~[q, p+2]~ ()
g y x = ... 
\end{lstlisting}

By following 
the priorities of the arguments' types,
we detect a mismatch:
function \lstinline|f| 
requires \lstinline|~p~| < \lstinline|~q~| (Lines 5 and 6)
whereas function \lstinline|g| requires \lstinline|~q~| < \lstinline|~p~|.
The type checker flags the mismatch when \lstinline|f| and 
\lstinline|g| run in parallel: no priorities \lstinline|~p~| and 
\lstinline|~q~| satisfy both constraints, so this deadlocked state is 
rejected statically. The fix is to swap the operations
receive/send and wait/close in one of the functions.

\paragraph{Recursion and Polymorphism over Priorities.}
\label{subsec:deadlocks-simple-rec}

To express repetitive protocols, session types leverage recursion and are 
usually equi-recursive. Recursion introduces 
challenges~\cite{padovani:hal-00954364} that are familiar to issues found in the theory of CFSTs~\cite{DBLP:conf/icfp/ThiemannV16}.

Consider a protocol that sends a
stream of unit values, defined by  the type 
\lstinline|Stream|  as
$\trec{\svar}{\tseq{!\tunit }{\svar}}$. We could  write a recursive program for 
two interacting \lstinline|Stream|s.
However, with equi-recursive types, iterating the protocol forces 
reusing the same priorities for later actions, 
which makes repeated sends/receives inconsistent; this motivates us to define \emph{polymorphic recursion} at the level of priorities.

We
propose \emph{priority abstractions} to allow the type to evolve 
while avoiding inconsistent states.
We use a quantifier
$\fforall\,\prvar^{\intervalvar}$ to abstract over priorities 
and bind a \emph{priority
variable} $\prvar$ to a \emph{priority interval} $\intervalvar$, as 
proposed in~\cite{padovani:hal-00954364}.
The type \lstinline|Stream| can be defined by:
$\trec{\svar}{(\tpolyp{}{\prvar}{\intervaloo{\bot}{\top}}{\tseq{\tout{\tunit}{\prvar}}{\svar}})}.$
A channel of this type first instantiates the 
priority variable $\priority{\prvar}$ 
with an integer in interval 
$\intervaloo{\prbot}{\prtop}$ that materializes the order/urgency of 
the send operation, sends a unit value and finally recurs.

The instantiation of priorities \emph{on channels} cannot be arbitrary.
When creating a channel (using \lstinline|new|), we assign a \emph{priority sequence}    {predefined} for the channel's endpoints. 
Such a sequence  intuitively acts as a `stack' of priorities,  common to both endpoints.
With this mechanism for channel creation, we could define a \lstinline|client| function 
as follows (noting that \lstinline|()| represents $\tunit$ in FreeST):

\lstset{firstnumber=1}
\begin{lstlisting}[label=lst:stream, caption={Example of polymorphism over priorities on session and functional types.},captionpos=b]
type Stream = ~forallp i belongsTo (bot,top) =>~ !~i~() ; Stream
  
client: ~forallp p belongsTo (bot,top)=> ~Stream ->~[top,bot]~ dualof Stream 
                                   1->~[p,top]~ ()
client c1 c2 =                                                                
  let c1 = send () (inst c1) in      -- c1 : Stream  
  let (_,c2) = receive (inst c2) in  -- c2 : Stream
  client~{next c1}~ c1 c2
\end{lstlisting}

We overload \lstinline|forallp| for priority-quantified sessions (Line 1) 
and functions (Line~3); the two are instantiated differently 
(\lstinline|inst| \emph{vs.} \lstinline|~{...}~|). In Line 6, \lstinline|inst c1| uses 
the next priority for \lstinline|c1|, sends \lstinline|()| at priority $\pr$, 
and pops $\pr$ from \lstinline|c1|'s stack (returning \lstinline|c1| with type 
\lstinline|Stream|); in Line 8, the recursive call to \lstinline|client| explicitly instantiates 
\lstinline|~p~| with the \lstinline|next| priority assigned to \lstinline|c1|.


The consistency of program operations 
is verified at compile time, according to the 
priorities fixed upon channel creation. 
The function \lstinline|main| below creates two channels of type 
\lstinline|Stream| that are used by two threads, one executing \lstinline|client|
and another executing \lstinline|newclient| (similar to \lstinline|client| with \lstinline|send| and \lstinline|receive| swapped):

\lstset{firstnumber=9}
\begin{lstlisting}[label=lst:main,caption={Creation of polymorphic channels with  priority sequences.},captionpos=b]
main : ()
main =
  let (w1, r1)  = new Stream ~1 2~ () in
  let (w2, r2)  = new Stream ~2 2~ () in
  fork (\_:() 1-> client~{next w1}~ w1 r2);
  newclient~{next w2}~ w2 r1
\end{lstlisting}
For priority-polymorphic sessions, \lstinline|new| allocates dual endpoints 
and fixes their priority sequence, ensuring that both endpoints  are always governed by the same priorities.
This sequence is determined by two user-provided numbers, 
an \emph{initial value} and an \emph{increment}. 
In Lines 11-12,  endpoints 
for \lstinline|r|eading and \lstinline|w|riting are created: 
\lstinline|r1| and 
\lstinline|w1| get
the priority sequence [1, 3, 5, \ldots]
from \lstinline|~1~| (initial value) and \lstinline|~2~| (increment), while \lstinline|r2| and \lstinline|w2| get
[2, 4, 6, \ldots] from \lstinline|~2 2~|; recursive calls instantiate their 
priority variables with the channels' \lstinline|next| priorities.


\paragraph{Non-Regular Recursion.}
The classic example that capitalises on all the features of CFSTs is exchanging binary 
trees~\cite{DBLP:conf/tacas/AlmeidaMV20,DBLP:conf/icfp/ThiemannV16}.
Consider the type \lstinline|TreeChannel|, defined as 
$\trec{\svar}{(\tpolyp{S}{\prvar}{\intervaloo{\bot}{\top}}{\&^{\priority{\prvar}}\{\keyword{LeafC}: \tskip, \keyword{NodeC}: \tseq{\tin{\tint}{\prvar+1}}{\tseq{\svar}{\svar}}\}{}}{})}$:
it first
instantiates the priority variable $\priority{\prvar}$
and then offers two branches:  one with label 
$\keyword{LeafC}$, which  does nothing; and another 
with label $\keyword{NodeC}$, which receives an integer 
under priority $\priority{\prvar+1}$ and then receives the 
left and right subtrees.  
Consider a function \lstinline|receiveTree| in our approach (assuming \lstinline!data Tree = Leaf | Node Int Tree Tree!): 
\lstset{firstnumber=1}
\begin{lstlisting}
receiveTree : ~forallp p belongsTo (bot, top) =>~ forall a :: ~p~ =>
             TreeChannel; a ->~[top,top]~ (Tree, a)
receiveTree c =
  match (inst c) with {
    LeafC c ->
      (Leaf, c),
    NodeC c ->
      let (x, c) = receive c in
      let (left, c) = receiveTree~{next c}~ @(TreeChannel;a) c in
      let (right, c) = receiveTree~{next c}~ @a c in
      (Node x left right, c) }
\end{lstlisting}
\lstset{firstnumber=1}

This function inputs a channel of type `\lstinline|TreeChannel; a|', where 
\lstinline|a| is a \emph{type variable} that represents a type
with priority \lstinline|~p belongsTo (bot, top)~|.
This variable is instantiated in a function call 
through the \lstinline|@| operator (Lines 9-10).
Since channel \lstinline|c| is polymorphic,  \lstinline|receiveTree| first instantiates its priority
variable (\lstinline|inst c| in Line~4).
The \lstinline|match| expression handles the 
choices on  channel \lstinline|c|.
The recursive calls apply polymorphic recursion 
on priorities \emph{and} types: they
instantiate the priority variable \lstinline|~p~| with the 
\lstinline|next| priority assigned to the operations in channel \lstinline|c|;
note that  \lstinline|next ~c~| has different values in Lines~9 and 10.
The function
returns a pair with a reconstruction of the node
and  the continuation of the channel.

\section{Context-Free Session Types, with Priorities}
\label{sec:types}

We define types in two syntactic categories:
functional types and session types.
We rely on some base sets: 
$\ell\in L$ stands for labels, $\tvar$ stands for
functional type variables, $\svar$ stands for session type
variables, and $m$ stands for multiplicities (relevant for arrow types).
We use $\gamma$ as a meta-variable for  type variables (functional or session).
As we have seen, $\priority{\pr}$, $\prvar$, and $\priority{\intervalvar}$ stand for priorities, priority variables, and priority intervals, respectively.
We use $\prbot$ and $\prtop$ to stand for the lowest and highest 
priority, respectively. 
Also, we write $\priority{\bar{\pi}}$ to denote a priority sequence, as motivated in~\cref{subsec:deadlocks-simple-rec}.
The full grammar of types is in~\cref{fig:types}.

\begin{figure}[!t]
  \begin{align*}
    T,& \,U \grmeq & \text{(Functional types)}  &\quad & S& \grmeq & \text{(Session types)}\\
     &\tunit &\text{unit type} &&\grmor&\tskip & \text{type skip}\\
    &\grmor \tarrow{T}{U}{\pr_1}{\pr_2}{m}  \hspace*{-4mm}& \text{function type} &&\grmor& \tout{T}{\pr} & \text{send value of type $T$}\\
    &\grmor \tprod{T}{U} & \text{product type} &&\grmor& \tin{T}{\pr} & \text{receive value of type $T$}\\
    &\grmor \alpha &\text{(functional) type variable} &&\grmor& \tintchoice{\ell}{S_\ell}{\ell\in L}{\pr} \hspace*{-8mm}& \text{internal choice}\\
    &\grmor \trec{\tvar}{T} &\text{rec. functional type} & &\grmor& \textchoice{\ell}{S_\ell}{\ell\in L}{\pr} \hspace*{-8mm}&\text{external choice}
    \\ 
    &\grmor \tpolyt{\alpha}{\pr}{T}{\prsigma} & \text{poly. functional type}&&\grmor& \tseq{S_1}{S_2} & \text{sequential composition} \\ 
    &\grmor \tpolyt{\svar}{\pr}{S}{\prsigma}  & \text{poly. session type}& &\grmor& \tclose{\pr} & \text{ready to be closed}\\
    &\grmor \tpolyp{F}{\iota}{\intervalvar}{T} & \text{priority abstraction} &&\grmor& \twait{\pr} & \text{wait to be closed}\\
    &\grmor S & \text{session type} & &\grmor & \svar & \text{(session) type variable}\\
    &&&&\grmor& \trec{\svar}{S} & \text{recursive session type}\\
    &&&&\grmor & \tpolyp{S}{\prvar}{\intervalvar}{S} & \text{priority abstraction}
 \end{align*}\vspace*{-5mm}
 \begin{align*}
    m \grmeq & \linmult \grmor \unmult &\tag{Multiplicities}\\
  \pr \grmeq & \prbot \grmor \prvalue \grmor \prtop \grmor \prvar
    \grmor \prdisp{\priority{\pr}}{\priority{n}}
    \tag{Priorities}\\
    \intervalvar \grmeq & \intervalcc{\priority{\pr_1}}{\priority{\pr_2}} \grmor \intervaloc{\priority{\pr_1}}{\priority{\pr_2}}
    \grmor \intervalco{\priority{\pr_1}}{\priority{\pr_2}} \grmor \intervaloo{\priority{\pr_1}}{\priority{\pr_2}}
    \tag{Intervals}
%
\\[1mm]
    \envpoly \grmeq& \emptyenv \grmor \envpoly, \tvar:: \pr \grmor \envpoly, \svar :: \pr 
    \qquad \tag{Polymorphic variables} \\
    \envpr \grmeq& \emptyenv \grmor \envpr, \priority{\prvar\in \intervalvar}
    \qquad \tag{Priority variables}
    \vspace{-3mm}
 \end{align*}
 \caption{Types, priorities and priority variables, 
  polymorphic and priority contexts.}
 \label{fig:types}
\end{figure}

Having two syntactic categories avoids the need for a kinding system~\cite{DBLP:journals/iandc/AlmeidaMTV22,DBLP:conf/icfp/ThiemannV16}, but
entails some duplication: polymorphic types can 
either have a {functional} body, in which case the 
quantifier binds a 
\emph{functional} type variable, or the body can be a session, 
in which case the quantifier binds a \emph{session} type 
variable; the same for recursive types. This means that in addition 
to $\tvar\neq \svar$, we also have two different (but overloaded)
$\forall$ and $\mu$ operators. 
Recursive 
types keep their body's original ``nature'', whereas polymorphic
types are always functional. (We leave the details of this 
choice to previous work on context-free session 
types~\cite{DBLP:conf/tacas/AlmeidaMV20}.) 

As discussed in \Cref{sec:deadlocks}, the distinguishing feature of our work is the extension of CFSTs with priority 
annotations to enforce deadlock freedom~\cite{DBLP:journals/lmcs/KokkeD23,padovani:hal-00954364,DBLP:conf/forte/PadovaniN15}.
The function type 
$\tarrow{T}{U}{\pr_1}{\pr_2}{m}$ features a 
\emph{lower priority bound} $\priority{\pr_1}$ which indicates that the closure of 
the function contains channels with priority $\priority{\pr_1}$ or 
greater, and an \emph{upper priority bound} $\priority{\pr_2}$, which stands for 
the maximum priority that is pending in the function and 
will be triggered once the function is applied (see the
explanation by 
Padovani and Novara~\cite{padovani:hal-00954364} for further details). 
The  $m$ denotes the function's {multiplicity}:
$m = \linmult$ means that the function can be used once, whereas
$m = \unmult$ means that it can be used zero or 
more times. 

In polymorphic types ($\tpolyt{\alpha}{\pr}{T}{\prsigma}$, 
$\tpolyt{\svar}{\pr}{S}{\prsigma}$), the type variable is given a 
priority~$\pr$. Priorities are  
defined via priority variables $\prvar$, that are
abstracted through a $\fforall$ binder that assigns a 
priority interval $\priority{\intervalvar}$  to $\prvar$. 
Since priority applications are handled differently, we distinguish between 
binders for functional types $\fforall^\fallstyle{F}$ 
and  for session types $\fforall^\fallstyle{S}$;  the superscript is omitted when clear from context.
To define channel creation, we write $\EuScript{S}^{\fforall}$ to denote session types that contain 
a priority abstraction $\fforall^\fallstyle{S}$;
session types without priority binders are denoted by $\EuScript{S}$.

Following~\cite{DBLP:journals/lmcs/KokkeD23,DBLP:conf/fossacs/Padovani14,padovani:hal-00954364}, 
our synchronizations (send/receive, internal/external choice, close/wait) operate under a priority $\pr$. Type $\tskip$ and sequential composition 
do not explicitly represent an operation, so do not have priorities.

\cref{fig:types} defines priorities $\pr$ and 
priority intervals $\intervalvar$. 
Priorities can be $\prtop$ (which are given to  \emph{unrestricted types}, cf.~\cite{vasconcelos2012fundamentals}), integers $\priority{n}$, 
priority variables $\prvar$, and 
a \emph{displacement} $\prdisp{\pr}{n}$, which
represents the addition of integers extended with the
  rules $\prdisp{\prbot}{n} = \prbot$ and 
$\prdisp{\prtop}{n} = \prtop$. Priority intervals can be closed or open to account for the inclusion of 
the limits.

We denote the set of \emph{free 
type variables} in a type $T$ as $\ftv{T}$ and the set of \emph{free priority variables} in a priority $\pr$
(resp., interval $\intervalvar$) as $\fpv{\pr}$ (resp., $\fpv{\intervalvar}$).
We adopt the usual notions of substitution for both type and priority substitutions:
we denote by $\subs{T}{U}{\gamma}$ the substitution
of a type $U$ for a type variable $\gamma$ in $T$ and by $\subs{T}{\pr}{\prvar}$
the substitution of a priority $\priority{\pr}$ for a priority variable $\prvar$ in $T$.

Not all types are well formed. For example, 
we do not wish to define the type 
`$\tpolyt{\tvar}{\prvar}{\tin{\tvar}{\pr}}{}$' if $\prvar$ and $\pr$
are free variables.
To define type formation we use typing contexts for
polymorphic and priority variables, denoted  $\envpoly$ and $\envpr$, respectively, are 
presented in~\cref{fig:types}. 
We record  polymorphic variables
with the corresponding priority assignment in
$\envpoly$ and priority variables with their respective 
range in $\envpr$. 

\begin{figure}[!t]
  \begin{gather*}
    \infer[\textsf{F-TAbs}]
        {\istype{\envpoly, \gamma :: \pr}{\envpr}{T}
        \quad 
        \fpv{\pr}\in \dom{\envpr}}
        {\istype{\envpoly}{\envpr}{\tpolyt{\gamma}{\pr}{T}{\prsigma}}}
      \quad
      \infer[\textsf{F-PAbsF}]
            {\istype{\envpoly}{\envpr, \priority{\prvar\in \intervalvar}}{T}
            \quad 
            \fpv{\intervalvar}\in \dom{\envpr}}
            {\istype{\envpoly}{\envpr}{\tpolyp{F}{\prvar}{\intervalvar}{T}}}
\\ 
      \infer[\textsf{F-PAbsS}]
            {\istype{\envpoly}{\priority{\prvar\in \intervalvar}}{S}
            \quad 
            \fpv{\intervalvar}\in \dom{\envpr}}
            {\istype{\envpoly}{\emptyenv}{\tpolyp{S}{\prvar}{\intervalvar}{S}}}
  \end{gather*}
  \caption{Type formation.}
  \label{fig:type-formation}
      \vspace{-3mm}
\end{figure}

The judgement  $\istype{\envpoly}{\envpr}{T}$ 
denotes that type $T$ is well-formed under the polymorphic context $\envpoly$
and the priority context $\envpr$.
The standard CFST type-formation rules are as in~\cite{DBLP:conf/icfp/ThiemannV16}; the complete set of rules appears 
in~\cref{ap:types}. \Cref{fig:type-formation} shows only the 
new or modified type formation rules introduced in this work.
%

Rule \textsf{F-TAbs} states that a polymorphic type $\tpolyt{\gamma}{\pr}{T}{}$ 
can only assign a type variable $\gamma$
to a priority $\pr$ if all the free priority variables 
in $\pr$ are defined in $\envpr$; furthermore, $T$ should be well-formed
under the extended environment $\envpoly, \gamma::\pr$.
Similarly, rule \textsf{F-PAbsF} states that a priority variable $\prvar$
can only be assigned to $\intervalvar$ through the $\fforall^{\fallstyle{F}}$ binder
if all the free variables occurring in $\intervalvar$ are defined 
in $\envpr$; furthermore, the body of the type should be well-formed 
under the priority environment $\envpr, \prvar\in \intervalvar$. 
Rule \textsf{F-PAbsS}
imposes the same conditions but requires the priority environment to be empty, thus enforcing that the
$\fforall^{\fallstyle{S}}$ binder occurs at most once in a well-formed type; this 
restriction is meant to simplify the representation of 
priority sequences (cf. the definition of priority maps in~\Cref{fig:terms}).


Priorities govern the order of  {operations}. This is apparent 
in polymorphic types, which can be instantiated with different priority 
values. For this reason, priorities are closely related to typing environments
and a natural \emph{priority ordering} will be imposed in the typing rules
(cf.~\cref{sec:semantics}), rather than in type formation rules.
The typing rules will allow to ensure that there are 
no viable programs for types such as  $\tout{(\tout{\tint}{1})}{2}$, 
where the operation that follows in the payload is \emph{more urgent}, 
i.e., has lower priority, than sending the message itself.


Duality is paramount to establish  {compatible} communication 
between two endpoints of a channel. This notion is defined only 
over session types; the definition is 
standard~\cite{DBLP:conf/tacas/AlmeidaMV20,DBLP:conf/icfp/ThiemannV16}
 and is presented in~\cref{ap:types}.
Crucially, duality preserves priorities (e.g., $\dualof{\tout{T}{\pr}} = \tin{T}{\pr}$).
The dual of a priority abstraction 
($\dualof{\tpolyp{}{\svar}{\intervalvar}{S}}$)
is the priority abstraction 
of the dual session type, $\tpolyp{}{\svar}{\intervalvar}{\dualof{S}}$.


\section{Statics: Configurations and Typing Rules}
\label{sec:semantics}

\begin{figure*}[!t]
    $\enat  \grmeq  1 \grmor 2 \grmor\ldots$ \hfill{(Natural numbers)}\smallskip\\
    $\prsigma  \grmeq  \pr \grmor \enext^k x$ \hfill{(Priority values)}\smallskip\\
    $ \econst  \grmeq \efork \grmor \esend \grmor \ereceive \grmor \eclose \grmor \ewait \grmor \eunit \grmor \eselect{k} \grmor \efix$ \hfill {(Constants)}\smallskip\\
    $ v  \grmeq x \grmor \econst \grmor \eabs{x}{T}{e}{m} \grmor \epair{v}{v}
     \grmor \etabs{\tvar}{v} \grmor \eprabs{\prvar}{v} \grmor
     \eprapp{\eclose}{\prsigma} 
     \grmor \eprapp{\ewait}{\prsigma} 
     \grmor \eprapp{\efork}{\priority{\prsigma_1}} $\smallskip\\
     \hspace*{2mm} $ \grmor \eprapp{\eprapp{\efork}{\priority{\priority{\prsigma_1}}}}{\priority{\prsigma_2}}
        \grmor \eprapp{\esend}{\priority{\prsigma_1}}
      \grmor \etapp{\eprapp{\esend}{\priority{\prsigma_1}}}{T}
      \grmor \eprapp{\etapp{\eprapp{\esend}{\priority{\prsigma_1}}}{T}}{\priority{\prsigma_2}}$\smallskip\\
      \hspace*{2mm} $ \grmor \eprapp{\eprapp{\etapp{\eprapp{\esend}{\priority{\prsigma_1}}}{T}}{\priority{\prsigma_2}}}{\priority{\prsigma_3}}
     \grmor \etapp{\eprapp{\eprapp{\etapp{\eprapp{\esend}{\priority{\prsigma_1}}}{T}}{\priority{\prsigma_2}}}{\priority{\prsigma_3}}}{S}\grmor 
     \eapp{\etapp{\eprapp{\eprapp{\etapp{\eprapp{\esend}{\priority{\prsigma_1}}}{T}}{\priority{\prsigma_2}}}{\priority{\prsigma_3}}}{S}}{v}$\smallskip\\
     \hspace*{2mm} $\grmor \eprapp{\ereceive}{\priority{\prsigma_1}} 
     \grmor \etapp{\eprapp{\ereceive}{\priority{\prsigma_1}}}{T}
     \grmor \eprapp{\etapp{\eprapp{\ereceive}{\priority{\prsigma_1}}}{T}}{\priority{\prsigma_2}}
     \grmor \eprapp{\eprapp{\etapp{\eprapp{\ereceive}{\priority{\prsigma_1}}}{T}}{\priority{\prsigma_2}}}{\priority{\prsigma_3}}$\smallskip\\
     \hspace*{2mm} $\grmor \etapp{\eprapp{\eprapp{\etapp{\eprapp{\ereceive}{\priority{\prsigma_1}}}{T}}{\priority{\prsigma_2}}}{\priority{\prsigma_3}}}{S}
     \grmor \fbox{$\eprapp{x}{\prsigma}$}
       $
     \hfill{(Values)}\smallskip\\
     $e  \grmeq v \grmor \eapp{e}{e} \grmor \epair{e}{e} \grmor \elet{x}{e}{e}
     \grmor \eletpair{x_1}{x_2}{e}{e}
     \grmor \eseq{e}{e} 
     \grmor \etapp{e}{T} \grmor \eprapp{e}{\prsigma}$\smallskip\\
     \hspace*{2mm} $\grmor \ematch{e}{\ell}{e_\ell}{\ell\in L} 
     \grmor \enew{S} \grmor \enewpoly{\EuScript{S}^\fforall}{n_1}{n_2}  
     \grmor \einst{x}{\envmap} $ \hfill
     {(Expressions)}\smallskip\\
     $\config{C}  \grmeq \confthread \phi \e \grmor \confpar{\config{C}}{\config{C}} 
     \grmor \confnu{x}{y}{\config{C}}{\bar{\pi}} $ \hfill{(Configurations)}\smallskip\\
     $\phi \grmeq  \childthread \grmor \mainthread $ \hfill{(Thread flags)}\smallskip\\
    $\envtvars \grmeq  \emptyenv \grmor \envtvars, x: T $ \hfill{(Typing contexts)}\smallskip\\
     $\envmap \grmeq  \emptyenv \grmor \envmap, x : T \mapsto \priority{\bar{\pi}} $ \hfill{(Priority maps)}
  \caption{The syntax of values, expressions, configurations, priorities and typing contexts.}
  \label{fig:terms}
    \vspace{-2mm}
\end{figure*}

Here we present our concurrent functional language with sessions, which extends earlier works~\cite{DBLP:conf/tacas/AlmeidaMV20,DBLP:conf/icfp/ThiemannV16} with applications and abstractions over priority 
variables. 
Our language adopts a \emph{call-by-value} strategy;
as usual in functional sessions (cf.~\cite{DBLP:journals/jfp/GayV10}),  session communication arises at the level of \emph{configurations}, which is where deadlocks could occur.
A configuration is either an expression, 
a parallel composition of configurations, or a configuration in the scope of a channel restriction $(\nu{x}{y}^{\priority{\bar{\pi}}})$ that binds together the (dual) endpoints $x$ and $y$. 
It is convenient to consider
configurations with a \emph{flag} $\phi$, useful to indicate whether the configuration is the main thread ($\phi = \mainthread$) or is a child thread ($\phi= \childthread$). 
The syntax of constants, (runtime) values, expressions and configurations is given in \cref{fig:terms}.

Further novelties concern priorities and their associated constructs.
Priority values $\prsigma$ can be a priority $\pr$ (cf. \Cref{fig:types}) or an element from the priority sequence associated with a channel $x$, denoted $\enext^k x $ (for the $k$-th element in the sequence).
Priority abstractions are of the form $\eprabs{\prvar}{e}$, where
$\mathbb{\Lambda}$ binds a priority variable $\prvar$ that might occur free
in $e$. Like type abstractions, priority abstractions of the form 
$\eprabs{\prvar}{v}$ are values. 
Priority applications over 
an expression $e$ are denoted by $\eprapp{e}{\prsigma}$ and $\einst{e}{\envmap}$
(we often omit  $\envmap$ as it is easily inferred from the context). 
In \Cref{fig:terms}, the box enclosing priority applications over channel names  $\eprapp{x}{\prsigma}$ means that they are \emph{runtime values}, not present in source programs.

There are several forms of partial 
applications for constants, which are considered values; they will be justified by 
the types presented in~\cref{fig:types-constants}. 
We denote by $\fv{\e}$ the set of free (term) variables occurring in $\e$.
We have two different ways of declaring sessions (`$\enew{-}$'), depending on whether the session type involved contains a priority abstraction or not; if that is the case, session declaration requires the initial value and increment for the associated priority sequence.
We use $\eseq{e_1}{e_2}$ as an abbreviation for $\elet{x}{e_1}{e_2}$
when $e_1$ has type $\tunit$.

\paragraph{Typing Contexts and Priority Split.}
\cref{fig:terms} also defines the syntax of \emph{typing contexts} (for term variables)
and of \emph{priority maps}. 
The priority map $\envmap$ relates 
type assignments $ x : T $ with 
priority sequences  $\priority{\bar{\pi}}$.
This map is essential to determine the \lstinline|next| priority assigned to a channel. 
We write $\envmap(x)$ instead of $\envmap(x:T)$ when the type $T$ is unimportant or clear from the context. A type $T$ may appear more than once in $\envmap$, associated with different variables.
We write $\envmap(T)$ to denote the (possibly empty) set of variables associated with $T$ in  $\envmap$.
Given $\envmap$, we write $\fst{\envmap(x)}$ to denote  
the  first element in the priority sequence for $x$
and $\tail{\envmap(x)}$ to denote the rest of the sequence. That is, $ \envmap(x) = \fst{\envmap(x)} \cdot \tail{\envmap(x)}$.

Given a priority map $\envmap = \{(x_i:T_i) \mapsto \priority{\bar{\pi}_i}\}_{i\in I}$, the \emph{priority split}
$\envmap = \envmap_1 \circ \envmap_2$ is defined if $\envmap_1 = \{(x_i:T_i) \mapsto \priority{\bar{\pi}_i^1}\}_{i\in I}$
and $\envmap_2 = \{(x_i:T_i) \mapsto \priority{\bar{\pi}_i^2}\}_{i\in I}$ and  $\priority{\bar{\pi}_i}=\priority{\bar{\pi}_i^1} \cdot \priority{\bar{\pi}_i^2}$ for each $i\in I$. Hence, 
$\priority{\bar{\pi}_i}$ is the   sequence that results from appending $\priority{\bar{\pi}_i^2}$ to $\priority{\bar{\pi}_i^1}$.
We denote by  $\left.\envmap\right|_{\Gamma}$ the restriction of priority map $\envmap$ to the typing context $\Gamma$.
%

We work modulo $\alpha$-conversion of bound names ($\lambda$-binders and restricted channel endpoints), 
applying the same renaming to $\envmap$ whenever applicable.

\paragraph{Priority of a Type.}
To statically check the order of operations, we must determine a type's priority. 
Remarkably, 
priorities of polymorphic types heavily depend on the particular \emph{instances} 
under consideration: for example, the session type
\lstinline|Stream|  in~\cref{lst:stream} is instantiated with two 
different priority sequences   in~\cref{lst:main}. Whereas 
\lstinline|r1| in~\cref{lst:main} has \emph{upcoming priority} 1, 
\lstinline|r2|'s next priority is 2, according to the priorities predefined
in Lines 11 and 12. For this reason, we must  
compute the priorities of types according to the typing contexts.

The priority of a type $T$, defined in \Cref{fig:newpriority-types}, depends on its leading operation. 
To expose it, we use the \emph{unravel function}~\cite{DBLP:journals/tcs/CostaMPV24}, which
unfolds a type appropriately, adapting the behavior of type abstractions to priority abstractions. 
The unravel of $T$, denoted  $\unr{T}$, is such that  $\unr{\tseq{\tskip}{\tout{S}{\pr}}} = \tout{S}{\pr}$,
$\unr{\tpolyp{F}{\prvar}{\intervalvar}{U}} = \tpolyp{F}{\prvar}{\intervalvar}{U}$
and $\unr{\tseq{(\tpolyp{S}{\prvar}{\intervalvar}{S_1})}{S_2}} = \tpolyp{S}{\prvar}{\intervalvar}{(\tseq{S_1}{S_2})}$, 
assuming that $\prvar$ does not occur in $S_2$.
Since recursive
types are required to be contractive
(cf. Rule \textsf{F-Rec} in~\cref{ap:semantics}),  
the function $\keyword{unr}$ terminates and never returns a recursive type.

Remarkably, the priority of a type 
is a set of integers. Typically, the set $\nuLowestpriorityEnv{T}{}$ is a singleton, unless $T$ is polymorphic. 
When $T$ is a polymorphic session type, $\nuLowestpriorityEnv{T}{}$ contains the elements of the priority sequence associated with $T$.
On the other hand, when $T$ represents a priority abstraction, $\nuLowestpriorityEnv{T}{}$ consists of 
a set of integers that can be instantiated to provide a concrete priority.
Formally, $\nuLowestpriorityEnv{T}{}$ relies on contexts $\envmap$, $\envpoly$, $\envpr$, and $\envtvars$. 
However, to lighten up notation, we shall write  $\nuLowestpriorityEnv{T}{\mathit{cond}}$ to denote the priority of $T$, where $\mathit{cond}$ specifies the least condition required from $\envmap$, $\envpoly$, $\envpr$, and $\envtvars$ (which are inferred from the typing rule where the definition is used).
We write $\nuLowestpriorityEnv{T}{}$ when no specific condition is used to compute the priority of~$T$.
Also, we write $\nuLowestpriorityEnv{\Gamma}{}$ for the expected extension to typing contexts $\Gamma$, enabled by $\prlub{\cdot}{\cdot}$.
With a little abuse of notation, we write for instance `$\pr < \nuLowestpriorityEnv{T}{}$'  
and `$ \nuLowestpriorityEnv{T}{} < \prtop$'
to mean that  for every $\priority{\pr_i} \in \nuLowestpriorityEnv{T}{}$, $\pr < \priority{\pr_i}$
and $\priority{\pr_i} < \prtop$, respectively.
Similarly, we may write `$  \prlub{{\nuLowestpriorityEnv{T_1}{}}}{{\nuLowestpriorityEnv{T_2}{}}}$' instead of   `$\prlub{\min{\nuLowestpriorityEnv{T_1}{}}}{\min{\nuLowestpriorityEnv{T_2}{}}}$'. Finally, we write $\pr$ instead of the singleton $\{\pr\}$ when clear from the context.

When the leading operation is a message exchange or a choice, the type 
inherits the priority of that operation; if it is a function, the priority 
of the type is the lower priority bound in the arrow. In either case, the priority is a singleton. The priority of a product type $\tprod{T_1}{T_2}$ is 
the minimum between the least elements of the sets
$\nuLowestpriorityEnv{T_1}{}$ and $\nuLowestpriorityEnv{T_2}{}$.

The priority of a polymorphic type of the form 
$\tpolyp{S}{\prvar}{\intervalvar}{U}$ is 
a set containing the values in an associated priority sequence; the definition 
depends on the particular instances in the context. Because different instances of the same session type
can result in different priority sequences, it is clear that the leading operation 
can occur with different priorities in distinct instances. 
When an instance of the type governs channel $x$, the priority of the type results from replacing the priority variable with the  priority sequence for $x$.
For functional types, we keep track 
of priority restrictions, updating the environment $\envpr$. For type abstractions, 
we update the environment for polymorphic variables $\envpoly$ with the
respective restriction. 
To compute the priority of a type variable $\gamma$, we check whether its associated $\pr$ (obtained via $\envpoly$) contains free priority variables: if $\fpv{\pr} = \emptyset$ then the priority of $\gamma$ is $\{\pr\}$;
 otherwise, we obtain the values of the variables in $\fpv{\pr}$ using $\envpr$. 
 In the figure, $ \Theta(\pr)$ represents $\Theta(\prvar)$ if $\pr = \prvar$ and the displacement of $\Theta(\priority{\pr_0})$ by $\priority{n}$ 
 when $\pr = \priority{\pr_0 + n}$.
In the remaining cases, the priority of a type is $\prtop$.

\begin{figure}[!t]
  \[ 
  \nuLowestpriorityEnv{T}{}= 
  \begin{cases} 
    \priority{n} \,&  \text{ if } \unr{T} = \tout{U}{n}, \tin{U}{n}, \textchoice{\ell}{S_\ell}{\ell\in L}{n}, 
    \tintchoice{\ell}{S_\ell}{\ell\in L}{n}, 
    \\
    & \qquad \tclose{n}, \twait{n},\tarrow{T_1}{T_2}{n}{\sigma}{m}, \tseq{\tout{U}{n}}{S} \text{ or } \tseq{\tin{U}{n}}{S}
    \\
    \prlub{{\nuLowestpriorityEnv{T_1}{}}}{{\nuLowestpriorityEnv{T_2}{}}} & \text{ if }\unr{T} = \tprod{T_1}{T_2} 

    \\
    \bigcup_{x\in \envmap(T)}\bigcup_{\pr\in \envmap(x:T)}\{\nuLowestpriorityEnv{\subs{U}{\pr}{\prvar}}{}\} & \text{ if } \unr{T} = \tpolyp{S}{\prvar}{\intervalvar}{U} 
    \\
    \nuLowestpriorityEnv{U}{\prvar\in\intervalvar} & \text{ if } \unr{T} = \tpolyp{F}{\prvar}{\intervalvar}{U}\\
    \nuLowestpriorityEnv{T}{\gamma :: \pr} & \text{ if } \unr{T} = \tpolyt{\gamma}{\pr}{T}{}
    \\
    \pr & \text{ if } \unr{T} = \gamma \text{ and } (\gamma :: \pr) \in \envpoly \text{ and } \fpv{\pr} = \emptyset
    \\
    \Theta(\pr) & \text{ if } \unr{T} = \gamma \text{ and } (\gamma :: \pr) \in \envpoly 
    \\
    \prtop & \text{ otherwise}
  \end{cases}
  \]
  \caption{Priority of types.}
  \label{fig:newpriority-types}
\end{figure}

\paragraph{Context Formation and Split.}
Context formation and context split are standard; see~\cref{ap:semantics} for their definition.
Context split handles the distribution 
of bindings depending on their priority, recalling that types with priority 
$\prtop$ are  {unrestricted}. The typing environments needed to determine the 
priority of $T$ are prescribed by the typing rules, where context split is invoked
(see~\cref{fig:typing-rules-expr}).

\begin{figure}[!t]
  \begin{align*}
    \esend & :  
    \tpolyp {}{\priority{\prvar_1}}
            {\intervaloc{\prbot}{\prtop}}
            {\tpolyt {\tvar} {\prvar_1} 
                     {\tpolyp{}{\priority{\prvar_2}}
                             {\intervaloo{\prbot}{\prvar_1}}
                             {\tpolyp{}{\prvar_3}{\intervaloc{\prvar_2}{\prtop}}{\tpolyt {\svar} 
                                       {\prvar_3}
                                       {\tarrow{\tarrow{\tvar}{\tseq{\tout{\tvar}{\prvar_2}}{\svar}}{\prtop}{\prbot}{\unmult}}{\svar}{\top}{\prvar_2}{\linmult}}{}
                                       }}
                              }
                      {?}}
                     \\
    \ereceive & : 
    \tpolyp {}{\prvar_1}
            {\intervaloc{\prbot}{\prtop}}
            {\tpolyt {\tvar}
                     {\prvar_1}
                     {\tpolyp {}{\prvar_2}
                              {\intervaloo{\prbot}{\prvar_1}}
                              {\tpolyp{}{\prvar_3}{\intervaloc{\prvar_2}{\prtop}}{\tpolyt {\svar}
                                       {\prvar_3}
                                       {\tarrow{\tseq{\tin{\tvar}{\prvar_2}}{\svar}}{\tprod{\tvar}{\svar}}{\top}{\prvar_2}{\unmult}}
                                       {?}}}
                              }
                      {?}}
    \\
                \efork & :  
    \tpolyp {}{\priority{\prvar_1}}
            {\intervaloc{\prbot}{\prtop}}
            {\tpolyp {}{\priority{\prvar_2}}
                     {\intervalcc{\prbot}{\prtop}}
                     {\tarrow{(\tarrow{\tunit}{\tunit}{\prvar_1}{\prvar_2}{\unmult})}{\tunit}{\prtop}{\prbot}{\unmult}}
                     } 
            \\
    \efix & : \tpolyp{}{\prvar}{\intervalcc{\prbot}{\prtop}}{\tpolyt{\tvar}{\prvar}{\tarrow{(\tarrow{(\tarrow{\tvar}{\tvar}{\prtop}{\prbot}{\unmult})}{(\tarrow{\tvar}{\tvar}{\prtop}{\prbot}{\unmult})}{\prtop}{\prbot}{\unmult})}{(\tarrow{\tvar}{\tvar}{\prtop}{\prbot}{\unmult})}{\prtop}{\prbot}{\unmult}}{}} 
    \\
    \eclose & :  \tpolyp{}{\prvar}{\intervaloc{\prbot}{\prtop}}{\tarrow{\tclose{\prvar}}{\tunit}{\prtop}{\prvar}{\unmult}} 
    \qquad \quad 
    \ewait  :  \tpolyp{}{\prvar}{\intervaloc{\prbot}{\prtop}}{\tarrow{\twait{\prvar}}{\tunit}{\prtop}{\prvar}{\unmult}} 
    \qquad \quad 
    \eunit  :  \tunit 
  \end{align*}
      \vspace{-6mm}
  \caption{Types for constants.}
  \label{fig:types-constants}
    \vspace{-2mm}
\end{figure}

Types for constants are defined in
\cref{fig:types-constants}. By upgrading 
priorities to the level of types, we can give proper (functional) types to all constants, 
instead of relying on type schemes~\cite{DBLP:journals/lmcs/KokkeD23,padovani:hal-00954364}.
Following Padovani and Novara~\cite{padovani:hal-00954364}, 
our priority binders control the order of operations  
in a type. 
A representative case is $\esend$.
This type first binds the priority 
$\priority{\prvar_1}$ of the payload, then requires the priority of the 
send operation to be lower (i.e., more urgent) than that of the payload 
and, finally, the lowest priority of the continuation channel should be 
higher than that of the send operation. Once every type and priority variables are instantiated, function $\esend$ expects an 
input value to be sent, followed by a channel of type 
$\tseq{\tout{\tvar}{\prvar_2}}{\svar}$, and returns the continuation channel 
$\svar$. It has latent effect $\priority{\prvar_2}$, meaning that as soon 
as function $\esend$ is fully applied, it triggers operations with 
priority at most $\priority{\prvar_2}$. 

The types for constants should clarify the different forms of values in~\cref{fig:terms}. Consider the type of $\esend$: we can apply it to a priority and get a partial application 
$\eprapp{\esend}{\priority{\prsigma_1}}$, or further make a type application and get 
$\etapp{\eprapp{\esend}{\priority{\prsigma_1}}}{T}$. Then, we can do one or two 
priority applications again, reaching $\eprapp{\etapp{\eprapp{\esend}{\priority{\prsigma_1}}}{T}}{\priority{\prsigma_2}}$ 
and $\eprapp{\eprapp{\etapp{\eprapp{\esend}{\priority{\prsigma_1}}}{T}}{\priority{\prsigma_2}}}{\priority{\prsigma_3}}$. At 
this point, we are left with another type application yielding  
$\etapp{\eprapp{\eprapp{\etapp{\eprapp{\esend}{\priority{\prsigma_1}}}{T}}{\priority{\prsigma_2}}}{\priority{\prsigma_3}}}{S}$,  
which has type $\tarrow{\tarrow{T}{\tseq{\tout{T}{\prsigma_2}}{S}}{\prtop}{\prbot}{\unmult}}{S}{\top}{\prsigma_2}{\linmult}$. 
Before we get $\esend$ fully applied, we can do a term-level application of the 
previous expression to a value $v$ and get the expression 
$\eapp{\etapp{\eprapp{\eprapp{\etapp{\eprapp{\esend}{\priority{\prsigma_1}}}{T}}{\priority{\prsigma_2}}}{\priority{\prsigma_3}}}{S}}{v}$, 
which is still   a value.
For the other constants, the reasoning is similar.

\paragraph{Typing Rules.}
The typing judgement  $\typingpoly{\envpoly}{\envtvars}{\envpr}{\envmap}{e}{T}{\bar{\pi}}{\pr}{\oracons}$ means
that an expression $e$ has type $T$ and latent effect $\pr$ under the typing environments $\envpoly, \envtvars$, $\envpr$, and $\envmap$.
%
The evaluation of expression $e$ might yield a sequence of priorities $\priority{\bar{\pi}}$ 
which predetermines the possible instantiations of polymorphic channels---this annotation is useful
for creating new channels. The predefined priorities are mapped to the corresponding channel names
and stored in $\envmap$.

\begin{figure}
  \begin{mathpar}
    \infer[\textsf{T-Const}]
      {\wfpolyctx{\envpoly}{\envpr}{\envtvars} \qquad  \nuLowestpriorityEnv{\envtvars}{} = \prtop}
      {\typingpoly{\envpoly}{\envtvars}{\envpr}{\envmap}{c}{\typeof c}{\emptyenv}{\bot}{\oracons}}
    \quad
    \infer[\textsf{T-Var}]
    {\wfpolyctx{\envpoly}{\envpr}{\envtvars} \qquad \nuLowestpriorityEnv{\envtvars}{} = \prtop}
      {\typingpoly{\envpoly}{\envtvars, \gamma : T}{\envpr}{\envmap}{\gamma}{T}{\emptyenv}{\bot}{\oracons}}
    \\
    \infer[\textsf{T-TAbs}]
    {\typingpoly{\envpoly, \gamma :: \pr}{\envtvars}{\envpr}{\envmap}{e}{T}{\emptyenv}{\prsigma}{\oracons} \\\\
      \nuLowestpriorityEnv{\envtvars}{} = \prtop \quad 
      \fpv{\pr}\in \dom{\envpr}
       }
      {\typingpoly{\envpoly}{\envtvars}{\envpr}{\envmap}{\etabs{\gamma}{e}}{\tpolyt{\gamma}{\pr}{T}{\prsigma}}{\emptyenv}{\bot}{\oracons}}
    \quad
    \infer[\textsf{T-PAbs}]
      {\typingpoly{\envpoly}{\envtvars}{\envpr, \prvar\in \intervalvar}{\envmap}{e}{T}{\emptyenv}{\prsigma}{\oracons}\\\\
      \nuLowestpriorityEnv{\envtvars}{} = \prtop
      \quad 
      \fpv{\intervalvar}\in \dom{\envpr}}
      {\typingpoly{\envpoly}{\envtvars}{\envpr}{\envmap}{\eprabs{\prvar}{e}}{\tpolyp{}{\prvar}{\intervalvar}{T}}{\emptyenv}{\bot}{\oracons}}
    \\
    \infer[\textsf{T-TApp}]
      {\typingpoly{\envpoly}{\envtvars}{\envpr}{\envmap}{e}{\tpolyt{\tvar}{\pr}{U}{\prsigma}}{\emptyenv}{\pr_1}{\oracons}
        \\\\
        \istype{\envpoly}{\envpr}{T}
        \quad 
        \pr \in \nuLowestpriorityEnv{T}{} 
        }
      {\typingpoly{\envpoly}{\envtvars}{\envpr}{\envmap}{\etapp {e}{T}}{\subs{U}{T}{\tvar}}{\emptyenv}{\pr_1}{\oracons}}
    \quad
    \infer[\textsf{T-PApp}]
      {\typingpoly{\envpoly}{\envtvars}{\envpr}{\envmap}{e}{\tpolyp{}{\prvar}{\intervalvar}{U}}{\emptyenv}{\pr_1}{\oracons}
        \quad 
        \pr \in \intervalvar}
      {\typingpoly{\envpoly}{\envtvars}{\envpr}{\envmap}{\eprapp {e}{\pr}}{\subs{U}{\pr}{\prvar}}{\emptyenv}{\pr_1}{\oracons}}
    \\
     \infer[\textsf{T-Inst}]
      {\typingpoly{\envpoly}{\envtvars}{\envpr}{\envmap}{x}{\tpolyp{S}{\prvar}{\intervalvar}{S}}{\emptyenv}{\pr_1}{\oracons}
        \quad 
        \pr = \fst{\envmap(x)}
        \quad
        \pr \in \intervalvar
        }
      {\typingpoly{\envpoly}{\envtvars}{\envpr}{\envmap}{\einst{x}{\envmap}}{\subs{S}{\pr}{\prvar}}{\emptyenv}{\pr_1}{\oracons}}
    \\
    \infer[\textsf{T-App}]
      {\typingpoly{\envpoly}{\envtvars_1}{\envpr}{\envmap}{e_1}{\tarrow{T_1}{T_2}{\pr}{\prsigma}{m}}{\emptyenv}{\rho_1}{\oracons}
        \quad
        \typingpoly{\envpoly}{\envtvars_2}{\envpr}{\envmap}{e_2}{T_1}{\emptyenv}{\rho_2}{\oracons}
        \\
        \priority{\rho_1} < \nuLowestpriorityEnv{\envtvars_2}{}
        \quad 
        \priority{\rho_2} < \pr
        \quad 
        }
      {\typingpoly{\envpoly}{\envtvars_1 \circ \envtvars_2}{\envpr}{\envmap}{\eapp {e_1}{e_2}}{T_2}{\emptyenv}{\prglb{\prglb{\pr_1}{\pr_2}}{\prsigma}}{\oracons}}
    \\
      \infer[\textsf{T-Let}]
      {\typingpoly{\envpoly}{\envtvars_1}{\envpr}{\envmap}{e_1}{T}{\emptyenv}{\pr_1}{\oracons}
        \quad 
        \exists \envmap_1 \,.\, \envmap = \envmap_1\circ \envmap_2\text{ and } \left.\envmap\right|_{\envtvars_2} = \left.\envmap_2\right|_{\envtvars_2}
        \\\\
        \typingpoly{\envpoly}{\envtvars_2, x : T}{\envpr}{\envmap_2}{e_2}{S}{\emptyenv}{\pr_2}{\oracons}
        \quad 
        \quad
        \priority{\rho_1} < \nuLowestpriorityEnv{\envtvars_2}{}
      }
      {\typingpoly{\envpoly}{\envtvars_1\circ \envtvars_2}{\envpr}{\envmap}{\elet{x}{e_1}{e_2}}{S}{\emptyenv}{\prglb{\pr_1}{\pr_2}}{\oracons}}
    \\
     \infer[\textsf{T-LetPair}]
      {\typingpoly{\envpoly}{\envtvars_1}{\envpr}{\envmap}{e_1}{\tprod{T_1}{T_2}}{\bar{\pi}}{\pr_1}{\oracons}
      \quad 
        \exists \envmap_1 \,.\, \envmap = \envmap_1\circ \envmap_2\text{ and } \left.\envmap\right|_{\envtvars_2} = \left.\envmap_2\right|_{\envtvars_2}
        \\\\
        \text{if } \priority{\bar{\pi}} \neq \priority{\emptyenv} \text{ then }
        \envmap_2' = \envmap_2,x_1:T_1 \mapsto \priority{\bar{\pi}}, x_2:T_2 \mapsto \priority{\bar{\pi}}
        \text{ else } \envmap_2' = \envmap_2
        \\\\
        \typingpoly{\envpoly}{\envtvars_2, x_1 : T_1, x_2 : T_2}{\envpr}{\envmap_2'}{e_2}{S}{\emptyenv}{\pr_2}{\oracons}
        \quad
        \priority{\rho_1} < \nuLowestpriorityEnv{\envtvars_2}{}
      }
      {\typingpoly{\envpoly}{\envtvars_1\circ \envtvars_2}{\envpr}{\envmap}{\eletpair{x_1}{x_2}{e_1}{e_2}}{S}{\emptyenv}{\prglb{\pr_1}{\pr_2}}{\oracons}}

    \infer[\textsf{T-Sel}]
      {\wfpolyctx{\envpoly}{\envpr}{\envtvars}
        \quad
        \istype{\envpoly}{\envpr}{S_\ell}
        \quad 
        \pr < \nuLowestpriorityEnv{S_\ell}{}
        \enspace
        (\forall \ell\in L)
        \quad
        k \in L
        \quad 
        \nuLowestpriorityEnv{\envtvars}{} = \prtop}
      {\typingpoly{\envpoly}{\envtvars}{\envpr}{\envmap}{\eselect{k}}{\tarrow{\tintchoice{\ell}{S_\ell}{\ell\in L}{\pr}}{S_k}{\prtop}{\glb_{\ell\in L}\nuLowestpriorityEnv{S_\ell}{}}{m}}{\emptyenv}{\bot}{\oracons}} 
      \\
    \infer[\textsf{T-Match}]
      { \typingpoly{\envpoly}{\envtvars_1}{\envpr}{\envmap_1}{e}{\textchoice{\ell}{S_\ell}{\ell\in L}{\pr}}{\emptyenv}{\priority{\pr_1}}{\oracons}
        \quad
        \priority{\pr_1} < \nuLowestpriorityEnv{\envtvars_2}{}
        \\
        \typingpoly{\envpoly}{\envtvars_2}{\envpr}{\envmap_2}{e_\ell}{\tarrow{S_\ell}{T}{\nuLowestpriorityEnv{\envtvars_2}{}}{\pr_2}{\linmult}}{\emptyenv}{\prbot}{\oracons}
      }
      {\typingpoly{\envpoly}{\envtvars_1 \circ \envtvars_2}{\envpr}{\envmap_1 \circ \envmap_2}{\ematch{e}{\ell}{e_\ell}{\ell\in L}}{T}{\emptyenv}{\prglb{\prglb{\pr_1}{\pr_2}}{\pr}}{\oracons}}
      \\
    \infer[\textsf{T-NewPoly}]
      {\wfpolyctx{\envpoly}{\envpr}{\envtvars}
      \\
      \istype{\envpoly}{\envpr}{\EuScript{S}^{\fforall}}
      \\
      \nuLowestpriorityEnv{\envtvars}{} = \prtop
      \\ 
      \priority{\bar{\pi}} = [\priority{n_1} + k \cdot \priority{n_2} \mid k\in\mathbb{N}]
      }
      {\typingpoly{\envpoly}{\envtvars}{\envpr}{\envmap}{\enewpoly{\EuScript{S}^\fforall}{n_1}{n_2}}{\tprod{\EuScript{S}^\fforall}{\dualof{\EuScript{S}^\fforall}}}{\bar{\pi}}{\bot}{\oracons}}
    \\
    \infer[\textsf{T-New}]
      {\wfpolyctx{\envpoly}{\envpr}{\envtvars}
      \\
      \istype{\envpoly}{\envpr}{\EuScript{S}}
      \\
      \nuLowestpriorityEnv{\envtvars}{} = \prtop
      \\ 
      \unr{\EuScript{S}} \neq \tskip
      }
      {\typingpoly{\envpoly}{\envtvars}{\envpr}{\envmap}{\enew{\EuScript{S}}}{\tprod{\EuScript{S}}{\dualof{\EuScript{S}}}}{\emptyenv}{\bot}{\oracons}}
  \end{mathpar}
  \caption{Typing rules for expressions (excerpt).}
  \label{fig:typing-rules-expr}
\end{figure}

\cref{fig:typing-rules-expr} presents selected typing rules for expressions and 
\cref{fig:typing-rules-proc} presents the typing rules for configurations, which 
result from incorporating the priority-based approach to 
deadlock freedom~\cite{DBLP:journals/lmcs/KokkeD23,padovani:hal-00954364,DBLP:conf/forte/PadovaniN15}
in the type system of context-free session types~\cite{DBLP:conf/tacas/AlmeidaMV20,DBLP:conf/icfp/ThiemannV16}. 
Now we discuss the typing rules of \cref{fig:typing-rules-expr}; 
the complete set of typing rules (together with additional discussion) 
is given in \cref{ap:semantics}.

Constants are 
typed through \textsf{T-Const} according to the types 
 in~\cref{fig:types-constants}; they have no 
priority effect ($\prbot$), do not yield any priority sequence  
($\priority{\bar{\pi}}=\emptyenv$), and assume that the context 
$\envtvars$ has priority $\prtop$, i.e., it is 
\emph{unrestricted}. Rule \textsf{T-Var} checks 
that variables are typed under an unrestricted environment, 
with no effects. 

We have abstractions on types, priorities, and terms.
Rule \textsf{T-TAbs} states that a type 
abstraction $\etabs{\gamma}{\e}$ has type $\tpolyt{\gamma}{\pr}{T}{}$ with no effect if $\e$ has type $T$ under the extended 
polymorphic environment $\envpoly, \gamma :: \pr$, assuming 
that the free priority variables occurring in $\pr$ are 
defined in $\envpr$ and that   $\envtvars$ is 
unrestricted, because we do not want to abstract over linear values. 
Rule \textsf{T-PAbs} is similar with the expected adaptation to priority abstractions.
Rules \textsf{T-AbsLin} and \textsf{T-AbsUn} consider linear and unrestricted functions, which are 
 presented in~\cref{ap:semantics}. 

Rule \textsf{T-TApp} types the application $\etapp {e}{T}$ requiring that the type variable $\tvar$ is associated to a priority  
$\pr \in \nuLowestpriorityEnv{T}{}$, interpreted under environment $\envpr$. 
Rule \textsf{T-PApp}, distinctive of our work, is applicable to functional and session types
(recalling that the ins\-tan\-tia\-tion of priorities in sessions is only a runtime value),
and states that a priority $\pr$ can instantiate 
$\prvar$ if $\pr$ belongs to the priority interval assigned to 
$\prvar$;  all the occurrences of $\prvar$ in $U$
are updated with the new value $\pr$. 
Rule \textsf{T-Inst} governs channel instantiation: it takes the first priority $\pr$ in the 
  sequence for $x$ ($\pr=\fst{\envmap(x)}$) and instantiates the 
continuation type (similarly to rule \textsf{T-PApp}).  

Rule \textsf{T-App} is standard for \emph{call-by-value} evaluation~\cite
{padovani:hal-00954364}: it requires $e_1$ to have a 
\emph{more urgent} priority than the typing context for $e_2$ and the effect of $e_2$ to be \emph{more 
urgent} than the priority of the type of $e_1$ (which has reduced to a value by then). 
The typing environment is split to properly distribute channels, 
but the priority map is passed to both premises to ensure that the 
priority sequence is available for instantiation.
Unlike $\envtvars$, the priority map $\envmap$ is not a linear resource: it is only 
consulted by $\mathsf{inst}^\envmap$ (and updated by `popping' in sequencing rules). 

Rule \textsf{T-Let} uses the full map $\envmap$ in the first premise to enable priority
instantiation (e.g., reasoning about $\mathsf{next}^k$), whereas the second premise is 
checked under an updated map $\envmap_2$ that reflects the priority information remaining 
after evaluating $\e_1$ (as captured by $\envmap_1$ in $\envmap=\envmap_1\circ\envmap_2$).
This rule also checks that the channels used for evaluating 
$\e_1$ have lower priority than those occurring in $\e_2$. 
This constraint enforces an increasing order of priorities between 
successive operations, thus complementing the type formation rules, as observed in~\cref{sec:types}.

Rule \textsf{T-LetPair} checks that the effect of $e_1$ 
is lower than the lowest priority of the typing environment for $\e_2$.
This rule is also responsible for updating the map of priorities when the 
priority sequence emerging from the evaluation of $e_1$ is not empty 
(that is, when the rule preceding \textsf{T-LetPair} is \textsf{T-NewPoly}, see below).
More specifically, when the priority sequence $\priority{\bar{\pi}}$ is not empty, 
$e_2$ is typed under the context with a priority extended with 
a mapping of the new channel endpoints ($x_1$ and $x_2$) to the corresponding priority sequence. 

Rule 
\textsf{T-Sel} assigns no effect ($\prbot$) to the evaluation of 
$\eselect{k}$, requires the priority of the choice operation to be 
lower than that of the continuation types and also that the typing 
environment has no channel left, i.e., has priority $\prtop$. Rule 
\textsf{T-Match} treats $\e_\ell$ as if it were a function (see \textsf
{T-AbsLin} in~\cref{ap:semantics}) and, according to the call-by-value evaluation strategy, 
requires the evaluation of $\e$ to be more urgent than that of 
$\e_\ell$; the distribution of channel names by the priority maps follows 
the usual rules for splitting linear contexts.

Typing rules for channel creation distinguish session types with a priority 
binder (denoted by $\EuScript{S}^\fforall$) from session types without (denoted by $\EuScript{S}$). Channels of the first 
kind must be assigned a  priority sequence. 
In rule \textsf{T-NewPoly}, the primitive $\enewk$ receives a session type with a priority 
abstraction $\EuScript{S}^\fforall$ 
that indicates the type of the 
channel being created and 
{two natural numbers ($\priority{n_1}$ and $\priority{n_2}$) representing 
the start and increment of its priority sequence. 
The rule defines a sequence starting at $\priority{n_1}$ and with increment $\priority{n_2}$, which is passed as a side 
effect through the first annotation $\priority{\bar{\pi}}$ (that should be eventually mapped to the 
corresponding endpoints and stored in $\envmap$, cf. rule \textsf{T-LetPair}).
Rule \textsf{T-New} requires the typing environment to be un\-re\-strict\-ed
and type $\EuScript{S}$ to be well-formed and such that $\unr{\EuScript{S}} \neq \tskip$, to 
ensure that there are some pending operations. 


\begin{figure}[!t]
    \begin{mathpar}
    \infer[\textsf{{C-Child}}]
      {\typingpoly{\emptyenv}{\envtvars}{\emptyenv}{\Psi}{e}{\tunit}{\bar{\pi}}{\pr}{\oracons}
      }
      {\runtimetyping{\Gamma}{\confthread{\childthread}{e}}{\childthread}{\Psi}}
    \qquad
    \infer[\textsf{{C-Main}}]
      {\typingpoly{\emptyenv}{\envtvars}{\emptyenv}{\Psi}{e}{T}{\bar{\pi}}{\pr}{\oracons}
      \quad 
      \nuLowestpriorityEnv{T}{} = \prtop
      }
      {\runtimetyping{\Gamma}{\confthread{\mainthread}{e}}{\mainthread}{\Psi}}
    \\
    \infer[\textsf{{C-Par}}]
      {\runtimetyping{\Gamma_1}{\config{C}_1}{\thread_1}{\Psi_1} \quad 
      \runtimetyping{\Gamma_2}{\config{C}_2}{\thread_2}{\Psi_2}}
      {\runtimetyping{\Gamma_1 \circ \Gamma_2}{\confpar{\config{C}_1}{\config{C}_2}}{\thread_1+\thread_2}{\Psi_1 \circ \Psi_2}}
    \qquad 
      \infer[\textsf{{C-New}}]
      {\runtimetyping{\Gamma, x: S, y: \dualof{S}}{\config{C}}{\thread} {\Psi_2}
      \quad 
      \istype{\emptyenv}{\emptyenv}{S}
      \quad 
      \Psi_2 = 
       \Psi_1,x: S \mapsto \priority{\bar{\pi}}, y: \dualof{S} \mapsto \priority{\bar{\pi}}
      }
      {\runtimetyping{\Gamma}{\confnu{x}{y}{\config{C}}{\bar{\pi}}}{\thread}{\Psi_1}} 
  \end{mathpar}
    \vspace{-2mm}
  \caption{Typing rules for configurations.}
  \label{fig:typing-rules-proc}
  \vspace{-2mm}
\end{figure}

The judgement for configurations  $\runtimetyping{\envtvars}{\config{C}}{\thread}{\Psi}$ means that  $\config{C}$ is well-typed under   $\envtvars$ with priority map $\envmap$.
The typing rules are in 
\cref{fig:typing-rules-proc}. 
Rules \textsf{C-Child} and \textsf{C-Main} distinguish between child and main threads: 
while any expression (closed for priority and polymorphic variables) 
with unrestricted type can run in the main thread, only expressions 
with type $\tunit$ can run in a child thread; the expression is evaluated 
under $\envtvars$ and $\envmap$. 
Rules \textsf{C-Par} splits the context and the priority map, and 
considers a flag obtained as follows:
 $\childthread+\childthread = \childthread$, 
    $\childthread+\;\mainthread = \mainthread$, and
    $\mainthread+\;\childthread = \mainthread$ 
    (no rule for $\mainthread+\;\mainthread$).
Rule \textsf{C-New} types  
$\confnu{x}{y}{\config{C}}{\bar{\pi}}$: it 
introduces two dual endpoints of a channel in $\envtvars$ and extends the map $\envmap$ for $x$ and $y$ with the sequence $\priority{\bar{\pi}}$.

\section{Dynamics: Type Preservation and Deadlock Freedom} 
\label{sec:safety}
We define reduction
for  expressions and 
configurations  considering a call-by-value evaluation 
strategy.
Reduction rules for expressions are standard for a call-by-value setting; see~\cref{ap:safety}.
Novelties are presented in~\cref{fig:reduction}. A new reduction rule for 
channel instantiations (\textsf{E-PInst}) and a
$\beta$-reduction rule for 
priority-level abstractions (\textsf{E-PApp}). 
Evaluation 
contexts also reflect the left-to-right evaluation strategy and 
are defined by the grammar in the top-left of~\cref{fig:reduction}. 

Reduction includes
rules to lift reductions from expressions to configurations 
(Rule \textsf{R-LiftE}) and from configurations to configuration contexts
(Rule \textsf{R-LiftC}). 
Rule \textsf{R-Fork} creates a new \emph{child} thread, while rules 
\textsf{R-New} and \textsf{R-NewPoly} create a new channel and bind its 
endpoints and the priority sequence in 
a restriction. 
Rules for send/receive, select/match, and close/wait and for configuration reduction 
are standard, and rely on a notion of structural congruence (\cref{ap:safety}).  

\begin{figure}[t!]
  \begin{minipage}{.65\textwidth}
    \declrel{Evaluation contexts}\smallskip\\
    $\evalctx \grmeq \evalhole{} \grmor \eapp{\evalctx}{e} \grmor \eapp{v}{\evalctx} 
    \grmor \etapp{\evalctx}{T} \grmor \eprapp{\evalctx}{\pr}
    \grmor \elet{\evar}{\evalctx}{e} \smallskip\\
    \grmor \eletpair{\evar_1}{\evar_2}{\evalctx}{e}
    \grmor \eseq{\evalctx}{e}$ 
    $  \grmor \epair{\evalctx}{e} 
    \grmor \epair{v}{\evalctx} 
    \grmor \einst{E}{\envmap}\smallskip\\
    \grmor \ematch{\evalctx}{\ell}{e_\ell}{\ell\in L}
    $ \hfill{(Evaluation contexts)}\smallskip\\
    $\threadctx \grmeq \confthread\thread\evalctx$ \hfill{(Thread evaluation contexts)}\smallskip\\
    $\configctx \grmeq \evalhole{} \grmor \confpar{\configctx}{\config{C}} \grmor \confnu{x}{y}{\configctx}{\bar{\pi}}$ \hfill
    {(Configuration contexts)}\smallskip\\
  \end{minipage}
  \begin{minipage}{0.3\textwidth}
    \declrel{Expression reduction}
  \begin{gather*}
    \infer[\textsf{E-PInst}]
      {}
      {\red{\einst{x}{\envmap}}{\eprapp{x}{\priority{\fst{\envmap(x)}}}}}
    \\\\
    \infer[\textsf{E-PApp}]
      {}
      {\red{\eprapp{(\eprabs{\prvar}{v})}{\pr}}{\subs{v}{\pr}{\prvar}}}
  \end{gather*}
  \end{minipage}
  \declrel{Configuration reduction}
  \begin{gather*}
    \infer[\textsf{R-LiftE}]
      {\red{\e_1}{\e_2}}
      {\red{\evalapp{\threadctx}{\e_1}}{\evalapp{\threadctx}{\e_2}}}
    \qquad
    \infer[\textsf{R-LiftC}]
      {\red{\config{C}_1}{\config{C}_2}}
      {\red{\evalapp{\configctx}{\config{C}_1}}{\evalapp{\configctx}{\config{C}_2}}}
    \quad
    \infer[\textsf{R-Fork}]
      {}
      {\red{\evalapp{\threadctx}{\eapp{\eprapp{\eprapp{\efork}{\_}}{\_}}{e}}}
           {\confpar{\evalapp{\threadctx}{\eunit}}{\confthread{\childthread}{\eapp{\e}{\eunit}}}}}
    \\
    \infer[\textsf{R-New}]
      {}
      {\red{\evalapp{\threadctx}{\enew{\EuScript{S}}}}{\confnu{x}{y}{\evalapp{\threadctx}{\epair{x}{y}}}{\emptyenv}}
      \text{\enspace $x, y\not \in \fv{\threadctx}$ and $\unr{\EuScript{S}} \neq \tskip$}}
    \\
    \infer[\textsf{R-NewPoly}]
      {}
      {\red{\evalapp{\threadctx}{\enewpoly{\EuScript{S}^\fforall}{n_1}{n_2}}}{\confnu{x}{y}{\evalapp{\threadctx}{\epair{x}{y}}}{\bar{\pi}}}
      \text{\enspace where $x, y\not \in \fv{\threadctx}$ and $\priority{\bar{\pi}} = [\priority{n_1} + k \cdot \priority{n_2} \mid k\in\mathbb{N}]$}}
  \end{gather*}
  \caption{Evaluation contexts and expression reduction (excerpt).}
  \label{fig:reduction}
\end{figure}

We now state our
type preservation results; their proofs are presented in~\cref{ap:safety}.

\begin{lemma}
  \label{lem:exp-reduction-type-pres}
  If  $\typingpoly{\envpoly}{\envtvars}{\envpr}{\envmap}{e_1}{T}{\bar{\pi}}{\pr}{\oracons}$
  and $\red{e_1}{e_2}$
  then $\typingpoly{\envpoly}{\envtvars}{\envpr}{\envmap}{e_2}{T}{\bar{\pi}}{\pr}{\oracons}$.
\end{lemma}

\begin{theorem}
  \label{lem:subject-red-processes}
  If $\runtimetyping{\envtvars}{\config{C}_1}{\thread}{\Psi}$
  and
   $\red{\config{C}_1}{\config{C}_2}$
  then $\runtimetyping{\envtvars}{\config{C}_2}{\thread}{\Psi}$.
\end{theorem}
We now focus on our goal of proving that well-typed 
(closed) configurations are deadlock-free.
Following Kokke and Dardha~\cite{DBLP:journals/lmcs/KokkeD23}, we rely on some useful notions on expressions.  
An expression \emph{acts on a channel $x$} if it is 
of one of the following forms:
$\eapp{\eapp{\etapp{\eprapp{\etapp{\eprapp{\esend}{\pr}}{T}}{\priority{\sigma}}}{S}}{v}}{x}$, 
or  $\eapp{\etapp{\eprapp{\etapp{\eprapp{\ereceive}{\pr}}{T}}{\priority{\sigma}}}{S}}{x}$, 
or 
 $\eapp{\eprapp{\eclose}{\pr}}{x}$, or  $\eapp{\eprapp{\ewait}{\pr}}{x}$.
In this case, the expression is said to be an \emph{action}.
Actions need to be extended to a notion that considers the 
other forms of \emph{readiness} to communicate.
We thus define an expression  $\e_1$ as \emph{ready} if it 
is of the form $\evalapp{\evalctx}{\e_2}$,
where $\e_2$ is of the form $\enew{S}$, $\enewpoly{\EuScript{S}^\fforall}{n_1}{n_2}$ or 
$\eapp{\eprapp{\eprapp{\efork}{\pr}}{\prsigma}}{\e_3}$ or 
$\e_2$ acts on channel $x$. In the latter case, 
$\e_1$ is said to be \emph{ready to act on $x$}.

We now state  progress  for the term language; see~\cref{ap:safety} for its proof.

\begin{theorem}
  \label{thm:progress}
  If $\typingpoly{\envpoly}{\envtvars}{\envpr}{\envmap}{\e}{T}{\bar{\pi}}{\pr}{\oracons}$ and
  for all $x: T\in \envtvars$, $T$ is a session type, then:
 (i)~$\e$ is a value, or
    (ii)~$\red{\e}{\e'}$ for some $\e'$, or
    (iii)~$\e$ is ready.
\end{theorem}

Following~\cite{DBLP:journals/lmcs/KokkeD23}, we 
say that a configuration $\config{C}$ is in \emph{canonical form}
if it is of the form $\confnu{x_1}{y_1}{\ldots  \confnu{x_n}{y_n}{\confthread{\mainthread}{\e}}{\bar{\pi}_n}}{\bar{\pi}_1}$
or, when it has child threads, it is of the form 
$\confnu{x_1}{y_1}{\ldots  \confnu{x_n}{y_n}{(\confpar{\confthread{\childthread}{e_1}}{\confpar{\ldots}{\confpar{\confthread{\childthread}{\e_m}}{\confthread{\mainthread}{\e}}}})}{\bar{\pi}_n}}{\bar{\pi}_1}$, 
where no expression $\e_i$ is a value.
Intuitively, a configuration   is \emph{deadlock-free} 
if either it reduces forever or, if it terminates, then
all child threads are able to reduce to unit and the main thread  
reduces to an unrestricted value. More precisely:

\begin{definition}
$\config{C}$ is \emph{deadlock-free} if 
$\config{C} \rightarrow^* \config{C}' \not \rightarrow$ then 
$\config{C}'\structcong \confthread{\mainthread}{v}$, where 
$v$ is such that $\typingpoly{\envpoly}{\envtvars}{\envpr}{\envmap}{v}{T}{\bar{\pi}}{\pr}{\oracons}$ with $\nuLowestpriorityEnv{T}{} = \prtop$.
\end{definition}


We can now establish deadlock freedom (see~\cref{ap:safety} for details): 

\begin{theorem}
  \label{thm:config-deadlock-free}
  Suppose $\runtimetyping{\emptyenv}{\config{C}}{\mainthread}{\emptyenv}$. 
 (1) If $\config{C}$ is in 
  canonical form, 
  then $\red{\config{C}}{\config{D}}$ for some 
  $\config{D}$ or $\config{C}\structcong {\confthread{\mainthread}{v}}$, where $v$ is an unrestricted value.
  (2) $\config{C}$  is deadlock-free.
\end{theorem}

\section{Closing Remarks}

\label{sec:conclusion}

We have extended CFSTs 
for high\-er-order concurrent programs~\cite{DBLP:conf/tacas/AlmeidaMV20,DBLP:conf/icfp/ThiemannV16} 
with mechanisms to enforce 
deadlock freedom~\cite{DBLP:journals/lmcs/KokkeD23,padovani:hal-00954364,DBLP:conf/forte/PadovaniN15}. 
Our work shows that the high expressiveness of CFSTs (which includes polymorphism and non-regular recursion) does not conflict with strong correctness guarantees, as CFSTs can be effectively extended with deadlock-freedom guarantees while still supporting programs not expressible in regular session types, including those featuring cyclic process networks (see \cref{ap:cyclic-scheduler} for examples).
Our framework's design is aligned with FreeST, a programming 
language with CFSTs~\cite{freest}. Future work includes developing ways of relieving programmers from defining correct priority sequences. We believe that inference techniques in prior 
works~\cite{DBLP:journals/mscs/CoppoDYP16,DBLP:journals/iandc/0001L17,DBLP:journals/acta/Kobayashi05} may 
provide a good starting point.

\paragraph{Related Work}
\label{sec:related-work}

The priority-based ap\-proach to deadlock freedom can be traced back to seminal  work by Kobayashi~\cite{DBLP:conf/unu/Kobayashi02,DBLP:conf/concur/Kobayashi06}. The resulting classes of typed, deadlock-free processes can express cyclic process networks and encode typed $\lambda$-calculi and the sequencing behavior distinctive of (standard) session types~\cite{DARDHA2017253,DBLP:conf/unu/Kobayashi02}. 
Curry-Howard correspondences between session types and linear logic induce type systems that elegantly ensure
deadlock freedom~\cite{DBLP:conf/concur/CairesP10,DBLP:conf/icfp/Wadler12}. 
A  limitation  is their support for  tree-like topologies only, which leaves out realistic scenarios with cyclic topologies~\cite{DBLP:journals/jlap/DardhaP22}. 
To address this shortcoming, logic-based type systems have been extended with priorities in 
\cite{DBLP:conf/fossacs/DardhaG18} (for synchronous processes without recursive types) and in \cite{DBLP:journals/corr/abs-2110-00146,DBLP:journals/lmcs/0001024} (for asynchronous processes with recursive types). 

All previously mentioned works concern processes, not programs. 
The difference is relevant, as the non-local analyses enabled by priorities can be more directly performed for processes 
than for programs, where communication arises via functional constructs (cf. \cref{fig:types-constants}). 
To our knowledge, Padovani and Novara~\cite{DBLP:conf/forte/PadovaniN15,padovani:hal-00954364} were the first to adapt the 
type system for processes 
in~\cite{padovani_linear_pi} to the case of higher-order concurrent programs. 
A source of inspiration for our work, they use polymorphism and a form of \emph{regular} recursion that appears  limited to specify 
some examples.
Our proof of deadlock freedom has been also influenced by  Kokke and Dardha's work~\cite{DBLP:journals/lmcs/KokkeD23}, who adapted the approach in~\cite{DBLP:conf/forte/PadovaniN15,padovani:hal-00954364} to the case of regular session types without recursion, which is subsumed by our work.

\paragraph{Acknowledgments}
We are grateful to the anonymous reviewers of the current and previous versions of this paper for their constructive feedback.
The first author was funded by  FCT -- Funda\c{c}\~ao para a Ci\^{e}ncia e a Tecnologia, I.P., 
under the project REACT, 
ref. 2023.13752.PEX, DOI \url{https://doi.org/10.54499/2023.13752.PEX}, and 
under the LASIGE Research Unit, ref. UID/00408/2025, 
DOI \url{https://doi.org/10.54499/UID/00408/2025}.
The second author gratefully acknowledges the support of the Dutch Research Council (NWO) under project  ``\href{https://doi.org/10.61686/FHYZO53064}{Cyclic Structures in Programs and
Proofs}'' (OCENW.XL.23.089).



\bibliographystyle{splncs04}
\bibliography{biblio}

\newpage
\appendix
\section{Auxiliary Rules, Concepts and Results}

Before diving into definitions and auxiliary results, we 
present a motivating example illustrating the non-regular 
recursive nature of CFSTs; we omitted it 
from~\cref{sec:deadlocks} due to space constraints.

\subsection{Motivating Example: Deadlock Freedom and Non-Regular Recursion}
\label{ss:trees}
CFSTs shine with \emph{non-regular} recursion. 
The classic example that capitalises on all their features is exchanging binary 
trees~\cite{DBLP:conf/tacas/AlmeidaMV20,DBLP:conf/icfp/ThiemannV16}.
Consider the type \lstinline|TreeChannel| given by:
\[\trec{\svar}{(\tpolyp{S}{\prvar}{\intervaloo{\bot}{\top}}{\&^{\priority{\prvar}}\{\keyword{LeafC}: \tskip, \keyword{NodeC}: \tseq{\tin{\tint}{\prvar+1}}{\tseq{\svar}{\svar}}\}{}}{})}.\]
A channel governed by \lstinline|TreeChannel| starts by 
instantiating the priority variable $\priority{\prvar}$
and then offers two branches:  one with label 
$\keyword{LeafC}$, which  does nothing (represented 
by type $\tskip$); and another 
with label $\keyword{NodeC}$, which receives an integer 
under priority $\priority{\prvar+1}$ and then receives the 
left subtree followed by the right subtree.  
Consider the implementation of function \lstinline|receiveTree| under 
our priority-based approach to deadlock freedom:
\lstset{firstnumber=1}
\begin{lstlisting}
data Tree = Leaf | Node Int Tree Tree

type TreeChannel = ~forallp i belongsTo (bot,top) => ~
    &~i~{LeafC: Skip, 
       NodeC: ?~(i+1)~Int; TreeChannel; TreeChannel}

receiveTree : ~forallp p belongsTo (bot, top) =>~ forall a :: ~p~ =>
             TreeChannel; a ->~[top,top]~ (Tree, a)
receiveTree c =
  match (inst c) with {
    LeafC c ->
      (Leaf, c),
    NodeC c ->
      let (x, c) = receive c in
      let (left, c) = receiveTree~{next c}~ @(TreeChannel;a) c in
      let (right, c) = receiveTree~{next c}~ @a c in
      (Node x left right, c)
  }
\end{lstlisting}
\lstset{firstnumber=1}

The signature of \lstinline|receiveTree| states that it receives a channel of type \lstinline|TreeChannel; a|, where 
\lstinline|a| is a \emph{type variable} that represents a type
with priority \lstinline|~p belongsTo (bot, top)~|.
This type variable is instantiated in a function call 
through the \lstinline|@| operator (Lines 15 and 16).
Finally, it should return a pair with a reconstruction of the
\lstinline|Tree| and the continuation of the input channel. 
Since channel \lstinline|c| is polymorphic, the function
\lstinline|receiveTree| starts by instantiating its priority
variable (\lstinline|inst c| in Line 10).
The \lstinline|match| expression handles the 
choices on the instantiated channel \lstinline|c|: for \lstinline|LeafC| it 
returns a pair with a \lstinline|Leaf| and the 
continuation of the channel; for \lstinline|NodeC|, the 
function  receives the integer on channel
\lstinline|c| and then receives the left and right subtrees. 

Notice that the recursive calls apply polymorphic recursion 
on priorities \emph{and} types. Both recursive calls
instantiate the priority variable \lstinline|~p~| with the 
\lstinline|next| priority assigned to the operations in channel \lstinline|c|;
note that channel \lstinline|c| evolves, so \lstinline|next ~c~| in Line 15
does not have the same value as \lstinline|next ~c~| in Line 16.
Furthermore, in the first call (Line 15), 
the  {type variable} is updated with the 
type of the continuation of the channel 
(\lstinline|TreeChannel;a|); in the second call 
the reminder of the 
channel has now type \lstinline|a|. Finally, the function
returns a pair composed of the reconstruction of the node
followed by the continuation of the channel.

\subsection{Context-Free Session Types, with Priorities}
\label{ap:types}

\paragraph{Type Formation.}
\begin{figure}[!t]
  \begin{gather*}
    \qquad
    \infer[\textsf{F-Unit}]
      {}
      {\istype{\envpoly}{\envpr}{\tunit}}
    \qquad
    \infer[\textsf{F-Abs}]
      {\istype{\envpoly}{\envpr}{T_1} \quad \istype{\envpoly}{\envpr}{T_2}}
      {\istype{\envpoly}{\envpr}{\tarrow{T_1}{T_2}{\pr_1}{\pr_2}{m}}}
    \qquad
    \infer[\textsf{F-Pair}]
      {\istype{\envpoly}{\envpr}{T_1} \quad \istype{\envpoly}{\envpr}{T_2}}
      {\istype{\envpoly}{\envpr}{\tprod{T_1}{T_2}}}
    \\
    \infer[\textsf{F-Close}]
      {}
      {\istype{\envpoly}{\envpr}{\tclose{\pr}}}
    \qquad
    \infer[\textsf{F-Wait}]
      {}
      {\istype{\envpoly}{\envpr}{\twait{\pr}}}
    \qquad
    \infer[\textsf{F-Skip}]
      {}
      {\istype{\envpoly}{\envpr}{\tskip }}
    \qquad
    \infer[\textsf{F-MsgOut}]
      {\istype{\envpoly}{\envpr}{T} \enspace 
      }
      {\istype{\envpoly}{\envpr}{\tout T \pr }}
    \qquad
    \infer[\textsf{F-MsgIn}]
      {\istype{\envpoly}{\envpr}{T} \enspace 
      }
      {\istype{\envpoly}{\envpr}{\tin T \pr}}
    \\
    \infer[\textsf{F-Seq}]
      {\istype{\envpoly}{\envpr}{S_1} \enspace \istype{\envpoly}{\envpr}{S_2}}
      {\istype{\envpoly}{\envpr}{\tseq{S_1}{S_2}}}
    \qquad
    \infer[\textsf{F-IntChoice}]
      {\istype{\envpoly}{\envpr}{S_\ell }
        \enspace (\forall \ell\in L)}
      {\istype{\envpoly}{\envpr}{\tintchoice \ell {S_\ell} {\ell\in L} \pr}}
    \qquad
    \infer[\textsf{F-ExtChoice}]
      {\istype{\envpoly}{\envpr}{S_\ell }
        \enspace (\forall \ell\in L)}
      {\istype{\envpoly}{\envpr}{\textchoice \ell {S_\ell} {\ell\in L} \pr }}
    \\
    \qquad 
    \infer[\textsf{F-TAbs}]
        {\istype{\envpoly, \gamma :: \pr}{\envpr}{T}
        \quad 
        \fpv{\pr}\in \dom{\envpr}}
        {\istype{\envpoly}{\envpr}{\tpolyt{\gamma}{\pr}{T}{\prsigma}}}
      \qquad
            \infer[\textsf{F-Rec}]
        {\istype{\envpoly, \gamma :: \pr}{\envpr}{T}
        \quad \text{T contractive}
        }
        {\istype{\envpoly}{\envpr}{\trec{\gamma}{T}}}
      \\
      \infer[\textsf{F-PAbsF}]
            {\istype{\envpoly}{\envpr, \priority{\prvar\in \intervalvar}}{T}
            \quad 
            \fpv{\intervalvar}\in \dom{\envpr}}
            {\istype{\envpoly}{\envpr}{\tpolyp{F}{\prvar}{\intervalvar}{T}}}
\quad 
      \infer[\textsf{F-PAbsS}]
            {\istype{\envpoly}{\priority{\prvar\in \intervalvar}}{S}
            \quad 
            \fpv{\intervalvar}\in \dom{\envpr}}
            {\istype{\envpoly}{\emptyenv}{\tpolyp{S}{\prvar}{\intervalvar}{S}}}
      \quad
      \infer[\textsf{F-Var}]
              {\gamma \in \dom{\envpoly}}
              {\istype{\envpoly}{\envpr}{\gamma}}
  \end{gather*}
  \caption{Type formation (complete). The main text includes an excerpt in~\cref{fig:type-formation}.}
  \label{fig:ap:type-formation}
\end{figure}

The judgement  $\istype{\envpoly}{\envpr}{T}$ 
denotes that type $T$ is well-formed under the polymorphic context $\envpoly$
and the priority context $\envpr$.
The rules for type formation are 
defined in~\cref{fig:type-formation}.
%
The type formation rules 
\textsf{F-Unit}, \textsf{F-Abs}, \textsf{F-Pair}, \textsf{F-Close}, \textsf{F-Wait}, 
\textsf{F-Skip}, \textsf{F-MsgOut}, \textsf{F-MsgIn}, \textsf{F-Seq}, \textsf{F-IntChoice}, 
\textsf{F-ExtChoice} and \textsf{F-Var} are standard for CFSTs~\cite{DBLP:conf/icfp/ThiemannV16}.

Rule \textsf{F-TAbs} states that a polymorphic type $\tpolyt{\gamma}{\pr}{T}{}$ 
can only assign a type variable $\gamma$
to a priority $\pr$ if all the free priority variables occurring
in $\pr$ are defined in $\envpr$; furthermore, $T$ should be well-formed
under the extended environment $\envpoly, \gamma::\pr$.
Similarly, rule \textsf{F-PAbsF} states that a priority variable $\prvar$
can only be assigned to $\intervalvar$ through the $\fforall^{\fallstyle{F}}$ binder
if all the free variables occurring in $\intervalvar$ are defined 
in $\envpr$; furthermore, the body of the type should be well-formed 
under the priority environment $\envpr, \prvar\in \intervalvar$. 
Rule \textsf{F-PAbsS}
imposes the same restrictions but requires the priority environment to be empty, leading 
to the condition that the
$\fforall^{\fallstyle{S}}$ binder occurs at most once in a well-formed type---this 
restriction is intended to simplify the representation of 
priority sequences (priority maps are defined in~\Cref{fig:terms}).

Rule \textsf{F-Rec} states that a recursive type $\trec{\gamma}{T}$ is 
well-formed if $T$ is well-formed under the extended context $\envpoly, \gamma :: \pr$
and $T$ is contractive, meaning that $T$ should have no subterms 
of the form $\mu \alpha. \mu \alpha_1. \ldots \mu \alpha_n . \alpha$, considering 
an adaptation of the conventional notion of 
contractivity~\cite{DBLP:journals/acta/GayH05} to account 
for $\tskip$ as the neutral element of  sequential composition. (A rule-based 
definition of contractivity can be found in, e.g.,~\cite{DBLP:conf/tacas/AlmeidaMV20,DBLP:conf/concur/SilvaMV23}.)

\paragraph{Type Duality.}

Duality is paramount to establish \emph{compatible} communication 
between two endpoints of a channel. This notion is defined only 
over session types; the definition is 
standard~\cite{DBLP:conf/tacas/AlmeidaMV20,DBLP:conf/icfp/ThiemannV16}
 and is presented in~\cref{fig:ap:duality}.
Since we are dealing with higher-order sessions,
in the sense that channels can be conveyed into messages, 
the duality of recursive types relies on  
\emph{negative references}, denoted by $\beta^-$, which serve to dualize 
recursive sessions properly~\cite{jlamp}. 
Crucially, duality preserves priorities.
The dual of a priority abstraction 
($\dualof{\tpolyp{}{\svar}{\intervalvar}{S}}$)
is the priority abstraction 
of the dual session type, $\tpolyp{}{\svar}{\intervalvar}{\dualof{S}}$.

\begin{figure}
  \begin{gather*}
    \dualof{\tskip} = \tskip
    \quad
    \dualof{\tseq{S_1}{S_2}} = \tseq{\dualof{S_1}}{\dualof{S_2}}
    \quad 
    \dualof{\tout{T}{\pr}} = \tin{T}{\pr} 
    \quad 
    \dualof{\tclose{\pr}} = \twait{\pr} 
    \quad 
    \dualof{\tin{T}{\pr}} = \tout{T}{\pr}
    \quad
    \dualof{\twait{\pr}} = \tclose{\pr}
    \\ 
    \dualof{\tintchoice{\ell}{S_\ell}{\ell\in L}{\pr}} = \textchoice{\ell}{\dualof{S_\ell}}{\ell\in L}{\pr}
    \quad 
    \dualof{\textchoice{\ell}{S_\ell}{\ell\in L}{\pr}} = \tintchoice{\ell}{\dualof{S_\ell}}{\ell\in L}{\pr}
    \\
    \dualof{\trec{\svar}{S}} = \trec{\svar}{\subs{\dualof{\subs{S}{\svar^-}{\svar}}}{\trec{\svar}{S}}{\svar^-}}
    \quad 
    \dualof{\tpolyp{S}{\svar}{\intervalvar}{S}} = \enspace \tpolyp{S}{\svar}{\intervalvar}{\dualof{S}}
  \end{gather*}
  \caption{Type duality on session types.}
  \label{fig:ap:duality}
\end{figure}

\subsection{Statics: Configurations and Typing Rules}
\label{ap:semantics}

\paragraph{Context Formation and Split.}
\cref{fig:ctx-formation} presents context formation, defined as usual from type formation (\Cref{fig:ap:type-formation}).
\cref{fig:ctx-split} presents context split, which handles the distribution 
of bindings depending on their priority, recalling that types with priority 
$\prtop$ are  {unrestricted}. The typing environments needed to determine the 
priority of $T$ are prescribed by the typing rules, where context split is invoked
(see~\cref{fig:typing-rules-expr}).

\begin{figure}[!t]
  \begin{mathpar}
    \infer[\textsf{Ctx-Emp}]
      {}
      {\wfpolyctx{\envpoly}{\envpr}{\emptyenv}}
    \qquad
    \infer[\textsf{Ctx-NEmp}]
      {\wfpolyctx{\envpoly}{\envpr}{\envtvars} \qquad \istype{\envpoly}{\envpr}{T}}
      {\wfpolyctx{\envpoly}{\envpr}{\envtvars,x:T}}
    \qquad
  \end{mathpar}
    \vspace{-2mm}
  \caption{Context formation rules.}
  \label{fig:ctx-formation}
\end{figure}

\begin{figure}[!t]
  \begin{mathpar}
    \infer[\textsf{S-Emp}]
      {}
      {\ctxsplit{\envpoly}{\envpr}{\emptyenv}{\emptyenv}{\emptyenv}}
    \and 
    \infer[\textsf{S-Top}]
      {\ctxsplit{\envpoly}{\envpr}{\envtvars}{\envtvars_1}{\envtvars_2}
      \qquad \istype{\envpoly}{\envpr}{T}
      \qquad \nuLowestpriorityEnv{T}{} = \prtop}
      {\ctxsplit{\envpoly}{\envpr}{\envtvars,x:T}{(\envtvars_1,x:T)}{(\envtvars_2,x:T)}}
    \\
    \infer[\textsf{S-Left}]
      {\ctxsplit{\envpoly}{\envpr}{\envtvars}{\envtvars_1}{\envtvars_2}
      \qquad \istype{\envpoly}{\envpr}{T}
      \qquad \nuLowestpriorityEnv{T}{} < \prtop}
      {\ctxsplit{\envpoly}{\envpr}{\envtvars,x:T}{(\envtvars_1,x:T)}{\envtvars_2}}
    \\  
    \infer[\textsf{S-Right}]
      {\ctxsplit{\envpoly}{\envpr}{\envtvars}{\envtvars_1}{\envtvars_2}
      \qquad \istype{\envpoly}{\envpr}{T}
      \qquad \nuLowestpriorityEnv{T}{} < \prtop}
      {\ctxsplit{\envpoly}{\envpr}{\envtvars,x:T}{\envtvars_1}{(\envtvars_2,x:T)}}
  \end{mathpar}
  \caption{Context split.}
  \label{fig:ctx-split}
\end{figure}


\paragraph{Typing Rules.}
The typing judgement for expressions  {\[\typingpoly{\envpoly}{\envtvars}{\envpr}{\envmap}{e}{T}{\bar{\pi}}{\pr}{\oracons}\]} means
that an expression $e$ has type $T$ and latent effect $\pr$ under the typing environments $\envpoly, \envtvars$, $\envpr$, and $\envmap$.
%
The evaluation of expression $e$ might yield a sequence of priorities $\priority{\bar{\pi}}$ 
which predetermines the possible instantiations of polymorphic channels---this annotation is useful
for creating new channels. The predefined priorities are mapped to the corresponding channel names
and stored in $\envmap$.

\begin{figure}
  \begin{mathpar}
    \infer[\textsf{T-Const}]
      {\wfpolyctx{\envpoly}{\envpr}{\envtvars} \qquad  \nuLowestpriorityEnv{\envtvars}{} = \prtop}
      {\typingpoly{\envpoly}{\envtvars}{\envpr}{\envmap}{c}{\typeof c}{\emptyenv}{\bot}{\oracons}}
    \quad
    \infer[\textsf{T-Var}]
    {\wfpolyctx{\envpoly}{\envpr}{\envtvars} \qquad \nuLowestpriorityEnv{\envtvars}{} = \prtop}
      {\typingpoly{\envpoly}{\envtvars, \gamma : T}{\envpr}{\envmap}{\gamma}{T}{\emptyenv}{\bot}{\oracons}}
    \\
    \infer[\textsf{T-TAbs}]
    {\typingpoly{\envpoly, \gamma :: \pr}{\envtvars}{\envpr}{\envmap}{e}{T}{\emptyenv}{\prsigma}{\oracons} \\\\
      \nuLowestpriorityEnv{\envtvars}{} = \prtop \quad 
      \fpv{\pr}\in \dom{\envpr}
       }
      {\typingpoly{\envpoly}{\envtvars}{\envpr}{\envmap}{\etabs{\gamma}{e}}{\tpolyt{\gamma}{\pr}{T}{\prsigma}}{\emptyenv}{\bot}{\oracons}}
    \quad
    \infer[\textsf{T-PAbs}]
      {\typingpoly{\envpoly}{\envtvars}{\envpr, \prvar\in \intervalvar}{\envmap}{e}{T}{\emptyenv}{\prsigma}{\oracons}\\\\
      \nuLowestpriorityEnv{\envtvars}{} = \prtop
      \quad 
      \fpv{\intervalvar}\in \dom{\envpr}}
      {\typingpoly{\envpoly}{\envtvars}{\envpr}{\envmap}{\eprabs{\prvar}{e}}{\tpolyp{}{\prvar}{\intervalvar}{T}}{\emptyenv}{\bot}{\oracons}}
    \\
    \infer[\textsf{T-AbsLin}]
      {\typingpoly{\envpoly}{\envtvars, x : T_1}{\envpr}{\envmap}{e}{T_2}{\emptyenv}{\rho}{\oracons}}
      {\typingpoly{\envpoly}{\envtvars}{\envpr}{\envmap}{\eabs x{T_1}e{\linmult}}{\tarrow{T_1}{T_2}{\nuLowestpriorityEnv{\Gamma}{}}{\rho}{\linmult}}{\emptyenv}{\bot}{\oracons}}
    \\
    \infer[\textsf{T-AbsUn}]
      {\typingpoly{\envpoly}{\envtvars, x : T_1}{\envpr}{\envmap}{e}{T_2}{\emptyenv}{\rho}{\oracons}
      \quad 
        \nuLowestpriorityEnv{\envtvars}{} = \prtop}
      {\typingpoly{\envpoly}{\envtvars}{\envpr}{\envmap}{\eabs x{T_1}e{\unmult}}{\tarrow{T_1}{T_2}{\nuLowestpriorityEnv{\Gamma}{}}{\rho}{\unmult}}{\emptyenv}{\bot}{\oracons}}
    \\
    \infer[\textsf{T-TApp}]
      {\typingpoly{\envpoly}{\envtvars}{\envpr}{\envmap}{e}{\tpolyt{\tvar}{\pr}{U}{\prsigma}}{\emptyenv}{\pr_1}{\oracons}
        \\\\
        \istype{\envpoly}{\envpr}{T}
        \quad 
        \pr \in \nuLowestpriorityEnv{T}{} 
        }
      {\typingpoly{\envpoly}{\envtvars}{\envpr}{\envmap}{\etapp {e}{T}}{\subs{U}{T}{\tvar}}{\emptyenv}{\pr_1}{\oracons}}
    \quad
    \infer[\textsf{T-PApp}]
      {\typingpoly{\envpoly}{\envtvars}{\envpr}{\envmap}{e}{\tpolyp{}{\prvar}{\intervalvar}{U}}{\emptyenv}{\pr_1}{\oracons}
        \quad 
        \pr \in \intervalvar}
      {\typingpoly{\envpoly}{\envtvars}{\envpr}{\envmap}{\eprapp {e}{\pr}}{\subs{U}{\pr}{\prvar}}{\emptyenv}{\pr_1}{\oracons}}
    \\
     \infer[\textsf{T-Inst}]
      {\typingpoly{\envpoly}{\envtvars}{\envpr}{\envmap}{x}{\tpolyp{S}{\prvar}{\intervalvar}{S}}{\emptyenv}{\pr_1}{\oracons}
        \quad 
        \pr = \fst{\envmap(x)}
        \quad
        \pr \in \intervalvar
        }
      {\typingpoly{\envpoly}{\envtvars}{\envpr}{\envmap}{\einst{x}{\envmap}}{\subs{S}{\pr}{\prvar}}{\emptyenv}{\pr_1}{\oracons}}
    \\
    \infer[\textsf{T-App}]
      {\typingpoly{\envpoly}{\envtvars_1}{\envpr}{\envmap}{e_1}{\tarrow{T_1}{T_2}{\pr}{\prsigma}{m}}{\emptyenv}{\rho_1}{\oracons}
        \quad
        \typingpoly{\envpoly}{\envtvars_2}{\envpr}{\envmap}{e_2}{T_1}{\emptyenv}{\rho_2}{\oracons}
        \\
        \priority{\rho_1} < \nuLowestpriorityEnv{\envtvars_2}{}
        \quad 
        \priority{\rho_2} < \pr
        \quad 
        }
      {\typingpoly{\envpoly}{\envtvars_1 \circ \envtvars_2}{\envpr}{\envmap}{\eapp {e_1}{e_2}}{T_2}{\emptyenv}{\prglb{\prglb{\pr_1}{\pr_2}}{\prsigma}}{\oracons}}
    \\
      \infer[\textsf{T-Pair}]
      {\typingpoly{\envpoly}{\envtvars_1}{\envpr}{\envmap_1}{e_1}{T_1}{\emptyenv}{\pr_1}{\oracons} \quad
      \typingpoly{\envpoly}{\envtvars_2}{\envpr}{\envmap_2}{e_2}{T_2}{\emptyenv}{\pr_2}{\oracons} \\
      \priority{\rho_1} < \nuLowestpriorityEnv{\envtvars_2}{} \quad
      \priority{\rho_2} < \nuLowestpriorityEnv{T_1}{}  
       }
      {\typingpoly{\envpoly}{\envtvars_1 \circ \envtvars_2}{\envpr}{\envmap_1\circ\envmap_2}{\epair{e_1}{e_2}}{\tprod{T_1}{T_2}}{\emptyenv}{\prglb{\pr_1}{\pr_2}}{\oracons}}
    \\
      \infer[\textsf{T-Let}]
      {\typingpoly{\envpoly}{\envtvars_1}{\envpr}{\envmap}{e_1}{T}{\emptyenv}{\pr_1}{\oracons}
        \quad 
        \exists \envmap_1 \,.\, \envmap = \envmap_1\circ \envmap_2\text{ and } \left.\envmap\right|_{\envtvars_2} = \left.\envmap_2\right|_{\envtvars_2}
        \\\\
        \typingpoly{\envpoly}{\envtvars_2, x : T}{\envpr}{\envmap_2}{e_2}{S}{\emptyenv}{\pr_2}{\oracons}
        \quad 
        \quad
        \priority{\rho_1} < \nuLowestpriorityEnv{\envtvars_2}{}
      }
      {\typingpoly{\envpoly}{\envtvars_1\circ \envtvars_2}{\envpr}{\envmap}{\elet{x}{e_1}{e_2}}{S}{\emptyenv}{\prglb{\pr_1}{\pr_2}}{\oracons}}
    \\
     \infer[\textsf{T-LetPair}]
      {\typingpoly{\envpoly}{\envtvars_1}{\envpr}{\envmap}{e_1}{\tprod{T_1}{T_2}}{\bar{\pi}}{\pr_1}{\oracons}
      \quad 
        \exists \envmap_1 \,.\, \envmap = \envmap_1\circ \envmap_2\text{ and } \left.\envmap\right|_{\envtvars_2} = \left.\envmap_2\right|_{\envtvars_2}
        \\\\
        \text{if } \priority{\bar{\pi}} \neq \priority{\emptyenv} \text{ then }
        \envmap_2' = \envmap_2,x_1:T_1 \mapsto \priority{\bar{\pi}}, x_2:T_2 \mapsto \priority{\bar{\pi}}
        \text{ else } \envmap_2' = \envmap_2
        \\\\
        \typingpoly{\envpoly}{\envtvars_2, x_1 : T_1, x_2 : T_2}{\envpr}{\envmap_2'}{e_2}{S}{\emptyenv}{\pr_2}{\oracons}
        \quad
        \priority{\rho_1} < \nuLowestpriorityEnv{\envtvars_2}{}
      }
      {\typingpoly{\envpoly}{\envtvars_1\circ \envtvars_2}{\envpr}{\envmap}{\eletpair{x_1}{x_2}{e_1}{e_2}}{S}{\emptyenv}{\prglb{\pr_1}{\pr_2}}{\oracons}}
  \end{mathpar}
  \caption{Typing rules for expressions (complete), part 1.}
  \label{fig:ap:typing-rules-expr}
\end{figure}

\begin{figure}[!t]
  \begin{mathpar} 
    \infer[\textsf{T-Sel}]
      {\wfpolyctx{\envpoly}{\envpr}{\envtvars}
        \quad
        \istype{\envpoly}{\envpr}{S_\ell}
        \quad 
        \pr < \nuLowestpriorityEnv{S_\ell}{}
        \enspace
        (\forall \ell\in L)
        \quad
        k \in L
        \quad 
        \nuLowestpriorityEnv{\envtvars}{} = \prtop}
      {\typingpoly{\envpoly}{\envtvars}{\envpr}{\envmap}{\eselect{k}}{\tarrow{\tintchoice{\ell}{S_\ell}{\ell\in L}{\pr}}{S_k}{\prtop}{\glb_{\ell\in L}\nuLowestpriorityEnv{S_\ell}{}}{m}}{\emptyenv}{\bot}{\oracons}} 
      \\
    \infer[\textsf{T-Match}]
      { \typingpoly{\envpoly}{\envtvars_1}{\envpr}{\envmap_1}{e}{\textchoice{\ell}{S_\ell}{\ell\in L}{\pr}}{\emptyenv}{\priority{\pr_1}}{\oracons}
        \quad
        \priority{\pr_1} < \nuLowestpriorityEnv{\envtvars_2}{}
        \\
        \typingpoly{\envpoly}{\envtvars_2}{\envpr}{\envmap_2}{e_\ell}{\tarrow{S_\ell}{T}{\nuLowestpriorityEnv{\envtvars_2}{}}{\pr_2}{\linmult}}{\emptyenv}{\prbot}{\oracons}
      }
      {\typingpoly{\envpoly}{\envtvars_1 \circ \envtvars_2}{\envpr}{\envmap_1 \circ \envmap_2}{\ematch{e}{\ell}{e_\ell}{\ell\in L}}{T}{\emptyenv}{\prglb{\prglb{\pr_1}{\pr_2}}{\pr}}{\oracons}}
      \\
    \infer[\textsf{T-NewPoly}]
      {\wfpolyctx{\envpoly}{\envpr}{\envtvars}
      \\
      \istype{\envpoly}{\envpr}{\EuScript{S}^{\fforall}}
      \\
      \nuLowestpriorityEnv{\envtvars}{} = \prtop
      \\ 
      \priority{\bar{\pi}} = [\priority{n_1} + k \cdot \priority{n_2} \mid k\in\mathbb{N}]
      }
      {\typingpoly{\envpoly}{\envtvars}{\envpr}{\envmap}{\enewpoly{\EuScript{S}^\fforall}{n_1}{n_2}}{\tprod{\EuScript{S}^\fforall}{\dualof{\EuScript{S}^\fforall}}}{\bar{\pi}}{\bot}{\oracons}}
      \\
      \infer[\textsf{T-New}]
      {\wfpolyctx{\envpoly}{\envpr}{\envtvars}
      \\
      \istype{\envpoly}{\envpr}{\EuScript{S}}
      \\
      \nuLowestpriorityEnv{\envtvars}{} = \prtop
      \\ 
      \unr{\EuScript{S}} \neq \tskip
      }
      {\typingpoly{\envpoly}{\envtvars}{\envpr}{\envmap}{\enew{\EuScript{S}}}{\tprod{\EuScript{S}}{\dualof{\EuScript{S}}}}{\emptyenv}{\bot}{\oracons}}
    \\
    \infer[\textsf{T-Eq}]
      {\typingpoly{\envpoly}{\envtvars}{\envpr}{\envmap}{\e}{T_1}{\emptyenv}{\pr}{\oracons}
      \\ 
      \envpoly\mid \envpr \vdash T_1 \simeq T_2 
      }
      {\typingpoly{\envpoly}{\envtvars}{\envpr}{\envmap}{\e}{T_2}{\emptyenv}{\pr}{\oracons}}
  \end{mathpar}
  \caption{Typing rules for expressions (complete), part 2.}
  \label{fig:ap:typing-rules-expr-core}
\end{figure}

\cref{fig:ap:typing-rules-expr,fig:ap:typing-rules-expr-core,fig:typing-rules-proc} present our typing rules for expressions and configurations, which 
result from incorporating the priority-based approach to 
deadlock freedom~\cite{DBLP:journals/lmcs/KokkeD23,padovani:hal-00954364,DBLP:conf/forte/PadovaniN15}
in the type system of context-free session types~\cite{DBLP:conf/tacas/AlmeidaMV20,DBLP:conf/icfp/ThiemannV16}. 
We discuss all the typing rules. 

Constants are 
typed through \textsf{T-Const} according to the types 
assigned in~\cref{fig:types-constants}; they have no 
priority effect ($\prbot$), do not yield any sequence of priorities
($\priority{\bar{\pi}}=\emptyenv$), and assume that the typing context 
$\envtvars$ has priority $\prtop$, i.e., it is 
\emph{unrestricted}. Rule \textsf{T-Var}  requires 
that variables are typed under an unrestricted environment, 
with no effects. 

We have three forms of abstraction: on types, priorities, and terms.
Rule \textsf{T-TAbs} states that a type 
abstraction $\etabs{\gamma}{\e}$ has type $\tpolyt{\gamma}{\pr}{T}{}$ with no effect if $\e$ has type $T$ under the extended 
polymorphic environment $\envpoly, \gamma :: \pr$, assuming 
that the free priority variables occurring in $\pr$ are 
defined in $\envpr$ and that the typing environment $\envtvars$ is 
unrestricted---because we do not want to abstract over linear values. 
Rule \textsf{T-PAbs} is similar with the expected adaptation to priority abstractions.
Rules \textsf{T-AbsLin} and \textsf{T-AbsUn} consider linear and unrestricted functions, respectively; in the latter case, the context $\envtvars$ should be unrestricted. 
In terms of priorities, both rules
state that the 
lower priority bound coincides with the lowest 
priority in the typing context $\envtvars$ and the upper
priority bound is the effect of evaluating the body of 
the function. 

Rule \textsf{T-TApp} types the application $\etapp {e}{T}$ requiring that the type variable $\tvar$ is associated to a priority  
$\pr \in \nuLowestpriorityEnv{T}{}$, interpreted under environment $\envpr$. 
Rule \textsf{T-PApp}, distinctive of our work, is applicable to functional and session types
(recalling that the ins\-tan\-tia\-tion of priorities in sessions is only a runtime value),
and states that a priority $\pr$ can instantiate 
$\prvar$ if $\pr$ belongs to the priority interval assigned to 
$\prvar$;  all the occurrences of $\prvar$ in $U$
are updated with the new value $\pr$. 
Rule \textsf{T-Inst} governs channel instantiation: it takes the first priority $\pr$ in the 
  sequence for $x$ ($\pr=\fst{\envmap(x)}$) and instantiates the 
continuation type (similarly to rule \textsf{T-PApp}).  

Rule \textsf{T-App} is standard 
for \emph{call-by-value} evaluation strategies~\cite
{padovani:hal-00954364}: it requires the first expression to have a 
\emph{more urgent} priority than the typing context for the second 
expression and the effect of the second expression to be \emph{more 
urgent} than the priority of the type of the first expression (which, by then, has reduced to a value). 
The typing environment is split to distribute channels across different parts of the program, 
but the priority map is passed to both premises to ensure that the 
priority sequence is available for priority instantiation.
Unlike $\envtvars$, the priority map $\envmap$ is not a linear resource: it is only 
consulted by $\mathsf{inst}^\envmap$ (and updated by `popping' in sequencing rules). 

Rule \textsf{T-Pair} types $\epair{e_1}{e_2}$ by splitting the typing environments and the priority map; the conditions on $\priority{\rho_1}$ and $ \priority{\rho_2}$ ensure that $e_1$ is evaluated before $e_2$.
Rule \textsf{T-Let} uses the full map $\envmap$ in the first premise to enable priority
instantiation (e.g., reasoning about $\mathsf{next}^k$), whereas the second premise is 
checked under an updated map $\envmap_2$ that reflects the priority information remaining 
after evaluating $\e_1$ (as captured by $\envmap_1$ in the split $\envmap=\envmap_1\circ\envmap_2$).
This rule further requires that the channels used for evaluating 
$\e_1$ have lower priority than those occurring in $\e_2$. 
This constraint enforces an increasing order of priorities between 
successive operations, thus complementing Rule \textsf{F-Seq} (\Cref{fig:ap:type-formation}), as observed in~\cref{sec:types}.

Rule \textsf{T-LetPair} checks that the effect of $e_1$ 
is lower than the lowest priority of the typing environment for $\e_2$.
This rule is also responsible for updating the map of priorities when the 
priority sequence emerging from the evaluation of $e_1$ is not empty 
(that is, when the rule preceding \textsf{T-LetPair} is \textsf{T-NewPoly}, see below).
More specifically, when the priority sequence $\priority{\bar{\pi}}$ is not empty, 
$e_2$ is typed under the context with a priority extended with 
a mapping of the new channel endpoints ($x_1$ and $x_2$) to the corresponding priority sequence. 
\smallskip

Rule 
\textsf{T-Sel} assigns no effect ($\prbot$) to the evaluation of 
$\eselect{k}$, requires the priority of the choice operation to be 
lower than that of the continuation types and also that the typing 
environment has no channel left, i.e., has priority $\prtop$. Rule 
\textsf{T-Match} treats $\e_\ell$ as if it were a function (see \textsf
{T-AbsLin}) and, according to the call-by-value evaluation strategy, 
requires the evaluation of $\e$ to be more urgent than that of 
$\e_\ell$; the distribution of channel names by the priority maps follows 
the usual rules for splitting linear contexts.

Typing rules for channel creation distinguish session types with a priority 
binder---denoted by $\EuScript{S}^\fforall$---from session types without priority binders---denoted by $\EuScript{S}$. Channels of the first 
kind must be assigned a  priority sequence. 
More specifically, in rule \textsf{T-NewPoly}, the primitive $\enewk$ receives a session type with a priority 
abstraction $\EuScript{S}^\fforall$ 
that indicates the type of the 
channel being created and 
{two natural numbers ($\priority{n_1}$ and $\priority{n_2}$) representing 
the start and increment of the associated priority sequence. 
A priority sequence starting at $\priority{n_1}$ and with increment $\priority{n_2}$
is defined in this rule and passed as a side 
effect through the first annotation $\priority{\bar{\pi}}$ (that should be eventually mapped to the 
corresponding endpoint names and stored in $\envmap$, cf. rule \textsf{T-LetPair}).
rule \textsf{T-New} requires the typing environment to be un\-re\-strict\-ed
and type $\EuScript{S}$ to be well-formed and such that $\unr{\EuScript{S}} \neq \tskip$, to 
ensure that there are some pending operations. 

Rule \textsf{T-Eq} concerns type equivalence: it 
states that an 
expression that is well typed against a type is well typed against any 
of its equivalent types.  
The judgement of type equivalence used in rule \textsf{T-Eq} is an 
extension of previous proposals (cf.~\cite{DBLP:journals/iandc/AlmeidaMTV22,almeida2020deciding,costa2022higher,DBLP:conf/icfp/ThiemannV16}) by considering only the comparison 
of types and type constructors with exactly the same priority.

\paragraph{Revisiting the Motivating Example.}

To illustrate the typing rules, we show how they are used to 
type-check the example presented in~\cref{subsec:deadlocks-simple-rec}. The process defined 
in~\cref{lst:main} reduces to the following configuration:
\[
\confnu{w_1}{r_1}{\confnu{w_2}{r_2}{(\confpar{\confthread{\mainthread}{\eapp{\eapp{\eprapp{\mathsf{client''}}{\enext\,w_2}}{w_2}}{r_1}}}{\confthread{\childthread}{\eapp{\eapp{\eprapp{\mathsf{client'}}{\enext\,w_1}}{w_1}}{r_2}}})}{\bar{\pi_2}}}{\bar{\pi_1}}
\]

The typing derivation for this configuration begins as follows:\\

{\small
\begin{prooftree}
        \hypo{\ldots}
       \Infer1[(C-Main)]{\runtimetyping{\envtvars_M}{\confthread{\mainthread}{\ldots}}{\mainthread}{\envmap_M}}
      \hypo{\text{(proceeds with T-App, T-App, T-PApp)}}
      \Infer1{\typingpoly{\emptyenv}{\envtvars_C}{\emptyenv}{\envmap_C}{\eapp{\eapp{\eprapp{\mathsf{client'}}{\enext\,w_1}}{w_1}}{r_2}}{\tunit}{\emptyenv}{\prbot}{\oracons}}
       \Infer1[]{\runtimetyping{\envtvars_C}{\confthread{\childthread}{\eapp{\eapp{\eprapp{\mathsf{client'}}{\enext\,w_1}}{w_1}}{r_2}}}{\childthread}{\envmap_C}}
      \Infer2[(C-Par)]{\runtimetyping{\envtvars}{(\confpar{\confthread{\mainthread}{\eapp{\eapp{\eprapp{\mathsf{client''}}{\enext\,w_2}}{w_2}}{r_1}}}{\confthread{\childthread}{\eapp{\eapp{\eprapp{\mathsf{client'}}{\enext\,w_1}}{w_1}}{r_2}}})}{\mainthread}{\envmap_1}}
      \Infer1[(C-New)]{\runtimetyping{w_1:\mathsf{Stream}, r_1:\dualof{\mathsf{Stream}}}{\confnu{w_2}{r_2}{(\confpar{\confthread{\mainthread}{\ldots}}{\confthread{\childthread}{\ldots}})}{\bar{\pi_2}}}{\mainthread}{\envmap_0}}
      \Infer1[(C-New)]{\runtimetyping{\emptyenv}{\confnu{w_1}{r_1}{\confnu{w_2}{r_2}{(\confpar{\confthread{\mainthread}{\eapp{\eapp{\eprapp{\mathsf{client''}}{\enext\,w_2}}{w_2}}{r_1}}}{\confthread{\childthread}{\eapp{\eapp{\eprapp{\mathsf{client'}}{\enext\,w_1}}{w_1}}{r_2}}})}{\bar{\pi_2}}}{\bar{\pi_1}}}{\mainthread}{\emptyenv}}
\end{prooftree}}\\\\

where:
{\small
\begin{align*} 
\envtvars &= \{w_1:\mathsf{Stream}, r_1:\dualof{\mathsf{Stream}}, w_2:\mathsf{Stream}, r_2:\dualof{\mathsf{Stream}}\},\\
\envtvars_C &= \{w_1:\mathsf{Stream}, r_2:\dualof{\mathsf{Stream}}\},\\
\envtvars_M &= \{r_1:\dualof{\mathsf{Stream}}, w_2:\mathsf{Stream}\},\\
\envmap_0 & = \{w_1:\mathsf{Stream} \mapsto\bar{\pi_1}, r_1:\dualof{\mathsf{Stream}}\mapsto\bar{\pi_1}\},\\ 
\envmap_1 &= \{w_1:\mathsf{Stream} \mapsto\bar{\pi_1}, r_1:\dualof{\mathsf{Stream}}\mapsto\bar{\pi_2}, w_2:\mathsf{Stream} \mapsto\bar{\pi_2}, r_2:\dualof{\mathsf{Stream}}\mapsto\bar{\pi_2}\}\\
\envmap_C &= \{w_1:\mathsf{Stream} \mapsto\bar{\pi_1}, r_2:\dualof{\mathsf{Stream}}\mapsto\bar{\pi_2}\}\\
\envmap_M &= \{r_1:\dualof{\mathsf{Stream}}\mapsto\bar{\pi_2}, w_2:\mathsf{Stream} \mapsto\bar{\pi_2}\}
\end{align*} }

For illustration, consider the function $\mathsf{client}$ from~\cref{lst:stream}. 
We now sketch a typing derivation showing that $\mathsf{client}'$ has the type declared 
in Lines~3–4 of~\cref{lst:stream}.
Since T-App 
does not split the priority map, the priority environment used to type-check 
$\mathsf{client}$ is $\envmap_C$. 

In the following sketch we $\alpha$-rename the formal parameters of $\mathsf{client}'$ 
so that they are written $w_1$ and $r_2$, matching the concrete endpoints occurring in the configuration. 
This is without loss of generality: 
programs and configurations are considered up to consistent $\alpha$-renaming of bound variables, 
applying the same renaming to $\envtvars$ and $\envmap$. 
Note that $\envmap$ is a map consulted by $\mathsf{inst}^{\envmap}$ 
that is consumed and updated as priorities advance.

To check that 
\[\typingpoly{\emptyenv}{\emptyenv}{\emptyenv}{\envmap_C}{\mathsf{client}}{\tpolyp{S}{\prvar}{(\prbot,\prtop)}{\tarrow{\tarrow{\mathsf{Stream}}{\dualof{\mathsf{Stream}}}{\prtop}{\prbot}{\unmult}}{\tunit}{\prvar}{\prtop}{\linmult}}}{\emptyenv}{\prbot}{\oracons}\]
we start by inverting T-PAbs followed by inversion of T-AbsUn/Lin and get:
\[
\typingpoly{\emptyenv}{\envtvars}{\envpr}{\envmap_C}{\elet{w_1}{\eapp{\eapp{\esend}{\eunit}}{(\einst{w_1}{})}}{\ldots}}{\tunit}{\emptyenv}{\pr}{\oracons}  
\]
where: $\envtvars = \{w_1:\mathsf{Stream}, r_2:\dualof{\mathsf{Stream}}\}$ and $\envpr = \{\prvar\in (\prbot,\prtop)\}$.
For brevity we type-check the unfolded body corresponding to $\mathsf{client}$ 
applied to $w_1$ and $r_2$ (i.e., the $\beta$-reduct of the two applications).
Making the type and priority instantiations explicitly, we actually need to prove that:
{\small
\[
  \typingpoly{\emptyenv}{\envtvars}{\envpr}{\envmap_C}{\elet{w_1}
  {\eapp{\eapp{\etapp{\eprapp{\eprapp{\etapp{\eprapp{\esend}{\prtop}}{\tunit}}{\eapp{\enext}{w_1}}}{\eapp{\enext^2}{w_1}}}{\mathsf{Stream}}}{\eunit}}{(\einst{w_1}{})}}
  {\ldots}}{\tunit}{\emptyenv}{\pr}{\oracons}  
\]}

Assuming the type of $\esend$ from~\cref{fig:types-constants}, inversion of T-Let yields~\eqref{eq:first-let1} and~\eqref{eq:first-let2}:
\begin{align}
  \label{eq:first-let1}
  &\typingpoly{\emptyenv}{w_1:\mathsf{Stream}}{\envpr}{\envmap_C}{
  {\eapp{\eapp{\esend\,\ldots}{\eunit}}{(\einst{w_1}{})}}}{\mathsf{Stream}}{\emptyenv}{\prbot}{\oracons} 
\\
  \label{eq:first-let2}
  &\typingpoly{\emptyenv}{w_1:\mathsf{Stream}, r_2:\dualof{\mathsf{Stream}}}{\envpr}{\envmap_C'}{\eletpair{a}{r_2}{\eapp{\ereceive}{(\einst{r_2}{})}}{\ldots}}{\tunit}{\emptyenv}{\pr}{\oracons} 
\end{align} 
The priority environment $\envmap_C'$ is obtained by popping the priority stack of $w_1$, i.e.: 
$$\envmap_C' = \{w_1:\mathsf{Stream} \mapsto\tail{\bar{\pi_1}},r_2:\dualof{\mathsf{Stream}}\mapsto\bar{\pi_2}\}.$$

Applying inversion of T-App to~\eqref{eq:first-let1} yields~\eqref{eq:first-let3} and~\eqref{eq:first-let4}.
\begin{align}
  \label{eq:first-let3}
  &\typingpoly{\emptyenv}{\emptyenv}{\envpr}{\envmap_C}{
  {\eapp{\esend\,\ldots}{\eunit}}}{\tarrow{\tseq{\tout{\tunit}{\fst{\bar{\pi_1}}}}{\mathsf{Stream}}}{\mathsf{Stream}}{\prtop}{\eapp{\enext}{w_1}}{\linmult}}{\emptyenv}{\prbot}{\oracons} \\
  \label{eq:first-let4}
  &\typingpoly{\emptyenv}{w_1:\mathsf{Stream}}{\envpr}{\envmap_C}{
  {\einst{w_1}{}}}{\tseq{\tout{\tunit}{\fst{\bar{\pi_1}}}}{\mathsf{Stream}}}{\emptyenv}{\prbot}{\oracons} 
\end{align}
Inversion of T-Inst followed by T-Var completes the derivation of~\eqref{eq:first-let4}, provided $\fst{\bar{\pi}}\in (\prbot,\prtop)$.

A further inversion of T-App on~\eqref{eq:first-let3} yields~\eqref{eq:first-let5}:
{\small
\begin{equation}
  \label{eq:first-let5}
  \typingpoly{\emptyenv}{\emptyenv}{\envpr}{\envmap_C}{
  {\etapp{\esend\,\ldots}{\mathsf{Stream}}}}{\tarrow{\tunit}{\tarrow{\tseq{\tout{\tunit}{\fst{\bar{\pi_1}}}}{\mathsf{Stream}}}{\mathsf{Stream}}{\prtop}{\eapp{\enext}{w_1}}{\linmult}}{\prtop}{\prbot}{\unmult}}{\emptyenv}{\prbot}{\oracons} 
\end{equation}}
Inversion of T-TApp yields~\eqref{eq:first-let6}, 
under the proviso that $\enext^2 w_1 \in \lowestpriority{\mathsf{Stream}}$.
{\small
\begin{equation}
  \label{eq:first-let6}
  \typingpoly{\emptyenv}{\emptyenv}{\envpr}{\envmap_C}{
  {\eprapp{\esend\,\ldots}{\enext^2 w_1}}}{\tpolyt{\beta}{\enext^2 w_1}{\tarrow{\tunit}{\tarrow{\tseq{\tout{\tunit}{\fst{\bar{\pi_1}}}}{\beta}}{\beta}{\prtop}{\eapp{\enext}{w_1}}{\linmult}}{\prtop}{\prbot}{\unmult}}{}}{\emptyenv}{\prbot}{\oracons} 
\end{equation}}
This requirement motivates two observations: (i) the priority of a polymorphic type must be a set rather than a single value; and (ii) the priority map must be preserved from the conclusion to the first premise of T-Let.

Two applications of T-PApp then reconnect the derivation with the type of $\esend$ 
in~\cref{fig:types-constants}, under the side conditions $\enext^2 w_1 \in (\prbot,\prtop)$ 
and $\enext w_1 \in (\prbot,\prtop)$.
This branch concludes by inverting T-TApp, T-PApp, and finally T-Const.

The derivation of~\eqref{eq:first-let2} follows analogously. By T-Let and T-LetPair, the priority map is 
popped, so the subsequent recursive call proceeds with the next elements of the priority sequence.

\subsection{Dynamics: Semantics, Type Preservation, and Deadlock Freedom} 
\label{ap:safety}

We define reduction
for  expressions and 
configurations  considering a call-by-value evaluation 
strategy; see \cref{fig:ap:reduction}.
Reduction rules for expressions are standard for a call-by-value setting; novelties are a new reduction rule for 
channel instantiations (\textsf{E-PInst}) and a
$\beta$-reduction rule for 
priority-level abstractions (\textsf{E-PApp}). 
Evaluation 
contexts also reflect the left-to-right evaluation strategy and 
are defined by the grammar at the top of~\cref{fig:reduction}. 

Configuration reduction includes
rules to lift reductions from expressions to configurations 
(Rule \textsf{R-LiftE}) and from configurations to configuration contexts
(Rule \textsf{R-LiftC}). 
Rule \textsf{R-Fork} creates a new \emph{child} thread, while rules 
\textsf{R-New} and \textsf{R-NewPoly} create a new channel and bind its 
endpoints in 
a restriction. Rules \textsf{R-Com}, \textsf{R-Ch} and \textsf{R-Close} 
match, respectively, send and receive, select and match, 
close and wait operations in parallel threads. Configuration reduction 
also includes the usual rules for parallel execution 
(\textsf{R-Par}), reduction under restriction (\textsf{R-Bind}) 
and reduction up to structural congruence (\textsf{R-Cong}).

\begin{figure}
  \declrel{Evaluation contexts}
  \begin{align*}
    \evalctx \grmeq & \evalhole{} \grmor \eapp{\evalctx}{e} \grmor \eapp{v}{\evalctx} 
    \grmor \etapp{\evalctx}{T} \grmor \eprapp{\evalctx}{\pr}
    \grmor \elet{\evar}{\evalctx}{e} \\
    &\grmor \eletpair{\evar_1}{\evar_2}{\evalctx}{e}
    \grmor \eseq{\evalctx}{e}
      \grmor \epair{\evalctx}{e} 
    \grmor \epair{v}{\evalctx} \\
    &
    \grmor \ematch{\evalctx}{\ell}{e_\ell}{\ell\in L}
    \grmor \einst{E}{\envmap}
    \tag{Evaluation contexts}\\
    \threadctx \grmeq& \confthread\thread\evalctx \tag{Thread evaluation contexts}\\
    \configctx \grmeq & \evalhole{} \grmor \confpar{\configctx}{\config{C}} \grmor \confnu{x}{y}{\configctx}{\bar{\pi}}
    \tag{Configuration contexts}
  \end{align*}
  \declrel{Expression reduction}
  \begin{gather*}
    \infer[\textsf{E-App}]
      {}
      {\red{\eapp{(\eabs{x}{\_}{e}{m})}{v}}{\subs{e}{v}{x}}}
    \quad
    \infer[\textsf{E-LetElim}]
      {}
      {\red{\elet{\evar}{v}{e}}{\subs{e}{v}{\evar}}}
    \quad 
    \infer[\textsf{E-PairElim}]
      {}
      {\red{\eletpair{\evar_1}{\evar_2}{\epair{v_1}{v_2}}{e}}{\subs{\subs{e}{v_1}{\evar_1}}{v_2}{\evar_2}}}
    \\ 
    \infer[\textsf{E-UnitElim}]
      {}
      {\red{\eseq{\eunit}{e}}{e}}
    \quad
    \infer[\textsf{E-PInst}]
      {}
      {\red{\einst{x}{\envmap}}{\eprapp{x}{\priority{\fst{\envmap(x)}}}}}
    \qquad
    \infer[\textsf{E-TApp}]
      {}
      {\red{\etapp{(\etabs{\tvar}{v})}{T}}{\subs{v}{T}{\tvar}}}
    \qquad
    \infer[\textsf{E-PApp}]
      {}
      {\red{\eprapp{(\eprabs{\prvar}{v})}{\pr}}{\subs{v}{\pr}{\prvar}}}
    \\
    \infer[\textsf{E-Ctx}]
      {\red{\e_1}{\e_2}}
      {\red{\evalapp{\evalctx}{\e_1}}{\evalapp{\evalctx}{\e_2}}}
  \end{gather*}
  \declrel{Configuration reduction}
  \begin{gather*}
    \infer[\textsf{R-LiftE}]
      {\red{\e_1}{\e_2}}
      {\red{\evalapp{\threadctx}{\e_1}}{\evalapp{\threadctx}{\e_2}}}
    \qquad
    \infer[\textsf{R-LiftC}]
      {\red{\config{C}_1}{\config{C}_2}}
      {\red{\evalapp{\configctx}{\config{C}_1}}{\evalapp{\configctx}{\config{C}_2}}}
    \quad
    \infer[\textsf{R-Fork}]
      {}
      {\red{\evalapp{\threadctx}{\eapp{\eprapp{\eprapp{\efork}{\_}}{\_}}{e}}}
           {\confpar{\evalapp{\threadctx}{\eunit}}{\confthread{\childthread}{\eapp{\e}{\eunit}}}}}
    \\
    \infer[\textsf{R-New}]
      {}
      {\red{\evalapp{\threadctx}{\enew{\EuScript{S}}}}{\confnu{x}{y}{\evalapp{\threadctx}{\epair{x}{y}}}{\emptyenv}}
      \text{\enspace $x, y\not \in \fv{\threadctx}$ and $\unr{\EuScript{S}} \neq \tskip$}}
    \\
    \infer[\textsf{R-NewPoly}]
      {}
      {\red{\evalapp{\threadctx}{\enewpoly{\EuScript{S}^\fforall}{n_1}{n_2}}}{\confnu{x}{y}{\evalapp{\threadctx}{\epair{x}{y}}}{\bar{\pi}}}
      \text{\enspace where $x, y\not \in \fv{\threadctx}$ and $\priority{\bar{\pi}} = [\priority{n_1} + k \cdot \priority{n_2} \mid k\in\mathbb{N}]$}}
    \\
    \infer[\textsf{R-Com}]
      {}
      {
      {\confnu{x}{y}{(\confpar{\evalapp{\threadctx}{\eapp{\eapp{\etapp{\eprapp{\eprapp{\etapp{\eprapp{\esend}{\_}}{\_}}{\_}}{\_}}{\_}}{v}}{x}}}
                             {\evalapp{\threadctx'}{\eapp{\etapp{\eprapp{\eprapp{\etapp{\eprapp{\ereceive}{\_}}{\_}}{\_}}{\_}}{\_}}{y}}})}{\bar{\pi}}}
      \rightarrow
      {\confnu{x}{y}{(\confpar{\evalapp{\threadctx}{x}}{\evalapp{\threadctx'}{\epair{v}{y}}})}{\bar{\pi}}}
      }
    \\
    \infer[\textsf{R-Ch}]
      {k \in L}
      {
      {\confnu{x}{y}{(\confpar{\evalapp{\threadctx}{\eapp{\eselect{k}}{x}}}
                             {\evalapp{\threadctx'}{\ematch{y}{\ell}{e_\ell}{\ell\in L}}})}{\bar{\pi}}}
      \rightarrow
      {\confnu{x}{y}{(\confpar{\evalapp{\threadctx}{x}}{\evalapp{\threadctx'}{\eapp{e_k}{y}}})}{\bar{\pi}}}
      }
    \\
    \infer[\textsf{R-Close}]
      {}
      {
      {\confnu{x}{y}{(\confpar{\evalapp{\threadctx}{\eapp{\eprapp{\eclose}{\_}}{y}}}
                             {\evalapp{\threadctx'}{\eapp{\eprapp{\ewait}{\_}}{x}}})}{\bar{\pi}}}
      \rightarrow
      {\confnu{x}{y}{(\confpar{\evalapp{\threadctx}{\eunit}}{\evalapp{\threadctx'}{\eunit}})}{\bar{\pi}}}
      }
    \\
    \infer[\textsf{R-Par}]
      {\red{\config{C}_1}{\config{C}_2}}
      {\red
      {\confpar{\config{C}_1}{\config{D}}}
      {\confpar{\config{C}_2}{\config{D}}}
      }
    \qquad
    \infer[\textsf{R-Bind}]
      {\red{\config{C}_1}{\config{C}_2}}
      {\red
      {\confnu{x}{y}{\config{C}_1}{\bar{\pi}}}
      {\confnu{x}{y}{\config{C}_2}{\bar{\pi}}}
      }
    \quad
    \infer[\textsf{R-Cong}]
        {\config{C}_1\structcong \config{C}_2
        \quad \red{\config{C}_1}{\config{D}_1}
        \quad \config{D}_1\structcong \config{D}_2}
        {\red {\config{C}_2}{\config{D}_2}}
  \end{gather*}
  \declrel{Structural congruence}
  \begin{align*} 
    \config{C} \structcong\,&  \confpar{\confthread{\childthread}{\eunit}}{\config{C}}&
    \confnu{x}{y}{\confthread{\childthread}{\eunit}}{\bar{\pi}}\structcong\, &\confthread{\childthread}{\eunit}  
    \\
    \confpar{\config{C}}{\config{D}} \structcong\,& \confpar{\config{D}}{\config{C}} &
    \confnu{x}{y}{\config{C}}{\bar{\pi}} \structcong\, & \confnu{y}{x}{\config{C}}{\bar{\pi}}\\
    \confpar{(\confpar{\config{C}}{\config{D}})}{\config{E}} \structcong\,& \confpar{\config{C}}{(\confpar{\config{D}}{\config{E}})} &
    \confnu{w}{x}{\confnu{y}{z}{\config{C}}{\bar{\pi}_1}}{\bar{\pi}_2}\structcong\,&{\confnu{y}{z}{\confnu{w}{x}{\config{C}}{\bar{\pi}_2}}{\bar{\pi_1}}}\\
    \confpar{(\confnu{x}{y}{\config{C}}{\bar{\pi}})}{\config{D}} \structcong\,& \confnu{x}{y}{(\confpar{\config{C}}{\config{D}})}{\bar{\pi}}
  \end{align*}
  \caption{Evaluation contexts, expression and configuration reduction and structural congruence (complete).
  The main text includes an excerpt in~\cref{fig:reduction}.}
  \label{fig:ap:reduction}
\end{figure}

We start by proving that structural congruence ($\structcong$, cf. \cref{fig:ap:reduction}, bottom) is type preserving (\Cref{t:cong}).
To this end, we require the following auxiliary results. 

\begin{lemma}[Strengthening]\
  \begin{enumerate}
    \item If $\istype{\envpoly, \tvar :: \pr}{\envpr}{T}$ and $\tvar\not\in\ftv{T}$ then $\istype{\envpoly}{\envpr}{T}$.
    \item If $\typingpoly{\envpoly}{\envtvars, x:T}{\envpr}{\envmap}{e}{U}{\bar{\pi}}{\pr}{\oracons}$
    and $x \not\in \fv{e}$, then \\ $\nuLowestpriorityEnv{T}{} = \prtop$
    and $\typingpoly{\envpoly}{\envtvars}{\envpr}{\envmap}{e}{U}{\bar{\pi}}{\pr}{\oracons}$.
    \item If $\runtimetyping{\envtvars, x: T}{\config{C}}{\thread}{\Psi}$ and $x \not\in \fv{\config{C}}$
    then $\nuLowestpriorityEnv{T}{} = \prtop$ and $\runtimetyping{\envtvars}{\config{C}}{\thread}{\Psi}$.
  \end{enumerate}  
\end{lemma}
\begin{proof}
  By rule induction on the first hypotheses.\qed
\end{proof}

\begin{lemma}
  If $\istype{\envpoly}{\envpr}{S}$ then $\istype{\envpoly}{\envpr}{\dualof{S}}$.
\end{lemma}
\begin{proof}
  By induction on the derivation $\istype{\envpoly}{\envpr}{S}$.\qed
\end{proof}

\begin{lemma}[Agreement for process formation]\
  \begin{enumerate}
    \item If $\typingpoly{\envpoly}{\envtvars}{\envpr}{\envmap}{e}{T}{\bar{\pi}}{\pr}{\oracons}$
    then $\wfpolyctx{\envpoly}{\envpr}{\envtvars}$ and $\istype{\envpoly}{\envpr}{T}$.
    \item If $\runtimetyping{\envtvars}{\config{C}}{\thread}{\Psi}$ then $\wfpolyctx{\emptyenv}{\emptyenv}{\envtvars}$.
  \end{enumerate}
  \begin{proof}\
    (1) By induction on the derivation of $\typingpoly{\envpoly}{\envtvars}{\envpr}{\envmap}{e}{T}{\bar{\pi}}{\pr}{\oracons}$.
    (2) By induction on the derivation of $\runtimetyping{\envtvars}{\config{C}}{\thread}{\Psi}$ using (1).\qed
  \end{proof}
\end{lemma}

\begin{lemma}[Weakening]\
  \label{lem:weakening}
  \begin{enumerate}
    \item If $\istype{\envpoly}{\envpr}{T}$ then $\istype{\envpoly, \tvar :: \pr}{\envpr}{T}$.
    \item If $\istype{\envpoly}{\envpr}{T}$ then $\istype{\envpoly}{\envpr, \prvar\in \intervalvar}{T}$.
    \item If $\wfpolyctx{\envpoly}{\envpr}{\envtvars}$ then $\wfpolyctx{\envpoly, \tvar :: \pr}{\envpr}{\envtvars}$.
    \item If $\wfpolyctx{\envpoly}{\envpr}{\envtvars}$ then $\wfpolyctx{\envpoly}{\envpr, \prvar \in \intervalvar}{\envtvars}$.
    \item If $\typingpoly{\envpoly}{\envtvars}{\envpr}{\envmap}{e}{T}{\bar{\pi}}{\pr}{\oracons}$ then 
    $\typingpoly{\envpoly, \tvar :: \pr}{\envtvars}{\envpr}{\envmap}{e}{T}{\bar{\pi}}{\pr}{\oracons}$.
    \item If $\typingpoly{\envpoly}{\envtvars}{\envpr}{\envmap}{e}{T}{\bar{\pi}}{\pr}{\oracons}$ and 
    $\nuLowestpriorityEnv{U}{} = \prtop$ then 
    $\typingpoly{\envpoly}{\envtvars, x:U}{\envpr}{\envmap}{e}{T}{\bar{\pi}}{\pr}{\oracons}$.
    \item If $\typingpoly{\envpoly}{\envtvars}{\envpr}{\envmap}{e}{T}{\bar{\pi}}{\pr}{\oracons}$ then
    $\typingpoly{\envpoly}{\envtvars}{\envpr, \prvar\in \intervalvar}{\envmap}{e}{T}{\bar{\pi}}{\pr}{\oracons}$.
    \item If $\runtimetyping{\envtvars}{\config{C}}{\thread}{\Psi}$ and $\istype{\envpoly}{\envpr}{T}$
    and $\nuLowestpriorityEnv{T}{} = \prtop$ then $\runtimetyping{\envtvars, x:T}{\config{C}}{\thread}{\Psi}$.
  \end{enumerate}
\end{lemma}
\begin{proof}
  By induction on the first hypothesis, for each item.\qed
\end{proof}

We then have:

\begin{theorem}\label{t:cong}
  If $\config{C} \structcong \config{C}'$ and $\runtimetyping{\envtvars}{\config{C}}{\thread}{\Psi}$
  then $\runtimetyping{\envtvars}{\config{C}'}{\thread}{\Psi}$.
\end{theorem}
\begin{proof}
  The proof follows by rule induction on the derivation of $\config{C} \structcong \config{C}'$, using the above lemmas.\qed
\end{proof}

\paragraph{Type Preservation.}
We start by 
showing that values produce no effect and the priority of the context coincides with the priority of their type.

\begin{lemma}
  \label{lem:values}
  If $\typingpoly{\envpoly}{\envtvars}{\envpr}{\envmap}{v}{T}{\bar{\pi}}{\pr}{\oracons}$ then
  $\pr = \prbot$ and $\nuLowestpriorityEnv{\envtvars}{} = \nuLowestpriorityEnv{T}{}$.
\end{lemma}
\begin{proof}
   By induction on the derivation of $\typingpoly{\envpoly}{\envtvars}{\envpr}{\envmap}{v}{T}{\bar{\pi}}{\pr}{\oracons}$.\qed
\end{proof}



The usual substitution lemmas now consider also priority substitution.

\begin{lemma}[Value substitution]
  \label{lem:value-subst}
  Assume that:
  \begin{itemize}
    \item $\ctxsplit{\envpoly}{\envpr}{\envtvars}{\envtvars_1}{\envtvars_2}$
    \item $\typingpoly{\envpoly}{\envtvars_2, x: U}{\envpr}{\envmap_2}{e}{T}{\bar{\pi}}{\pr_2}{\oracons}$
    \item There exists a priority map prefix $\envmap_p$ such that $\envmap = \envmap_p \circ \envmap_2$ and for
    every $y\in (\fv{e}\setminus\{x\})\cap \dom{\envmap}$, $\envmap_p(y) =\emptyenv$,
    \item Either (i) $\typingpoly{\envpoly}{\envtvars_1}{\envpr}{\envmap_p}{v}{U}{\emptyenv}{\pr_1}{\oracons}$
    or (ii) $\typingpoly{\envpoly}{\envtvars_1}{\envpr}{\envmap}{v}{U}{\emptyenv}{\pr_1}{\oracons}$.
  \end{itemize} 
    Then, $\typingpoly{\envpoly}{\envtvars}{\envpr}{\envmap}{\subs{e}{v}{x}}{T}{\bar{\pi}}{\pr_2}{\oracons}$.
   
\end{lemma}
\begin{proof}
    By induction on the derivation of $\typingpoly{\envpoly}{\envtvars_1, x: U}{\envpr}{\envmap_1}{e}{T}{\bar{\pi}}{\pr_1}{\oracons}$,
    using~\cref{lem:values}. We sketch the proof for case T-Inst to illustrate why the empty prefix map is needed.

    \textbf{Case T-Inst.}
    Assume that $\typingpoly{\envpoly}{\envtvars_2, x: U}{\envpr}{\envmap_2}{e}{T}{\bar{\pi}}{\pr_2}{\oracons}$
    and exists $\envmap_p$ in the stated conditions. Inversion of T-Inst yields: $e = \einst{y}{}$ with $y \neq x$ 
    and 
    $\typingpoly{\envpoly}{\envtvars_2, x: U}{\envpr}{\envmap_2}{y}{ \tpolyp{S}{\prvar}{\intervalvar}{S}}{\emptyenv}{\pr_2}{\oracons}$
     for $T = \subs{S}{\pr}{\intervalvar}$ and $\pr = \fst{\envmap_2(y)}\in \intervalvar$. We have $\subs{(\einst{y}{})}{v}{x}= \einst{y}{}$.
    The premise of T-Inst holds because the derivation of variables does not depend on the content of $\envmap$ beyond well-formedness.
    Let $\priority{\pr'} = \fst{\envmap(y)}$. Because $\envmap = \envmap_p \circ \envmap_2$, we have 
    $\envmap(y) = \envmap_p(y) \cdot \envmap_2(y)$. Since $\envmap_p(y) = \emptyenv$, $\envmap(y) = \envmap_2(y)\in \intervalvar$.
    We conclude applying rule T-Inst.
    \qed
 \end{proof}

 \begin{lemma}[Type substitution]\
  \label{lem:type-subst}
  \begin{enumerate}
    \item Suppose that $\istype{\envpoly}{\envpr}{U}$. If 
    $\istype{\envpoly, \tvar::\pr}{\envpr}{T}$ 
    and $\nuLowestpriorityEnv{U}{} = \pr$
    then 
    $\istype{\envpoly}{\envpr}{\subs{T}{U}{\tvar}}$.
    \item If $\ctxsplit{\envpoly, \tvar::\pr}{\envpr}{\envtvars}{\envtvars_1}{\envtvars_2}$
    and $\istype{\envpoly}{\envpr}{U}$
    and $\nuLowestpriorityEnv{U}{} = \pr$
    then \\ $\ctxsplit{\envpoly}{\envpr}{\subs{\envtvars}{U}{\tvar}}{\subs{\envtvars_1}{U}{\tvar}}{\subs{\envtvars_2}{U}{\tvar}}$
    \item If $\typingpoly{\envpoly, \tvar :: \pr}{\envtvars}{\envpr}{\envmap}{e}{T}{\emptyenv}{\pr_1}{\oracons}$
    and $\istype{\envpoly}{\envpr}{U}$
    and $\nuLowestpriorityEnv{U}{} = \pr$
    then \\ $\typingpoly{\envpoly}{\subs{\envtvars}{U}{\tvar}}{\envpr}{\envmap}{\etapp{e}{U}}{\subs{T}{U}{\tvar}}{\emptyenv}{\pr_1}{\oracons}$.
  \end{enumerate}
 \end{lemma}
 \begin{proof}
    (1) By induction on the derivation of $\istype{\envpoly, \tvar::\pr}{\envpr}{T}$.
    (2) By induction on the derivation of $\ctxsplit{\envpoly, \tvar::\pr}{\envpr}{\envtvars}{\envtvars_1}{\envtvars_2}$.
    (3) By induction on the derivation of $\typingpoly{\envpoly, \tvar :: \pr}{\envtvars}{\envpr}{\envmap}{e}{T}{\emptyenv}{\pr_1}{\oracons}$, 
    using items 1~and~2.\qed
 \end{proof}




\begin{lemma}[Priority substitution]\
  \label{lem:priority-subst}
  \begin{enumerate}
    \item If $\istype{\envpoly}{\envpr, \prvar\in\intervalvar}{T}$
    and $\pr \in \intervalvar$
    then $\istype{\subs{\envpoly}{\pr}{\prvar}}{\subs{\envpr}{\pr}{\prvar}}{\subs{T}{\pr}{\prvar}}$.
    \item If $\ctxsplit{\envpoly}{\envpr, \prvar\in\intervalvar}{\envtvars}{\envtvars_1}{\envtvars_2}$
    and $\pr \in \intervalvar$ then
    $\ctxsplit{\subs{\envpoly}{\pr}{\prvar}}{\subs{\envpr}{\pr}{\prvar}}{\subs{\envtvars}{\pr}{\prvar}}{\subs{\envtvars_1}{\pr}{\prvar}}{\subs{\envtvars_2}{\pr}{\prvar}}$.
    \item\label{it:3-lemma9} If $\typingpoly{\envpoly}{\envtvars}{\envpr, \prvar\in \intervalvar}{\envmap}{e}{T}{\emptyenv}{\pr_1}{\oracons}$
    and $\pr\in \intervalvar$
    then \\ $\typingpoly{\subs{\envpoly}{\pr}{\prvar}}{\subs{\envtvars}{\pr}{\prvar}}{\subs{\envpr}{\pr}{\prvar}}{\envmap}{\subs{\e}{\pr}{\prvar}}{\subs{T}{\pr}{\prvar}}{\emptyenv}{\subs{\pr_1}{\pr}{\prvar}}{\oracons}$.
  \end{enumerate}
\end{lemma}
\begin{proof}
    (1) By induction on the derivation of $\istype{\envpoly}{\envpr, \prvar\in\intervalvar}{T}$.
    (2) By induction on the derivation of $\ctxsplit{\envpoly}{\envpr, \prvar\in\intervalvar}{\envtvars}{\envtvars_1}{\envtvars_2}$.
    (3) By induction on the derivation of $\typingpoly{\envpoly}{\envtvars}{\envpr, \prvar\in \intervalvar}{\envmap}{e}{T}{\emptyenv}{\pr_1}{\oracons}$
    using items 1 and 2.\qed
\end{proof}




We may now prove that expression reduction preserves typing
(\cref{lem:exp-reduction-type-pres}).

\begin{proof}[\cref{lem:exp-reduction-type-pres}]
    The proof is done by induction on the derivation of $\red{e_1}{e_2}$, using Lemmas~\ref{lem:weakening},
    \ref{lem:values}, 
    \ref{lem:value-subst}, \ref{lem:type-subst} and \ref{lem:priority-subst}.
    We sketch some cases:

    \textbf{Case E-LetElim.} Assume that 
    $\typingpoly{\envpoly}{\envtvars}{\envpr}{\envmap}{\elet{\evar}{v}{e}}{T}{\bar{\pi}}{\pr}{\oracons}$. 
    By T-Let we know that $\envtvars = \envtvars_1 \circ \envtvars_2$ and $\envmap = \envmap_1 \circ \envmap_2$
    and $\pr = \priority{\pr_1} \glb \priority{\pr_2}$ and 
    $\priority{\bar{\pi}} =\priority{\emptyenv}$.
    Premises of T-Let yield \inlineequation[eq:type-pr-eletelim1]{\typingpoly{\envpoly}{\envtvars_1}{\envpr}{\envmap}{v}{U}{\emptyenv}{\pr_1}{\oracons}}
    and \inlineequation[eq:type-pr-eletelim2]{\typingpoly{\envpoly}{\envtvars_2, x:U}{\envpr}{\envmap_2}{e}{T}{\emptyenv}{\pr_2}{\oracons}}
    where $\priority{\rho_1} < \nuLowestpriorityEnv{\envtvars_2}{}$. Using~\cref{lem:values} we get $\priority{\pr_1} = \prbot$, so $\priority{\pr_2} = \pr$.
    Using~\eqref{eq:type-pr-eletelim1},~\eqref{eq:type-pr-eletelim2} and~\cref{lem:value-subst} we get 
    $\typingpoly{\envpoly}{\envtvars}{\envpr}{\envmap}{\subs{e}{v}{\evar}}{T}{\emptyenv}{\pr}{\oracons}$.

    \textbf{Case E-PairElim.} Assume that 
    $\typingpoly{\envpoly}{\envtvars}{\envpr}{\envmap}{\eletpair{\evar_1}{\evar_2}{\epair{v_1}{v_2}}{e}}{T}{\bar{\pi}}{\pr}{\oracons}$. 
    By T-LetPair we know that 
    $\envtvars = \envtvars_1 \circ \envtvars_2$ and $\envmap = \envmap_1 \circ \envmap_2$
    and $\pr = \priority{\pr_1} \glb \priority{\pr_2}$ and 
    $\priority{\bar{\pi}} =\priority{\emptyenv}$.
    The first premise of T-LetPair yields 
    \begin{equation}
      \label{eq:type-pr-eletpairelim1}
      \typingpoly{\envpoly}{\envtvars_1}{\envpr}{\envmap}{\epair{v_1}{v_2}}{\tprod{U_1}{U_2}}{\bar{\pi_0}}{\pr_1}{\oracons}
    \end{equation}
    to which we apply inversion of T-Pair and conclude that $\priority{\bar{\pi_0}} = \priority{\emptyenv}$ 
    and  
    \begin{align}
      \label{eq:type-pr-eletpairelim2}\typingpoly{\envpoly}{\envtvars'}{\envpr}{\envmap_1'}{v_1}{U_1}{\emptyenv}{\pr_1'}{\oracons}\\
      \label{eq:type-pr-eletpairelim3}\typingpoly{\envpoly}{\envtvars''}{\envpr}{\envmap_1''}{v_2}{U_2}{\emptyenv}{\pr_1''}{\oracons}
    \end{align}
    for $\envtvars_1 = \envtvars_1'\circ \envtvars_1''$ and $\envmap = \envmap' \circ \envmap''$ and 
    $\priority{\pr_1} = \priority{\pr_1'}\glb \priority{\pr_1''}$. Using ~\cref{lem:values} we get $\priority{\pr_1'} =\priority{\pr_1''}=\prbot$, 
    so $\priority{\pr_1} = \prbot$. The last premise of T-LetPair yields  
    \begin{equation}
      \label{eq:type-pr-eletpairelim4}
      \typingpoly{\envpoly}{\envtvars_2, x_1:U_1, x_2:U_2}{\envpr}{\envmap_2}{\e}{T}{\emptyenv}{\pr}{\oracons}
    \end{equation}
    Using~\cref{lem:value-subst} twice, first on~\eqref{eq:type-pr-eletpairelim4},~\eqref{eq:type-pr-eletpairelim2}
    and then using~\eqref{eq:type-pr-eletpairelim3}, we conclude that 
    $\typingpoly{\envpoly}{\envtvars}{\envpr}{\envmap}{\subs{\subs{e}{v_1}{\evar_1}}{v_2}{\evar_2}}{T}{\emptyenv}{\pr}{\oracons}$

    \textbf{Case E-PInst.} Assume that 
    $\typingpoly{\envpoly}{\envtvars}{\envpr}{\envmap}{\einst{x}{\envmap}}{T}{\bar{\pi}}{\pr}{\oracons}$.
    By T-Inst we know that $T = \subs{U}{\fst{\envmap(x)}}{\prvar}$ and $\priority{\bar{\pi}} =\priority{\emptyenv}$.
    Inversion of T-Inst yields 
    $\typingpoly{\envpoly}{\envtvars}{\envpr}{\envmap}{x}{\tpolyp{S}{\prvar}{\intervalvar}{U}}{\emptyenv}{\pr}{\oracons}$
    and $\fst{\envmap(x)} \in \intervalvar$. We conclude applying T-PApp.

    \textbf{Case E-PApp.} Assume that 
    $\typingpoly{\envpoly}{\envtvars}{\envpr}{\envmap}{\eprapp{(\eprabs{\prvar}{v})}{\priority{\pr'}}}{T}{\bar{\pi}}{\pr}{\oracons}$.
    By T-PApp we know that $T = \subs{U}{\pr'}{\prvar}$ and 
    $\priority{\bar{\pi}} =\priority{\emptyenv}$.
    Inversion of T-PApp yields 
    $\typingpoly{\envpoly}{\envtvars}{\envpr}{\envmap}{\eprabs{\prvar}{v}}{\tpolyp{}{\prvar}{\intervalvar}{U}}{\emptyenv}{\pr}{\oracons}$
    and $\priority{\pr'} \in \intervalvar$.
    Inversion of T-PAbs then leads to 
    $\typingpoly{\envpoly}{\envtvars}{\envpr, \prvar\in\intervalvar}{\envmap}{v}{U}{\emptyenv}{\pr}{\oracons}$.
    Applying~\cref{it:3-lemma9} of~\cref{lem:priority-subst} we conclude that 
    $\typingpoly{\envpoly}{\envtvars}{\envpr}{\envmap}{\subs{v}{\priority{\pr'}}{\prvar}}{T}{\emptyenv}{\pr}{\oracons}$.\qed
\end{proof}

Now we can prove that configuration reduction preserves typing.

\begin{proof}[\cref{lem:subject-red-processes}]
    The proof is done by induction on the derivation of $\red{\config{C}_1}{\config{C}_2}$, using~\cref{lem:exp-reduction-type-pres}.
    We sketch the proofs for some cases:
    
    \textbf{Case R-Fork.}
    Assume that $\runtimetyping{\envtvars}{\evalapp{\threadctx}{\eapp{\eprapp{\eprapp{\efork}{\_}}{\_}}{e}}}{\thread}{\Psi}$.
    Depending on $\threadctx$, we invert C-Main or C-Child and get $\typingpoly{\emptyenv}{\envtvars}{\emptyenv}{\Psi}{\eapp{\eprapp{\eprapp{\efork}{\_}}{\_}}{e}}{T}{\bar{\pi}}{\pr}{\oracons}$
    where $\nuLowestpriorityEnv{T}{} = \prtop$. (Note that $T=\tunit$ when inversion of C-Child is used.)
    By T-App we know that $\priority{\bar{\pi}} =\priority{\emptyenv}$.
    Inversion of T-App yields 
    \begin{align}
      \label{eq:type-pr-conf-fork1} \typingpoly{\emptyenv}{\emptyenv}{\emptyenv}{\Psi}{\eprapp{\eprapp{\efork}{\_}}{\_}}{\tarrow{U}{T}{\prtop}{\prbot}{}}{\epsilon}{\pr_1}{\oracons}\\
      \label{eq:type-pr-conf-fork2}\typingpoly{\emptyenv}{\envtvars}{\emptyenv}{\Psi}{e}{U}{\epsilon}{\pr_2}{\oracons}
    \end{align}
    where $\pr = \prglb{\pr_1}{\pr_2}$.
    Inversion of T-Const on~\eqref{eq:type-pr-conf-fork1} yields  $\priority{\pr_1} =\prbot$, so $\priority{\pr_2} =\pr$.
    On the other hand, looking at the type for $\keyword{fork}$ in~\cref{fig:types-constants} we infer that 
    $U = \tarrow{\tunit}{\tunit}{-}{-}{\unmult}$.
    From T-Const and~\cref{fig:types-constants} we know that 
    $\typingpoly{\emptyenv}{\emptyenv}{\emptyenv}{\emptyenv}{\eunit}{\tunit}{\epsilon}{\prbot}{\oracons}$,
    so we can use T-App and~\eqref{eq:type-pr-conf-fork2} to conclude that
    \inlineequation[eq:type-pr-conf-fork3]{\typingpoly{\emptyenv}{\envtvars}{\emptyenv}{\Psi}{\eapp{e}{\eunit}}{\tunit}{\epsilon}{\pr}{\oracons}}.
    We apply C-Child and C-Main (depending on $\threadctx$) and C-Par to conclude.

    \textbf{Case R-NewPoly.}
    Assume that $\runtimetyping{\envtvars}{\evalapp{\threadctx}{\enewpoly{\EuScript{S}^\fforall}{n_1}{n_2}}}{\thread}{\Psi}$.
    Inversion of C-Main yields
    $\typingpoly{\emptyenv}{\envtvars}{\emptyenv}{\Psi}{\enewpoly{\EuScript{S}^\fforall}{n_1}{n_2}}{T}{\bar{\pi}}{\pr}{\oracons}$.
    Inversion of T-NewPoly yields 
    \inlineequation[eq:type-pr-conf-newpoly2]{\wfpolyctx{\envpoly}{\envpr}{\envtvars}}
    and 
    \inlineequation[eq:type-pr-conf-newpoly3]{\nuLowestpriorityEnv{\envtvars}{} = \prtop}
    and 
    $\priority{\bar{\pi}} = [\priority{n_1} + k \cdot \priority{n_2} \mid k\in\mathbb{N}]$.

    On the other hand, for
    $\runtimetyping{\envtvars}{\confnu{x}{y}{\evalapp{\threadctx}{\epair{x}{y}}}{\bar{\pi}}}{\thread}{\Psi}$:
    inverting C-New yields 
    \[\runtimetyping{\envtvars, x:\EuScript{S}^\fforall, y:\dualof{\EuScript{S}^\fforall}}{\evalapp{\threadctx}{\epair{x}{y}}}{\thread}{\Psi'},\]
    where $\Psi' = \Psi,x: \EuScript{S}^\fforall \mapsto \priority{\bar{\pi}}, y: \dualof{\EuScript{S}^\fforall} \mapsto \priority{\bar{\pi}}$.
    By inversion of C-Child/C-Main, we get 
    \inlineequation[eq:type-pr-conf-newpoly4]{\typingpoly{\emptyenv}{\envtvars, x:\EuScript{S}^\fforall, y:\dualof{\EuScript{S}^\fforall}}{\emptyenv}{\Psi'}{\epair{x}{y}}{T}{\bar{\pi_0}}{\pr}{\oracons}}.
    We apply inversion of T-Pair and T-Var (twice, for $x$ and $y$), and conclude using~\eqref{eq:type-pr-conf-newpoly2} and~\eqref{eq:type-pr-conf-newpoly3}.\qed
\end{proof}

\paragraph{Progress and Deadlock Freedom}

With the concepts introduced in~\cref{sec:safety} to describe the state of
communication in a given expression, we can prove progress for the term language (\cref{thm:progress}).

\begin{proof}[\cref{thm:progress}]
  By rule induction on the hypothesis $\typingpoly{\envpoly}{\envtvars}{\envpr}{\envmap}{\e}{T}{\bar{\pi}}{\pr}{\oracons}$.
    
    \begin{enumerate}
    	\item     \emph{Cases T-Var, T-Const, T-TAbs, T-PAbs, T-AbsLin, T-AbsUn, T-Sel.}
    Variables, constants, term, type and priority abstractions, as well as $\eselect{k}$
    are all values.

   \item  \emph{Case T-App.}
    In this case, $\e = \eapp{e_1}{e_2}$, $\Gamma = \Gamma_1\circ \Gamma_2$ and $\pr = \prglb{\pr_1}{\pr_2}$
    and $\priority{\bar{\pi}} = \priority{\emptyenv}$. 
    Premises to the rule include
    $\typingpoly{\envpoly}{\envtvars_1}{\envpr}{\envmap}{e_1}{\tarrow{T_1}{T_2}{\pr}{\prsigma}{m}}{\emptyenv}{\pr_1}{\oracons}$
    and 
    $\typingpoly{\envpoly}{\envtvars_2}{\envpr}{\envmap}{e_2}{T_1}{\emptyenv}{\pr_2}{\oracons}$.
    Since $\envtvars$ only has session types, $\envtvars_1$ and $\envtvars_2$ also only have sessions.
    By induction hypothesis on $e_1$, we can have three cases:\\
    \textbf{Case 1: $e_1$ is a value $v_1$.} Here we proceed by induction on $e_2$ and 
      we get three subcases:
      \begin{itemize}
        \item $e_2$ is a value $v_2$. Since $v_1$ has type $\tarrow{T_1}{T_2}{\pr}{\prsigma}{m}$, then one of the following holds: \\
        $v_1 = \eabs{x}{U}{e}{m}$, in which case $e$ reduces by \textsf{E-App}, or\\
        $v_1 = \eselect{k}$, which yields $\typingpoly{\envpoly}{\envtvars_2}{\envpr}{\envmap}{v_2}{\tintchoice{\ell}{S_\ell}{L}{\pr'}}{\emptyenv}{\pr_2}{\oracons}$
          and so, by inversion on the typing rules, $v_2 = x$ and $e$ is ready, or\\
        $v_1 = \etapp{\eprapp{\etapp{\eprapp{\esend}{\pr}}{T}}{\prsigma}}{S}$, in which case 
          $\eapp{v_1}{v_2}$ is a value, or \\
        $v_1 = \eapp{\etapp{\eprapp{\etapp{\eprapp{\esend}{\pr}}{T}}{\prsigma}}{S}}{v}$: similarly to  
          $\eselect{k}$, by inversion we should have $v_2 = x$ and so $e$ is ready, or\\
        $v_1 = \etapp{\eprapp{\etapp{\eprapp{\ereceive}{\pr}}{T}}{\prsigma}}{S}, 
          \eprapp{\eclose}{\pr'}, \eprapp{\ewait}{\pr'}$, for which the analysis is similar, or\\
        $v_1 = \eprapp{\eprapp{\efork}{\pr'}}{\pr''}$, in which case $\eapp{v_1}{v_2}$ is ready.
        \item $e_2$ reduces. Take $E = \eapp{v_1}{\evalhole{}}$ and $e=E[e_2]$ reduces by \textsf{E-Ctx}.
        \item $e_2$ is ready and has the form $\evalapp{E}{e_2'}$. Considering $E'=\eapp{v_1}{E}$, 
        we get $e=E'[e_2']$ and so is ready.
      \end{itemize}
      \textbf{Case 2: $e_1$ reduces.} Considering $E = \eapp{\evalhole{}}{e_2}$, we get $e = E[e_1]$, which reduces by \textsf{E-Ctx}.\\
      \textbf{Case 3: $e_1$ is ready and has the form $\evalapp{E}{e_1'}$.} Consider $E' = \eapp{E}{e_2}$ and notice that $e = E'[e_1']$; 
      this means that $e$ is ready.

    \item  \emph{Cases T-TApp, T-PApp} are similar.

    \item \emph{Cases T-Inst}. In this case, $e=\einst{x}{\envmap}$ so it reduces by \textsf{E-PInst}.

    \item \emph{Case T-Pair}. In this case, $e = \epair{e_1}{e_2}$. By induction hypothesis, either both $e_1$ and $e_2$ are values, in which case $e$ is a value, 
    or one of them is not a value, say $e_1$. If $e_1$ reduces, consider the 
    context $E = \epair{\evalhole{}}{e_2}$ and conclude that $e$ reduces by \textsf{E-Ctx}. 
    If $e_1$ is ready and has the form $\evalapp{E}{e_1'}$, 
    consider the context $\epair{\evalapp{E}{e_1'}}{e_2}$ and conclude that $e$ is also ready.
    \item \emph{Case T-Let}. In this case $e = \elet{x}{e_1}{e_2}$. By induction hypothesis on $e$, one of the 
    following holds: $e_1$ is a value, in which case $e$ reduces by \textsf{E-LetElim};
    or $e_1$ reduces, in which case we consider $E = \elet{x}{\evalhole{}}{e_2}$ to conclude
    that $e$ reduces; or $e_1$ is ready and has the form $\evalapp{E}{e'}$ and we consider
    $E' = \elet{x}{\evalapp{E}{e'}}{e_2}$ and conclude that $e$ is ready.
    \item \emph{Case T-LetPair}. In this case $e = \eletpair{x_1}{x_2}{e_1}{e_2}$. By induction 
    hypothesis on $e_1$, one of the following holds: $e_1$ is a value, in which case $e_1 = \epair{v_1}{v_2}$
    and so $\e$ reduces by \textsf{E-PairElim}; or $e_1$ reduces, in which case we 
    consider $E = \eletpair{x_1}{x_2}{\evalhole{}}{e_2}$ to conclude that $e$ reduces;
    or $e_1$ is ready and has the form $\evalapp{E}{e_1'}$ we consider $E' = \eletpair{x_1}{x_2}{\evalapp{E}{e_1'}}{e_2}$ 
    and conclude that $e$ is ready.

    \item  \emph{Case T-Match} is similar.

   \item  \emph{Cases T-New, T-NewPoly}. In these cases, $e=\enew{S}$ or $e = \enewpoly{\EuScript{S}^\fforall}{n_1}{n_2}$. So, $e$ is ready by definition.

   \item  \emph{Case T-Eq} follows by induction hypothesis.\qed
        \end{enumerate}

\end{proof}

To ultimately prove that well-typed closed configurations are 
deadlock free (\cref{{thm:config-deadlock-free}}), we need an 
auxiliary lemma.

\begin{lemma}
  \label{lem:congruent-to-canonical-form}
  If $\runtimetyping{\envtvars}{\config{C}}{\mainthread}{\Psi}$ then there is some
  $\config{D}$ such that $\config{C}\structcong\config{D}$ and
  $\config{D}$ is in canonical form.
\end{lemma}
\begin{proof}
  It follows from structural congruence, 
  observing that the only values that child threads might have 
  are of the form $\eunit$.\qed
\end{proof}

We can now prove the results on deadlock freedom:

\begin{proof}[\cref{thm:config-deadlock-free}]
  \begin{enumerate}
    \item 
  Let $\config{C}$ be a configuration in canonical form. Assume that $\config{C}$ is of the form:
  \[\config{C}=\confnu{x_1}{y_1}{\ldots  \confnu{x_n}{y_n}{(\confpar{\confthread{\childthread}{e_1}}{\confpar{\ldots}{\confpar{\confthread{\childthread}{\e_m}}{\confthread{\mainthread}{\e}}}})}{\bar{\pi}_n}}{\bar{\pi}_1}.\]
  By \Cref{thm:progress}, we know that for each $e_i$, either:
  \begin{itemize}
    \item $\red{e_i}{e_i'}$, in which case we apply \textsf{R-LiftM}, \textsf{R-LiftC}
    and find $\config{D}$ such that $\red{\config{C}}{\config{D}}$, or
    \item $e_i$ is ready (provided that it cannot be a value by definition
    of canonical form).
  \end{itemize}
  Likewise, for $\e$, we either have a reduction $\red{\e}{\e'}$
  that is lifted to the configuration or, either, $\e$ is ready or it is a value.
  Assume \emph{all expressions} $\e_i$ are ready and pick the ready expression $\e' \in\{e_1, \ldots, \e_m, \e\}$ with 
  the lowest	 priority bound.\\
  \textbf{Case: $\e' = \evalapp{\evalctx}{\enew{S}}$.} We apply \textsf{R-New} and 
    then \textsf{R-LiftC}.\\
  \textbf{Case: $\e' = \evalapp{\evalctx}{\enewpoly{\EuScript{S}^\fforall}{n_1}{n_2}}$.} We apply \textsf{R-NewPoly} 
  and 
    then \textsf{R-LiftC}.\\
    \textbf{Case: $\e' = \evalapp{\evalctx}{\eapp{\eprapp{\eprapp{\efork}{\pr}}{\prsigma}}{\e_3}}$.} We apply \textsf{R-Fork} and then
    \textsf{R-LiftC}.\\
    \textbf{Case: $\e'$ is an action.} Suppose the action is on 
    endpoint $x: S$. Let $y:\dualof{S}$ be the dual endpoint of $x$. 
    Since dual endpoints have the same priority sequence, 
    the action on which $e'$ is stuck has the same priority as its dual counterpart, so 
    the dual action cannot occur in $e'$.
    Since channels are linear, there must be an expression $\e''\in \{\e_1,\ldots, \e_m, \e\}\setminus\{\e'\}$
    that uses $y$. The expression $\e''$ must either be ready or a value. We analyze these two sub-cases:
    
      \textbf{Sub-case: $\e''$ is ready.} 
      Expression $\e''$ must be ready to act on $y$, otherwise there would
      be another action on $y$ with priority  lower than that of ${\e''}$,
      which contradicts our choice that $e'$ has lowest priority.
      Hence, $e'$ and $e''$ reduce using \textsf{E-Com}, \textsf{E-Ch} 
      or \textsf{R-Close}. Then, we apply \textsf{R-LiftC}.
      
      \textbf{Sub-case: $\e''$ is a value.} Since $\config{C}$ is in canonical form, 
      $\e''$ can only be a value if $\e'' = \e$. Since $\e''$ is a value, it 
      is either $y$, $\epair{y}{v}$, $\epair{v}{y}$, or $\eabs{\_}{\_}{\e'''}{m}$, where $\e'''$
      contains $y$. In all these cases typing the expression in the main thread 
      would not lead to priority $\prtop$ (as is imposed by \textsf{C-Main}), 
      which leads to a contradiction.
      
  Suppose now that 
 $\config{C} =\confnu{x_1}{y_1}{\ldots  \confnu{x_n}{y_n}{\confthread{\mainthread}{\e}}{\bar{\pi}_n}}{\bar{\pi}_1}$.
  By \Cref{thm:progress} we either have a reduction $e\rightarrow e'$, or $e$ is ready, or $e$ is a value. By the reasoning above, if $e$ is ready then it cannot be an action, otherwise we would need to have an endpoint in scope with the same (smallest) priority within the same expression, which is not possible; we are then left with \textsf{R-New} and \textsf{R-Fork} which were analyzed above. If $e$ is a value, by the reasoning above it cannot include endpoints and so $v$ is an unrestricted value.\qed

\item 
  By \Cref{lem:congruent-to-canonical-form}, there exists $\config{D}_1$
  in canonical form 
  such that $\config{C}\structcong\config{D}_1$. 
  We have that either $\config{D}_1$ reduces forever, in which case $\config{C}$ is deadlock free, or
  we apply \Cref{thm:config-deadlock-free} (1), followed
  by \Cref{lem:subject-red-processes,lem:congruent-to-canonical-form}
  to get $\config{D}_1\rightarrow^*{\config{D}_n}$, where 
  $\runtimetyping{\emptyenv}{\config{D}_n}{\mainthread}{\emptyenv}$  and $\config{D}_n \not \rightarrow$.
  To show that $\config{D}_n \structcong \confthread{\mainthread}{v}$, we 
  first prove that $\config{D}_n$ does not contain any session type in scope. 
  
  For that purpose, let us assume that 
  $\config{D}_n = \confnu{x}{y}{\config{D}_{n+1}}{\bar{\pi}}$.
  By \textsf{C-New} we get $\runtimetyping{x:S, y:\dualof{S}}{\config{D}_{n+1}}{\mainthread}{\envmap}$, 
  where $\envmap = \{x:S\mapsto{\bar{\pi}}, y:\dualof{S}\mapsto{\bar{\pi}}\}$.
  That is, $x$ and $y$ occur in $\config{D}_{n+1}$ and have
  compatible types. Thus, one of the following reductions should hold on
  $\config{D}_n$:
  \textsf{R-Com}, \textsf{R-Ch} or \textsf{R-Close}. This contradicts our
  hypothesis that $\config{D}_n \not \rightarrow$, so this case does not hold.
  
  This means that $\config{D}_n$ does not have any restrictions. Since it is 
  well-typed under the $\emptyenv$ context, $\config{D}_n$
  should be closed. Then, it must be structurally congruent to 
  $\confthread{\mainthread}{v}$, where $v$ is an unrestricted value. 
  We conclude that $\config{C}$ is deadlock free.\qed
\end{enumerate}
\end{proof}

\subsection{A Cyclic Network: Milner's Scheduler}
\label{ap:cyclic-scheduler}

To illustrate the applicability of our type system and language to cyclic network topologies, 
let us now consider Milner's cyclic scheduler~\cite{DBLP:books/daglib/0067019}. 
This example is based on Van den Heuvel and P\'{e}rez's version of Milner's scheduler
under asynchronous session communication~\cite{DBLP:journals/lmcs/0001024}.

Milner's scheduler consists of a set of $k$ partial schedulers distributed in 
a ring topology and assisted by $k$ workers. Tasks are performed in rounds, 
initiated by a \emph{leader} and continued by the \emph{followers}. 
Now, we describe an implementation of the system. 
The communication of each scheduler with its worker is governed by 
type \lstinline|Worker|, which consists of selecting \lstinline|Start| at
priority \lstinline|~a~| and 
then awaiting an acknowledgement \lstinline|()| at priority \lstinline|~a+2~| 
(according to the implementation below):
\begin{lstlisting}
type Worker = ~forallp a belongsTo (bot,top) =>~ o+~a~{Start: ?~(a+2)~(); Worker}
\end{lstlisting}

The communication of a scheduler with its succeeding neighbour consists of 
selecting \lstinline|Start| and then selecting \lstinline|Next| after (up to) four operations:
\lstset{firstnumber=2}
\begin{lstlisting}
type Sched = ~forallp b belongsTo (bot,top) =>~ o+~b~{Start: o+~(b+4)~{Next: Sched}}
\end{lstlisting}

In a given round, each \lstinline|follower| waits for a \lstinline|Start| signal from its 
\lstinline|prev|ious neighbour and then notifies its corresponding \lstinline|worker| 
and the \lstinline|succ|eeding neighbour to \lstinline|Start| their tasks. After 
having \lstinline|receive|d an acknowledgement 
from its \lstinline|worker| and a \lstinline|Next| round notification from its \lstinline|prev|ious 
neighbour, this \lstinline|follower| sends a \lstinline|Next| round notification 
to its \lstinline|succ|eeding neighbour and then (recursively) waits for a new round.\footnote{To simplify notation, we omit type and priority applications when using 
the primitives \lstinline|send|, \lstinline|receive|, and \lstinline|fork|---the instantiations are easily 
identifiable from the context.} 

\lstset{firstnumber=3}
\begin{lstlisting}
follower : ~forallp p belongsTo (bot,top) =>~
           dualof Sched ->~[top,bot]~ Worker 1->~[p,bot]~ Sched 1->~[p,top]~ ()
follower prev worker succ = 
    match (inst prev) with {                                    
        Start prev -> 
            let worker = select Start (inst worker) in          
            let succ = select Start (inst succ) in              
            let (_, worker) = receive worker in                 
            match prev with {                                   
                Next prev ->                                    
                    let succ = select Next succ in              
                    follower~{next prev}~ prev worker succ
            }
    }
\end{lstlisting}

The first round is 
initiated by a \lstinline|leader|, who notifies its \lstinline|worker|
and \lstinline|succ|eeding neighbour
to \lstinline|Start| their tasks. After having \lstinline|receive|d an acknowledgment from 
its \lstinline|worker|, the \lstinline|leader|
sends a \lstinline|Next| round signal to its \lstinline|succ|eeding neighbour 
and behaves like a \lstinline|follower| for the next rounds.

\lstset{firstnumber=17}
\begin{lstlisting}
leader : ~forallp p belongsTo (bot,top) =>~ 
         Worker ->~[top,bot]~ dualof Sched 1->~[p,bot]~ Sched 1->~[p,bot]~ 
         ()
leader worker prev succ = 
    let worker = select Start (inst worker) in      
    let succ = select Start (inst succ) in              
    let (_, worker) = receive worker in                
    let succ = select Next succ in            
    follower~{next prev}~ prev worker succ                 
\end{lstlisting}

The \lstinline|worker| is notified to \lstinline|Start| and \lstinline|send|s an ack:

\lstset{firstnumber=26}
\begin{lstlisting}
worker : dualof Worker ->~[top,top]~ ()
worker x = 
    match (inst x) with { 
        Start x ->        
            let x = send () x in 
            worker x
    }
\end{lstlisting}

The \lstinline|main| function is responsible for creating communication channels, 
coordinating priorities 
and ensuring that threads are launched in accordance with 
the cyclic network topology. For a scheduler with $k = 3$, we define the following implementation:

\lstset{firstnumber=33}
\begin{lstlisting}
main : ()
main = 
    let (a1, b1) = new Worker ~1 6~ in
    let (a2, b2) = new Worker ~3 6~ in 
    let (a3, b3) = new Worker ~5 6~ in 
    let (c1, d1) = new Sched ~2 6~ in 
    let (c2, d2) = new Sched ~4 6~ in 
    let (c3, d3) = new Sched ~6 6~ in 

    fork (\_:()1-> leader~{next a1}~ a1 d3 c1);   --A1
    fork (\_:()1-> follower~{next d1}~ d1 a2 c2); --A2
    fork (\_:()1-> follower~{next d2}~ d2 a3 c3); --A3

    fork (\_:()1-> worker b1);                   --worker1
    fork (\_:()1-> worker b2);                   --worker2
    worker b3                                    --worker3
\end{lstlisting}

The rounds start with the leader. Looking at the implementation of 
the \lstinline|leader| we see that 
the first communication is done with the \lstinline|worker|, so 
the communication channel established with the first \lstinline|Worker|
starts at priority \lstinline|~1~| (Line 35). Looking to the \lstinline|leader| again,
we now see that the channel created in Line 42 should start at 
priority \lstinline|~2~| (*). Proceeding on function \lstinline|leader|, 
Line 23 occurs at priority \lstinline|~3~| (according to type \lstinline|Worker|). 
Then, looking at type \lstinline|Sched|, Line 24 is done at priority \lstinline|~6~| 
(which is okay). Now it is time to look at the succeding \lstinline|follower|
and ensure that priorities of dual operations match. Since this \lstinline|follower|
starts at priority \lstinline|~2~| (as a consequence of observation (*) above), the second 
\lstinline|Worker| starts at priority \lstinline|~3~| (Line 36) and the communication to the 
\lstinline|succ|eeding neighbour starts at priority \lstinline|~4~| (Line 39)---which is the information  
that we need to establish the priority of the third \lstinline|Worker| in Line 37 and the priority of 
the channel that connects the two followers in Line 40. Going through a second round of
analysis, we observe that communication with the first \lstinline|Worker| starts 
again at priority \lstinline|~7~|. {This shows that the global increment 
is }\lstinline|~6~|. 
We see how this increment is directly related to the number of created sessions, to the sequential behavior in their types \lstinline|Worker| and \lstinline|Sched|, and to the session interleaving in \lstinline|main|.
Notice that any other multiple 
of }\lstinline|~6~| (such as \lstinline|~12~|) {would fail to preserve the correct ordering 
of communications.}




\paragraph{Deadlock Freedom.}
The program above is well-typed and therefore deadlock free.
The typing derivation starts by applying \textsf{T-LetPair} and \textsf{T-NewPoly}
six times to put \lstinline|a1|--\lstinline|a3|, \lstinline|b1|--\lstinline|b3|, 
\lstinline|c1|--\lstinline|c3|, \lstinline|d1|--\lstinline|d3| in the context
and associate each endpoint to the corresponding priority sequence. 
Then, it isolates the
first \lstinline|fork| using \textsf{T-Let}. 
To check the typing of the first \lstinline|fork|, we apply \textsf{T-App} to split
the typing of constant \lstinline|fork| from the typing of the 
thunk \lstinline|\_:()1-> leader~{next a1}~ a1 d3 c1|. 
The typing of  \lstinline|fork| succeeds after applying \textsf{T-PApp}
twice to validate the ommited annotations \lstinline|~{top}{top}~|, i.e., checking that 
the first priority instantiation \lstinline|~top~|\,$\in\intervaloc{\prbot}{\prtop}$ 
and the second \lstinline|~top~|\,$\in\intervalcc{\prbot}{\prtop}$, as prescribed in the 
typing of \lstinline|fork| (\cref{fig:types-constants}). 
For   \lstinline|\_:() 1-> leader~{next a1}~ a1 d3 c1|, we 
apply \textsf{T-AbsLin}, followed by \textsf{T-App} three times, 
to account for each argument of \lstinline|leader~{next a1}~|. 
We apply \textsf{T-Var} to the arguments \lstinline|a1|, \lstinline|d3| and \lstinline|c1|.
Then, use \textsf{T-PApp} to check that 
\lstinline|leader~{next a1}~| conforms to the provided type.
For the other \lstinline|fork|s the reasoning is similar.

To check the rest of the program, we 
apply \textsf{T-App}, apply \textsf{T-Var} to the argument \lstinline|b3| and then 
evaluate function \lstinline|worker| under the updated priority sequences.

\section{Additional Related Work}
\label{sec:related-work}

Closely related work has been discussed throughout the paper; here we provide a broader context by commenting on other works. 
\paragraph{Session Types.}
Session types were proposed by Honda et al.~\cite{DBLP:conf/concur/Honda93,DBLP:conf/esop/HondaVK98} 
as a way of structuring communication in heterogeneously typed channels. The 
original theory of session types has been extended in many ways, including 
incorporating bounded polymorphism~\cite{DARDHA2017253,Gay2008895}, label 
dependency~\cite{10.1145/3371135} and types with non-regular recursion via 
con\-text-free session types~\cite{DBLP:conf/icfp/ThiemannV16}. 

Our work 
capitalizes on developments for CFSTs, 
including extensions with higher-order types~\cite{Costa2022}, impredicative polymorphism~\cite{DBLP:journals/iandc/AlmeidaMTV22},  
higher-order polymorphism~\cite{10.1007/978-3-031-30044-8_15}, algorithms for type equivalence~\cite{DBLP:conf/tacas/AlmeidaMV20}, subtyping~\cite{DBLP:conf/concur/SilvaMV23} and kind inference~\cite{Almeida_2023,jlamp}. These 
developments have been integrated in the compiler of the FreeST programming language~\cite{freest}. 
This line of work has focused on increasing the expressiveness of types,  adopting the view that programmers should have the liberty of writing typable programs with deadlocked behaviors.
Our work shows that aiming at expressivity does not really conflict with strong, type-based correctness guarantees, as CFSTs can be effectively extended with deadlock freedom guarantees while still supporting programs not expressible in standard session types. 

\paragraph{Type Systems for Deadlock Freedom.}
The priority-based ap\-proach that we have incorporated on context-free session types can be traced back to seminal  work by Kobayashi~\cite{DBLP:conf/unu/Kobayashi02,DBLP:conf/concur/Kobayashi06}, who developed several type systems based on priorities for statically enforcing (dead)lock freedom for $\pi$-calculus processes, under linear type systems with \emph{usages}. The resulting classes of concurrent processes are deadlock free by typing, and rather expressive: they can express process networks in cyclic topologies and encode typed $\lambda$-calculi and the sequencing behavior distinctive of (standard) session types~\cite{DARDHA2017253,DBLP:conf/unu/Kobayashi02}. 

Caires and Pfenning~\cite{DBLP:conf/concur/CairesP10} established 
a Curry-Howard correspondence between session types and intuitionistic linear logic, from which correctness guarantees such as 
deadlock freedom and confluence follow from proof-theoretical principles. 
Wadler~\cite{DBLP:conf/icfp/Wadler12} transferred this connection  to the classical setting. 
This Curry-Howard approach to session types enabled a number of extensions, including forms of 
typed observational equivalence~\cite{DBLP:journals/iandc/PerezCPT14} 
and dependent types~\cite{10.1145/2003476.2003499} as
well as integrations of processes, functions, and sessions~\cite{DBLP:conf/esop/ToninhoCP13}, to name just a few.
Also in this line we find nested session types~\cite{10.1007/978-3-030-72019-3_7}, which were proved to strictly extend context-free session types.
An outstanding limitation of the logic-based approach to session types is its support for  tree-like topologies only, which leaves out many real-world scenarios with cyclic topologies~\cite{DBLP:journals/jlap/DardhaP22}. 
Extensions of the logic-based approach with priorities, in order to support cyclic topologies, have been proposed in 
\cite{DBLP:conf/fossacs/DardhaG18} (for synchronous processes without recursive types) and in \cite{DBLP:journals/corr/abs-2110-00146,DBLP:journals/lmcs/0001024} (for asynchronous processes with recursive types). 

\paragraph{From Processes to Programs.}
Strictly speaking, all previously mentioned works concern processes, not programs. 
Although processes are long known to encode concurrent programs, the difference is relevant, as the non-local analyses enabled by priorities can be more directly performed for processes (where communication constructs are explicit) than for programs (where communication arises via functional constructs, cf. \cref{fig:types-constants}). 
To our knowledge, Padovani and Novara were the first to adapt a type system for deadlock freedom  for processes~\cite{padovani_linear_pi} to the case of higher-order concurrent programs~\cite{DBLP:conf/forte/PadovaniN15} (see also their   technical report~\cite{padovani:hal-00954364}); a source of inspiration for our developments, their work uses polymorphism and a form of \emph{regular} recursion that appears too limited to specify examples such as the one discussed in \cref{ss:trees}.
Our proof of deadlock freedom has been also influenced by  Kokke and Dardha's work~\cite{DBLP:journals/lmcs/KokkeD23}, who adapted Padovani and Novara's approach to the case of regular session types without recursion, which is subsumed by our work.



\end{document}